\documentclass[ecta,nameyear,draft]{econsocart}

\usepackage{threeparttable}

\usepackage{amsmath,blkarray,bm,dsfont}
\usepackage{amssymb}
\usepackage{tikz}
\usepackage{subfigure}
\usepackage[mathscr]{euscript}
\usepackage{caption}
\usepackage{xcolor}
\usepackage{colortbl} 
\usepackage{arydshln} 
\usepackage{float}
\usepackage{multirow}
\usetikzlibrary{positioning}
\usepackage{enumitem}

\usetikzlibrary{patterns,intersections}

\newcommand{\diag}{\operatorname{diag}}
\newcommand{\ess}{\operatorname{ess}}
\newcommand{\conv}{\operatorname{conv}}
\newcommand{\cl}{\operatorname{cl}}

\RequirePackage[colorlinks,citecolor=blue,linkcolor=blue,urlcolor=blue,pagebackref]{hyperref}

\newcommand{\E}{\mathcal{E}} 
\renewcommand{\ss}{\mathcal{S}} 
\newcommand{\sF}{\mathcal{F}} 
\newcommand{\sPF}{\mathcal{PF}} 
\newcommand{\sY}{\mathcal{Y}} 
\newcommand{\sX}{\mathcal{X}} 
\newcommand{\bS}{\mathbb{S}} 
\newcommand{\bQ}{\mathbb{Q}} 
\newcommand{\bB}{\mathbb{B}} 
\newcommand{\sC}{\mathcal{C}^*} 
\newcommand{\cC}{\mathcal{C}} 
\newcommand{\linseg}{\Lambda} 

\renewcommand{\v}{v} 
\newcommand{\q}{q} 

\newcommand{\er}{e_r} 
\newcommand{\eb}{e_b} 
\newcommand{\R}{R} 
\newcommand{\B}{B} 
\newcommand{\G}{G} 
\newcommand{\F}{F} 

\newcommand{\X}{X} 
\newcommand{\Y}{Y} 
\newcommand{\LL}{L} 
\newcommand{\bLL}{\mathbf{\LL}} 
\newcommand{\cM}{\mathcal{M}} 
\newcommand{\hcM}{\widehat{\mathcal{M}}} 
\newcommand{\tcM}{\widetilde{\mathcal{M}}} 
\newcommand{\bI}{\boldsymbol{I}} 

\newcommand{\h}{h} 
\newcommand{\hn}{\widehat{\h}} 
\newcommand{\btheta}{\boldsymbol{\theta}} 
\newcommand{\bvartheta}{\boldsymbol{\vartheta}} 
\newcommand{\bzeta}{\boldsymbol{\zeta}} 
\newcommand{\bxi}{\boldsymbol{\xi}} 

\newcommand{\HH}{\mathsf{H}} 
\newcommand{\bP}{\mathbb{P}} 

\startlocaldefs
\theoremstyle{remark}
\newtheorem{theorem}{Theorem}[section]

\newtheorem{defn}{Definition}

\newtheorem{prop}{Proposition}[section]

\newtheorem{asm}{Assumption}
\newtheorem{remark}{Remark}[section]
\newtheorem{proc}{Procedure}

\defcitealias{lia:lu:mu:oku24}{LLMO} 

\defcitealias{obe:pow:vog:mul19}{OPVM} 
\endlocaldefs

\begin{document}
\begin{frontmatter}

\title{Inference for an Algorithmic Fairness-Accuracy Frontier}

\begin{aug}
\author[id=au1,addressref={add1}]{\fnms{Yiqi}~\snm{Liu}\ead[label=e1]{yl3467@cornell.edu}}
\author[id=au2,addressref={add2}]{\fnms{Francesca}~\snm{Molinari}\ead[label=e2]{fm72@cornell.edu}}
\address[id=add1]{%
\orgdiv{Department of Economics},
\orgname{Cornell University}}

\address[id=add2]{%
\orgdiv{Department of Economics},
\orgname{Cornell University}}
\end{aug}

\vspace{-2em}

\support{\emph{This draft: June 13, 2025}\\
We thank 
Levon Barseghyan, Gillian Hadfield, Hiroaki Kaido, Nathan Kallus, Jens Ludwig, Chuck Manski, Alice Qi, Chen Qiu, Andres Santos, Vira Semenova, Rahul Singh, Alex Tetenov, Lars Vilhuber, Davide Viviano, reviewers for the EC24 conference, seminar participants at Chicago, Cornell, Geneva, JHU, LSE, MSU, Munich, SciencesPo, Stanford, Toulouse, UCL, Warwick, EC24, ESIF: Economics and AI+ML Meeting, the 2024 Brown University workshop ``Using Data to Make Decisions,'' and especially Jos{\'e} Montiel-Olea and Thomas Russell for helpful comments. All data and replication files can be accessed at \href{https://github.com/yiqi-liu/TestAlgFair}{github.com/yiqi-liu/TestAlgFair}.}
\begin{abstract}
Algorithms are increasingly used to aid with high-stakes decision making. Yet, their predictive ability frequently exhibits systematic variation across population subgroups. 
To assess the trade-off between fairness and accuracy using finite data, we propose a debiased machine learning estimator for the fairness-accuracy frontier introduced by  \citet*{lia:lu:mu:oku24}. We derive its asymptotic distribution and propose inference methods to test key hypotheses in the fairness literature, such as (i) whether excluding group identity from use in training the algorithm is optimal and (ii) whether there are less discriminatory alternatives to a given algorithm. 
In addition, we construct an estimator for the distance between a given algorithm and the fairest point on the frontier, and characterize its asymptotic distribution. Using Monte Carlo simulations, we evaluate the finite-sample performance of our inference methods. We apply our framework to re-evaluate algorithms used in hospital care management and show that our approach yields alternative algorithms that lie on the fairness-accuracy frontier, offering improvements along both dimensions.
\end{abstract}


\bigskip

\end{frontmatter}

\newpage
\section{Introduction}
\label{sec:intro} 
Algorithms are increasingly used in many aspects of life, often to guide or support high-stake decisions, for example by predicting job performance, re-offense risk, loan default, college success, or patient health.
These predictions feed, respectively, into the determination of who should be hired; which defendants should receive bail; who should be granted a loan; which students should be admitted to college; and which patients to treat. 
Yet, a growing body of literature documents that algorithms may exhibit bias against legally protected groups, both in their predictive accuracy and in the decisions they lead to \citep[see, e.g.,][]{ang:lar:mat:kir16, arn:dob:hul21,obe:pow:vog:mul19,ber:jei:jab:kea:rot21}.
The bias may arise, for example, due to the choice of labels the algorithm is trained on, the objective function that the algorithm optimizes, the training procedure, and many other factors involved in the design of the algorithm \citep[see, e.g.,][]{cow:tuc20}.

Designing an algorithm often entails a trade-off between making it more \emph{fair}, i.e., less likely to disproportionately harm a protected class, and more \emph{accurate}, e.g., better at assigning treatment to those who benefit from it and withholding it from those who do not. 
As a result, improving fairness often comes at the cost of accuracy.
Regulators, policymakers, algorithm designers, and actors affected by algorithmic predictions all have an interest in assessing various aspects of this trade-off.

We provide a set of tools for estimation of and statistical inference on a \emph{fairness-accuracy} (FA) frontier recently characterized by \citet*[LLMO henceforth]{lia:lu:mu:oku24}, where fairness is measured by the gap between group-specific expected losses.
The theoretical analysis in \citetalias{lia:lu:mu:oku24} assumes perfect knowledge of the population distribution of the observable variables and formalizes the trade-off between accuracy and fairness, shading light on how to use properties of the data distribution to determine whether it is optimal for the designer of the algorithm to exclude certain inputs from use. 
However, in practice, regulators and policymakers typically have access to only finite data. Hence, statistical inference tools are crucial for analyzing properties of algorithms and for their regulation.

We put forward a consistent estimator for \citetalias{lia:lu:mu:oku24}'s FA-frontier and derive its asymptotic distribution. For each point on the FA-frontier, we characterize an algorithm that achieves it.
We then develop a method to test hypotheses such as: Is it optimal to fully exclude group identity from use in an algorithm? Does a particular algorithm lead to group-specific expected losses that are on the FA-frontier? How far from the fairest point on the FA-frontier are the group-specific expected losses associated with a given algorithm?

Answers to the first two questions inform the regulation of algorithms and the determination of whether discrimination occurred. 
The law recognizes two main categories of discrimination: \emph{disparate treatment}, where individuals are deliberately treated differently based on their membership in a protected class; and \emph{disparate impact}, where protected classes are adversely affected disproportionately, no matter the intent \citep{kle:lud:mul:sun18,bla:spi22}.
Often, as part of an effort to avoid disparate treatment, algorithms are designed so that they do not take race, gender, or other sensitive attributes as input.
Even class-blind algorithms, however, may lead to disparate impact. 
Our first test informs a fairness-minded policymaker interested in assessing whether banning group identity has the potential to mitigate disparate impact.\footnote{This question is of interest, e.g., when assessing the recent U.S. Supreme Court decision to rule out the use of race in college admissions; for an overview, see \href{https://highered.collegeboard.org/recruitment-admissions/policies-research/access-diversity/2023-scotus-decision/overview}{College Board}.} 

Our second test evaluates whether a given algorithm lies on the frontier—and thus whether a less discriminatory alternative (LDA) exists. This test is relevant to both plaintiffs (e.g., job applicants) and defendants (e.g., hiring companies) in disparate impact disputes.
For example, if a selection process yields disparate impact, the hiring company may invoke business necessity to justify it. 
The challenger must then show the existence of an LDA, i.e., a fairer algorithm that is just as accurate. If our test rejects the null that the current algorithm is on the frontier, it supports the plaintiff's claim. Conversely, if the test fails to reject, there is no statistical evidence that the hiring company can build a fairer algorithm without sacrificing accuracy, supporting the business necessity defense.
When a given algorithm is not on the frontier, we characterize alternative algorithms that improve upon it in terms of accuracy or fairness (or both).

The third inferential method yields an estimator of the distance from a given algorithm to the fairest point on the frontier and constructs a confidence interval around it. 
This tool may interest any fairness-minded agent (e.g., a college) willing to trade some accuracy for reduced disparate impact (e.g., via affirmative action), as it provides a measure of the trade-off between promoting equity and achieving accuracy.

The key insight underlying our proposed inference methods is that since the feasible set of group-specific expected losses associated with all possible algorithms is convex, it can be fully represented by its \emph{support function}. As the FA-frontier is a portion of the boundary of the feasible set, we characterize and estimate it through this support function. We express the hypotheses listed above as restrictions on the support function, yielding easy to understand test statistics that essentially rely on judicious use of the separating hyperplane theorem. 
Throughout our analysis, the support function serves as a unifying tool for inference on properties of algorithms and of the FA-frontier.

We provide a consistent debiased machine learning (DML) estimator of the support function and establish that it converges to a tight Gaussian process as sample size increases, building on and extending results in \citet{ber:mol08}, \citet{cha:che:mol:sch18} and \cite{sem23}. We show how to allow for infimum-type test statistics that are directionally-differentiable mappings of the support function,  building on results of \citet{fan:san19}. 
Earlier uses of the support function for inference in partially identified models \cite[e.g., ][]{ber:mol08, bon:mag:mau12, kai:san14, kai16, cha:che:mol:sch18,mol20,sem23} did not include tests for hypotheses such as the ones we consider. Expressing these hypotheses in terms of restrictions on the support function is one of our main contributions.

We evaluate the finite-sample performance of our inference toolkit using extensive Monte Carlo simulations. We then demonstrate its empirical value by reanalyzing algorithms for high-risk care assignment at a research hospital studied by \citet*{obe:pow:vog:mul19}. We fail to reject the hypothesis that the hospital’s status-quo algorithm admits an LDA, and document fairness and accuracy gains from several alternative algorithms on the frontier that we characterize.

\textbf{Related Literature.} 
A growing literature in computer science and statistics studies algorithmic fairness; see \citet{cho:rot18}, \citet{bar:har:nar23}, and \citet{Corbett24} for comprehensive overviews and open questions.
Models have been developed to explain algorithmic bias by decomposing disparity sources \citep[e.g.,][]{ram:kle:lud:mul20} or incorporating taste-based discrimination and unobservables in label generation \citep[e.g.,][]{ram:rot20}.
Fairness has been modeled as a constraint or regularizer in the objective function that maximizes predictive accuracy \citep[e.g.,][]{dwork12,ber:hei:jab:jos:kea:mor:jam:nee:roth17} and incorporated in the preferences of a social planner that uses algorithms in their decision-making process \citep{kle:lud:mul:ram18,ram:kle:mul:lud20}.
In optimal policy targeting, fairness has been set as the criterion to be maximized when choosing a policy from the set of welfare-maximizing rules \citep[e.g.,][]{viv:bra23}.
When protected class membership is not observed in the data but proxy variables are available, data combination methods have been proposed to partially identify disparity measures \citep{kal:mao:zhou22}.
Yet, tests of hypotheses for properties of the trade-off between fairness and accuracy of algorithms are scant in the literature.
\citet{aue:lia:tab:oku24} propose a test, which can incorporate exogenous constraints on the algorithm space, using sample splitting for the union null hypothesis that a status quo algorithm can be weakly improved in terms of both fairness and accuracy. Their test is based on finding another algorithm within a user-specified subclass, subject to the constraint that the alternative algorithm is at least as accurate as the status quo algorithm.

In contrast, our test for existence of an LDA is a one-shot test, valid across all algorithms rather than a specific subclass, that reveals if the status quo algorithm is on the FA-frontier and does not require first estimating an alternative algorithm.
Characterizing the entire FA-frontier allows us to provide a comprehensive toolkit for statistical inference that can be useful for regulators to determine what algorithm design restrictions and reporting requirements to impose on entities making decisions using algorithms.

\textbf{Outline.} Section \ref{sec:setup} lays out notation and summarizes the derivation of the FA-frontier in \citetalias{lia:lu:mu:oku24}.
Section \ref{sec:sup_fun} characterizes the support function of interest and uses it to describe the FA-frontier.
Section \ref{sec:estimation} derives the asymptotic properties of our DML estimator for the support function.
Section \ref{sec:set_estimation} uses these results to obtain a consistent estimator and an asymptotically valid confidence set for the FA-frontier, and characterizes algorithms attaining points on it.
Section \ref{sec:test} formulates hypotheses of interest in the fairness literature as restrictions on the support function and proposes asymptotically valid tests.
Section \ref{sec:distance_F} provides an estimator and inference method for the distance between the expected group-specific losses associated with a given algorithm and the fairest point on the frontier.
Section \ref{sec:MC_and_empirical} presents our Monte Carlo simulations and re-evaluation of \citet{obe:pow:vog:mul19}'s study on a research hospital’s use of algorithms for assigning patients to a high-risk care management program.
Section \ref{sec:conclude} concludes.
Our main proofs are in Appendix \ref{appn:A};
Appendix \ref{appn:B} reports auxiliary results and extensions, and Appendix \ref{appn:C} includes supplemental empirical results.

\vspace{-.15cm}
\section{Setup}
\label{sec:setup}
\vspace{-.15cm}
Let a population of individuals be described by an outcome $\Y\in\sY\subset\mathbb{R}$, a binary group identity $\G\in\{r,b\}$ (red or blue), and a vector of covariates $\X\in\sX\subset\mathbb{R}^{d_\X}$, with the population distribution of $(\Y,\G,\X)$ denoted $\bP$.
For example, $\Y$ may denote an individual's number of active chronic illnesses in the subsequent year, $\G$ may denote their race, and $\X$ may include age, gender, biomarkers, comorbidity, costs and medication variables.
The relation between $\G$ and $\X$ is left unspecified{, but $\G$ is not part of $\X$}. 
Throughout, we assume $\G$ is binary, though the results extend to multiple groups.
Each individual receives a binary decision $D\in\{0,1\}$, e.g., whether they are automatically enrolled in a high-risk care management program.
An algorithm $a:\sX \mapsto [0,1]$ assigns a probability distribution to $D$; e.g., the algorithm assigns each patient a health risk score in $[0,1]$,
which for simplicity we take to be the only input to the {enrollment decision, and hence to coincide with the enrollment probability}.
Let $\mathcal A(\sX)$ denote the set of all algorithms that map from the input space $\sX$ to a probability distribution over $D$, and $\ell:\{0,1\}\times\sY\mapsto\mathbb R$ be a function that measures the loss associated with decision $d\in\{0,1\}$ for an individual with outcome $y\in\sY$. 
In the example discussed so far, the algorithm designer observes training data consisting of covariates $\X$ and a binary outcome $\Y$ indicating whether someone has a number of chronic illnesses exceeding a given threshold; the loss function $\ell$ may be the classification loss, $\ell(D,\Y)=\mathds{1}\{D\neq \Y\}$, which returns the value $1$ if the algorithm mistakenly enrolls a healthy person in the high-risk care program or fails to enroll someone who is very sick. 
We assume throughout that the training data is drawn from the same distribution as the population that we eventually apply the algorithm to (i.e., the subpopulation for which labels are observed is representative of the entire population). 

Given an algorithm $a\in\mathcal A(\sX)$, let the population expected loss for group $g\in\{r,b\}$ be
\begin{align}
e_g(a)\equiv\mathbb{E
}\left[a(\X)\ell(1,\Y) + (1-a(\X))\ell(0,\Y)|\G=g\right],\label{eq:define_e}
\end{align}
where the expectation is taken with respect to $\bP$. We refer to the group-specific expected loss in Eq.~\eqref{eq:define_e} as \emph{group risk}. Following \citetalias{lia:lu:mu:oku24}, we define a preference ordering over group risk pairs so that $e=(\er,\eb)$ is preferred to $e'=(\er',\eb')$, denoted $e>_{FA}e'$, if
\begin{align}
    \er\le e'_r, \quad \eb\le e'_b, \quad \text{and}\quad |\er-\eb|\le |e'_r-e'_b|,\label{pref}
\end{align}
with at least one strict inequality.
As shown in  \citetalias{lia:lu:mu:oku24}, all of utilitarian, Rawlsian, egalitarian, and various other preferences are consistent with this ordering.
One can then define the \emph{feasible set} of group risk pairs across algorithms from the class $\mathcal{A}(\sX)$ as 
\begin{align}
\E\big(\bP, \mathcal{A}(\sX)\big)\equiv\left\{\bigl(\er(a),\eb(a)\bigr)\in\mathbb{R}^2:a\in\mathcal{A}(\sX)\right\},\label{eq:feasible_set}
\end{align}
and the \emph{fairness-accuracy (FA) frontier} as
\begin{align}
\sF\big(\bP, \mathcal{A}(\sX)\big)\equiv\left\{e\in\E\big(\bP, \mathcal{A}(\sX)\big):~\nexists~e'\in\E\big(\bP, \mathcal{A}(\sX)\big)~\text{such that}~e'>_{FA}e\right\}.\label{eq:FA_frontier}
\end{align}

For finite $(\sX,\sY)$, \citetalias{lia:lu:mu:oku24} show that $\E\big(\bP, \mathcal{A}(\sX)\big)$ is a closed convex set {(we extend this convexity result to general $(\sX,\sY)$ in the proof of Proposition~\ref{prop:sf})} and $\sF\big(\bP, \mathcal{A}(\sX)\big)$ is a specific portion of its boundary connecting three points:  the feasible point that minimizes the risk for group $r$, denoted $\R$; the feasible point that minimizes the risk for group $b$, denoted $\B$, and the feasible point that minimizes the absolute difference in group risks, denoted $\F$.\footnote{Ties are broken in favor of the other group's risk for $\R$ or $\B$. If there are multiple feasible points that minimize the absolute difference in group risks, $\F$ is chosen to be the one that has the lowest risk for both groups.} 
Adapting \citetalias{lia:lu:mu:oku24} nomenclature, call $\big(\bP, \mathcal{A}(\sX)\big)$ \textit{group-balanced} if $\E\big(\bP, \mathcal{A}(\sX)\big)$ has $\R$ and $\B$ such that either $\R=\B=\F$, or $e_r < e_b$ at $\R$ and $e_r > e_b$ at $\B$; call $\big(\bP, \mathcal{A}(\sX)\big)$ \textit{$r$-skewed} if $e_r<e_b$ at $\R$ and $e_r\leq e_b$ at $\B$, and \textit{$b$-skewed} if $e_r\geq e_b$ at $\R$ and $e_r > e_b$ at $\B$.

To ease notation, we drop the dependence of $\E$ and $\sF$ on $\big(\bP, \mathcal{A}(\sX)\big)$ unless explicitly needed. Figure \ref{fig1} illustrates these sets and key points on them under a smoothness condition stated in Assumption \ref{asm:smooth-X}. \citetalias{lia:lu:mu:oku24} (Theorem 1) show that the shape of $\sF$ depends entirely on whether $\big(\bP, \mathcal{A}(\sX)\big)$ is group-balanced or $g$-skewed. 
If group-balanced, $\sF$ is the curve connecting $\R$ and $\B$, coinciding with the Pareto frontier (panel (a)); if $g$-skewed, $\sF$ connects $\F$ and the feasible point minimizing risk for group $g$ (panels (b) and (c) for the $r$-skewed case; omitted panels for the $b$-skewed case).
\citetalias{lia:lu:mu:oku24} (Proposition 6) further show that excluding group identity as an algorithmic input is uniformly welfare-reducing under strict group balance, where $R$ and $B$ are strictly separated by the 45-degree line.

{
\begin{figure}
\centering
\includegraphics[width=\textwidth]{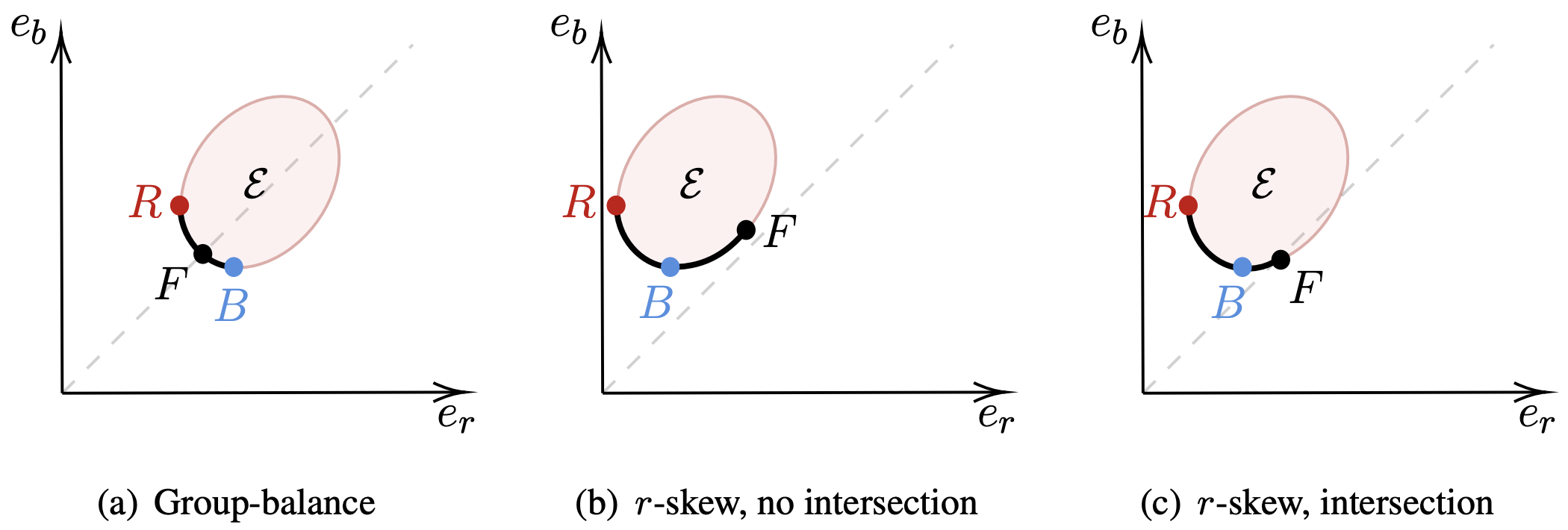}
\caption{\footnotesize{The feasible set $\E$ in pink and the frontier $\sF$ in black under different configurations of $\big(\bP, \mathcal{A}(\sX)\big)$.}}\label{fig1}
\vspace{-.35cm}
\end{figure}
}
\textbf{Notation.} We denote by $\|\cdot\|_E$, $\|\cdot\|_{L^2(\bP)}$, $\|\cdot\|_{\infty}$, respectively, the Euclidean norm, the $L^2$-norm under the probability measure $\bP$, and the $L^\infty$-norm (or sup-norm). For a vector $\mathbf{a}$, let $\|\mathbf{a}\|_{L^2(\bP)}\equiv\bigl\|\|\mathbf{a}\|_E\bigr\|_{L^2(\bP)}$ and $\|\mathbf{a}\|_{\infty}$ be the supremum over the largest component of $\mathbf{a}$. 
For a matrix $\mathbf{A}$, let $\|\mathbf{A}\|_{\max}$ denote its max norm (the maximum absolute value among its entries). 
For two sequences $a_n$ and $b_n$, $a_n\lesssim b_n$ means $a_n\leq c\cdot b_n$ for some constant $c>0$.

\section{Support Function Based Characterizations}\label{sec:sup_fun}
We leverage the convexity of the feasible set $\E$ to characterize it by its support function and express the points $\R$, $\B$, $\F$, and the FA-frontier $\sF$ in Eq.~\eqref{eq:FA_frontier} through this support function.
We begin by observing that Eq.~\eqref{eq:define_e} and the law of iterated expectations yield
\begin{align}
e_g(a) 
&=\mathbb{E}\left[a(\X)\mathbb{E}[\ell(1,\Y)\mathds{1}(\G=g)|\X]+(1-a(\X))\mathbb{E}[\ell(0,\Y)\mathds{1}(\G=g)|\X]\right]/\bP(\G=g)\notag\\
&\equiv\mathbb{E}\left[a(\X)\tfrac{\theta_1^g(\X)}{\mu_g}+(1-a(\X))\tfrac{\theta_0^g(\X)}{\mu_g}\right],\label{eq:expression:e_g}
\end{align}
where we denote $\theta_d^g(\X)\equiv\mathbb{E}[\ell(d,Y)\mathds{1}\{G=g\}|\X]$ the (measurable) conditional expectation of $\LL_d^g\equiv\ell(d,Y)\mathds{1}\{G=g\}$ given $\X$; $\mu_g\equiv \bP(\G=g)$ the population proportion of group $g\in\{r,b\}$; 
and the expectation in Eq.~\eqref{eq:expression:e_g} is taken with respect to the population marginal distribution of the covariates, $\bP(\X)$.
To make sure that Eq.~\eqref{eq:expression:e_g} is well defined, we assume:
\begin{asm}[(Moment Restrictions)]
\label{asm:moments}
\textit{For some constants $0<c_1<1$ and $0<c_2<\infty$, $\mu_g\in(c_1,1-c_1)$ and $\ess\sup_{X\in\sX}\mathbb{E}\left[\left(\LL_d^g\right)^2\big|\X\right]<c_2$, for all $d\in\{0,1\}, g\in\{r,b\}$.
}
\end{asm}
Throughout, we let $\btheta(\X)\equiv[\theta_1^r(\X) ~~ \theta_0^r(\X) ~~ \theta_1^b(\X) ~~ \theta_0^b(\X)]^\intercal$ and
\begin{align}    
\btheta_d(\X)&\equiv[\theta_d^r(\X)~~\theta_d^b(\X)]^\intercal, ~~d\in\{0,1\},\label{eq:def_theta_d}\\
\cM&\equiv\diag(1/\mu_r,1/\mu_b).\label{eq:cM}
\end{align}

\subsection{Support Function of the Feasible Set}
\label{sec:deriv-supp-func}
Given Eqs.~\eqref{eq:expression:e_g}-\eqref{eq:def_theta_d}-\eqref{eq:cM}, $\E$ can be written as
\begin{align}
\E&\equiv \left\{\bigl(\er(a),\eb(a)\bigr)\in\mathbb{R}^2:a\in\mathcal{A}(\sX)\right\}\notag\\
&=\big\{\mathbb{E}[\cM\vartheta(\X)]: \vartheta(\X)\in \conv\left(\{\btheta_0(\X), \btheta_1(\X)\}\right) \big\}=\mathbf{E}\left[\cM\linseg(\X)\right],\label{eq:E_as_Aumann}
\end{align}
with $\conv(\cdot)$ the convex hull of the set in parentheses, $\linseg(\X)\equiv \conv\left(\{\btheta_0(\X),
\btheta_1(\X)\}\right)$ a random interval, $\cM$ in Eq.~\eqref{eq:cM}, and  $\mathbf{E}\left[\cM\linseg(\X)\right]$ the \emph{Aumann expectation} of the scaled random interval $\cM\linseg(\X)$ \citep[Example 1.11 and Def. 3.1]{mol:mol18}.

As the set $\E$ is non-empty, compact, and convex, its \textit{support function} in each direction 
$\q=[\q_1~\q_2]^\intercal\in\bS^1\equiv\{\v\in\mathbb{R}^2:\, \Vert \v\Vert_E=1\}$, defined as
\begin{align*}
    \h_{\E}(\q) \equiv \max_{e\in\E} \q^\intercal e,
\end{align*}
uniquely characterizes $\E$ through the identity \cite[Chapter 13]{roc97}
\begin{align}
    \E=\bigcap_{\q\in\bS^1}\left\{z\in\mathbb{R}^2:\q^\intercal z\leq \h_\E(\q)\right\}.\label{eq:E_through_h}
\end{align}
We next provide a closed-form expression for $\h_{\E}(\q)$.
\begin{prop}
    \label{prop:sf}
    \textit{Let Assumption~\ref{asm:moments} hold. Then:
        \begin{align}
        \h_\E(\q)&=\mathbb{E}\left[\max\{(\cM\q)^\intercal\btheta_0(\X),(\cM\q)^\intercal\btheta_1(\X)\}\right]\notag\\
        &=\mathbb{E}\bigl[(\cM\q)^\intercal\bLL_0+(\cM\q)^\intercal(\bLL_1-\bLL_0)\mathds{1}\{k(\btheta(\X),\cM\q)>0\}\bigr],\label{eqn:sf-new}
    \end{align}
    where $k(\btheta,\cM\q)\equiv(\cM\q)^\intercal\bigl(\btheta_1(\X)-\btheta_0(\X)\bigr)$, $\bLL_d\equiv[\LL_d^r~~\LL_d^b]^\intercal$, and $\LL_d^g\equiv\ell(d,Y)\mathds{1}\{G=g\}$. }
\end{prop}

\noindent The support function $\h_\E(\q)$ is our key inferential tool. One can garner the intuition behind its closed-form expression by rewriting $e_g(a)=\mathbb{E}\left[\tfrac{\theta_0^g(\X)}{\mu_g}\right]+\mathbb{E}\left[a(\X)\tfrac{\theta_1^g(\X)-\theta_0^g(\X)}{\mu_g}\right]$ and
\begin{align*}
    \h_\E(\q)&=\mathbb{E}\left[\q_1\tfrac{\theta_0^r(\X)}{\mu_r}+\q_2\tfrac{\theta_0^b(\X)}{\mu_b}\right]+\max_{a\in\mathcal{A}(\sX)}\mathbb{E}\left[a(\X)k\left(\btheta(\X),\cM\q\right)\right].
\end{align*}
The maximum in the above expression is achieved by the algorithm $a^{\texttt{opt}}(\X;\q)=\mathds{1}\{k(\btheta(\X),\cM\q)>0\}$, yielding Eq.~\eqref{eqn:sf-new} upon applying the law of iterated expectations.
\begin{remark}
    We allow for \emph{randomized decision rules} and for $\mathcal{A}(\sX)$ to be unrestricted.
    If instead the family of algorithms is restricted a priori (e.g., by capacity constraints) so that $a(\X)=\Pr(D=1|\X)\in[\underline{a}(\X),\bar{a}(\X)]$, $0\le\underline{a}(\X)\le\bar{a}(\X)\le 1$, with $\underline{a}(\cdot),\bar{a}(\cdot)$ known functions, our analysis continues to apply by replacing $\{\btheta_0(\X),\btheta_1(\X)\}$ with $\{\bar{a}(\X)\btheta_0(\X)+(1-\bar{a}(\X))\btheta_1(\X),\underline{a}(\X)\btheta_0(\X)+(1-\underline{a}(\X))\btheta_1(\X)\})$.
    In Appendix~\ref{app:sec:threshold_rules}, we also show that our results continue to hold if one restricts attention to \emph{threshold rules} of the form $D = \mathds{1}\{a(\X) \ge 0\}$ for unrestricted $a:\sX\mapsto\mathbb{R}$ (and in fact $a^{\texttt{opt}}(\X;\q)$ is a threshold rule), or to \emph{linear threshold rules} where $a(\X)=[1;~\X]^\intercal\beta$ for some $\beta\in\mathbb{R}^{d_\X+1}$, provided the space of algorithms is sufficiently rich (see Assumption~\ref{asm:rich-A}).
\end{remark}

\subsection{Best Group-Specific Points on the FA-Frontier}
\label{sec:RB}
We next define the \textit{support set} of $\E$ in direction $\q\in\bS^1$:
\begin{align}
    \ss_\E(\q)\equiv\E\cap \left\{z\in\mathbb{R}^2: \q^\intercal z = \h_{\E}(\q)\right\},\label{eq:define_SupportSet}
\end{align}
i.e., $\ss_\E(\q)$ is the intersection between $\E$ and the hyperplane with normal vector $\q$ and constant $\h_{\E}(\q)$, hence collecting the extreme point(s) of $\E$ in direction $\q$. 
To derive a closed-form expression for $\ss_\E(\q)$, we impose the following assumption:
\begin{asm}[(Margin Condition)]
\label{asm:smooth-X}
\textit{There exists $0<m\leq1$ such that for any $\delta>0$,
$\sup_{\q\in \bS^1} \mathbb{P}(|k(\btheta(\X),\cM\q)|<\delta)\lesssim \delta^m$, with the probabilities taken with respect to $\bP(\X)$.}
\end{asm}

Assumption \ref{asm:smooth-X} is a margin condition that guarantees sufficient smoothness in the distribution of $\btheta_1(\X)-\btheta_0(\X)$ for us to show that $\h_\E(\q)$ is differentiable in $\q\in\bS^1$ and consequently $\ss_\E(\q)$ includes a single element in each direction $\q$ \cite[Corollary 1.7.3]{sch93}.
We further show that $\ss_\E(\q)$ equals the gradient of the support function $\h_{\E}(\cdot)$ with respect to $\q$. 
We denote by $\ss_\E(\q)$ both the singleton set and its only element.
\begin{prop}
\label{prop:ss}
\textit{
Let Assumptions~\ref{asm:moments}-\ref{asm:smooth-X} hold.
Then, 
    \begin{align}
        \nabla_\q\h_{\E}(\q)=
        \mathbb{E}\bigl[\cM\bLL_0+\cM(\bLL_1-\bLL_0)\mathds{1}\{k(\btheta(\X),\cM\q)>0\}\bigr]= \ss_\E(\q),\label{eq:ss}
\end{align}
uniformly in $\q\in\bS^1$, where $\ss_\E(\q)$ is uniformly continuous in $\q\in\bS^1$. 
}
\end{prop}
Let $\mathfrak{u}_1\equiv[-1\,\,\,0]^\intercal$ and $\mathfrak{u}_2\equiv[0\,\,\,-1]^\intercal$. The best group-specific points satisfy:
\begin{align}
    \R = \ss_\E(\mathfrak{u}_1)  
    \quad
    \text{and}\quad
    \B =  \ss_\E(\mathfrak{u}_2). 
   \label{eq:coord}
\end{align}
\begin{remark}
    \label{rem1}
    Assumption \ref{asm:smooth-X} allows for discrete covariates (see Proposition \ref{prop:no_kinks_condition}), but is violated if $\X$ includes discrete covariates \textit{only} (see Appendix~\ref{app:sec:threshold_rules}), in which case we can use data jittering to satisfy Assumption \ref{asm:smooth-X} by adding to one discrete covariate a small amount of smoothly distributed noise, thereby garbling that input. The feasible set constructed with a jittered covariate can be made arbitrarily close to the true feasible set \citep[this result can be proved by adapting arguments in][Lemma 8]{cha:che:mol:sch18}.
\end{remark}

\begin{remark}
    \label{remark:kinks}
    The argument in \citet[Supplemental Appendix B.2.3]{bon:mag:mau12} shows that $\E$ has no \emph{kinks} (i.e., no support points such that there exist at least two distinct vectors $\q$ and $\v$ satisfying $\ss_\E(\q)=\ss_\E(\v)$) if and only if for any $\q,\v\in\bS^1,\q\neq\v$, 
    \begin{align}
    \mathbb{P}\big(k(\btheta(\X),\cM\q)>0,k(\btheta(\X),\cM\v)<0\big)>0.\label{eq:no_kinks_cond}
    \end{align}
    If $(\btheta_1(\X)-\btheta_0(\X))$ admits a positive density function on a ball of positive radius that includes zero, Eq.~\eqref{eq:no_kinks_cond} is satisfied. Assumption~\ref{asm:condition_smoothX1X2} in Appendix \ref{appn:B} is an example of low level conditions yielding this result.
    The absence of kinks renders simpler limit distributions for the test statistics that we put forward in Sections \ref{sec:set_estimation}-\ref{sec:distance_F}. Nonetheless, Eq.~\eqref{eq:no_kinks_cond} is \emph{not} needed for our results to apply and we provide a full treatment allowing for the presence of kinks.
\end{remark}

\vspace{-.5cm}
\subsection{Fairest Point on the FA-Frontier}
\label{sec:fairest}

{
\begin{figure}\captionsetup[subfigure]{font=footnotesize}
\centering
\includegraphics[width=0.85\linewidth]{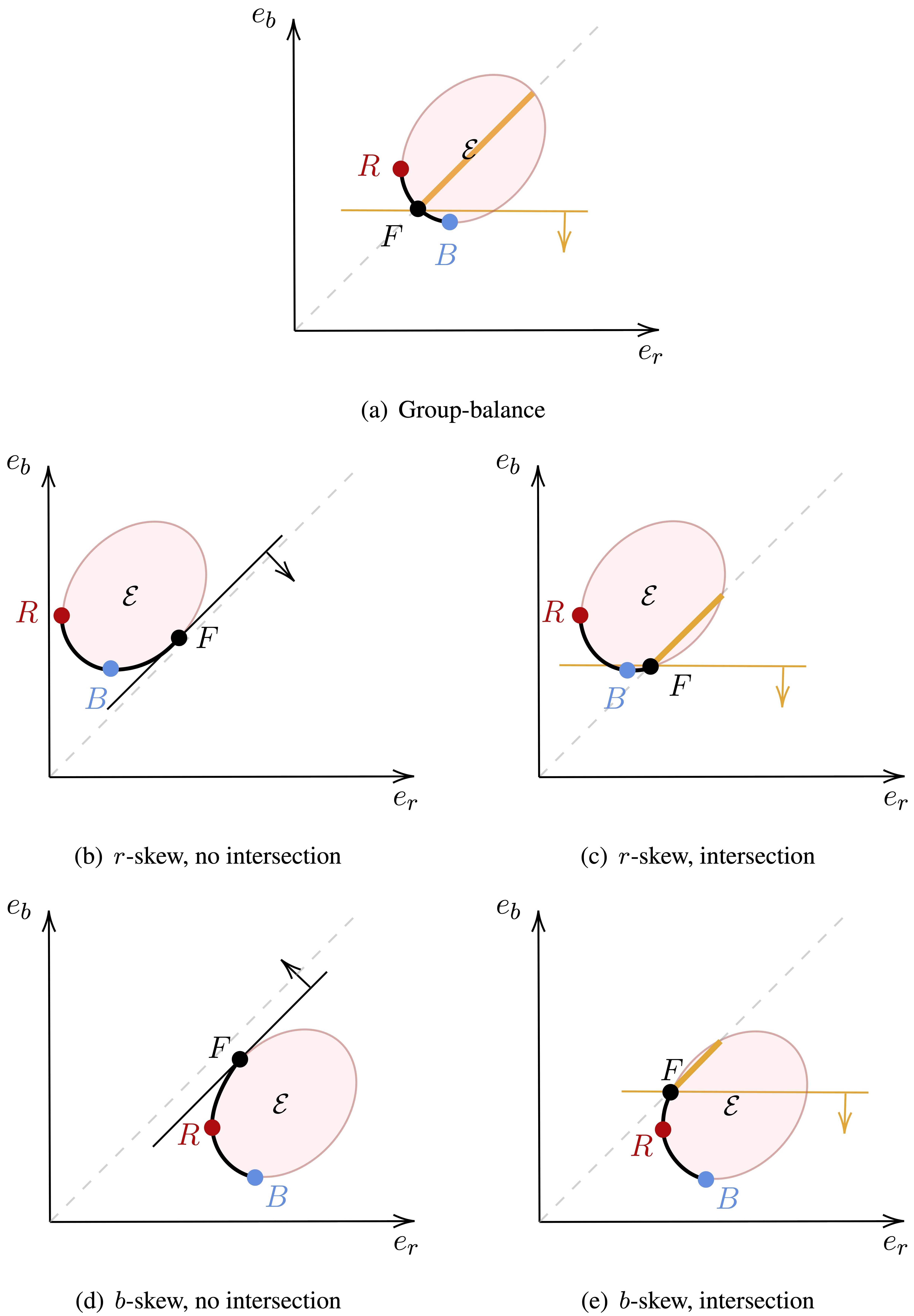}
\caption{\footnotesize{All possible locations of the feasible set $\E$ relative to the 45$^\circ$ line, $\mathcal{H}_{45}$. In panels (a), (c), and (e), $\E$ intersects with $\mathcal{H}_{45}$, and the fairest point $\F$ is the support set of $\Tilde{\E}$ in direction $\mathfrak{u}_1$, where $\Tilde{\E}$ is the intersection between $\E$ and $\mathcal{H}_{45}$ (depicted as an orange line segment). In panels (b) and (d), $\E\cap\mathcal{H}_{45}=\emptyset$, and $\F$ is the support set of ${\E}$ in direction $\mathfrak{u}_2-\mathfrak{u}_1$ for the $r$-skewed case in (b) and $\mathfrak{u}_1-\mathfrak{u}_2$ for the $b$-skewed case in (d).}}\label{fig:F}
\end{figure}
}
Determining the coordinates of the fairest point $\F$ is more laborious, as they depend on whether $\E$ lies entirely above, entirely below, or on top of the 45-degree line. Figure \ref{fig:F} illustrates all possible locations of the feasible set $\E$ relative to the 45-degree line. When $\E$ lies entirely on one side of the 45-degree line, as depicted in panels (b) and (d), $\F$ is the support set of $\E$, respectively, in directions $\mathfrak{u}_2-\mathfrak{u}_1=[1\,\,-1]^\intercal$ and $\mathfrak{u}_1-\mathfrak{u}_2=[-1\,~~\,1]^\intercal$:
\begin{align}
    F&=\ss_\E(\mathfrak{u}_2-\mathfrak{u}_1)
    ~~~~~\text{when } \E \text{ lies entirely above the 45-degree line,}\label{eq:F_above}\\
    F&=\ss_\E(\mathfrak{u}_1-\mathfrak{u}_2)
    ~~~~~\text{when } \E \text{ lies entirely below the 45-degree line.}\label{eq:F_below} 
\end{align}
Complications arise when $\E$ intersects with the 45-degree line, as depicted in panels (a), (c), and (e) of Figure \ref{fig:F}. In this case, the direction at which we can obtain $\F$ as the support set of $\E$ is difficult to determine. To circumvent this challenge, we propose a different approach.
We focus on the convex set that results when $\E$ intersects the 45-degree line:
\begin{align}
    \Tilde{\E}\equiv\E\cap \mathcal{H}_{45}, \quad \text{where} \quad \mathcal{H}_{45}\equiv\{e\in\mathbb{R}^2: e_r=e_b\}.\label{eq:H45}
\end{align}
The new set $\Tilde{\E}$ is depicted in panels (a), (c), and (e) of Figure \ref{fig:F} as an orange line segment. In these cases, $\F$ is the support set of $\Tilde{\E}$ in direction $\mathfrak{u}_1$ with identical values for its two coordinates. Hence,
\begin{align}
    \F=(\mathfrak{u}_1+\mathfrak{u}_2)\h_{\Tilde{\E}}(\mathfrak{u}_1) ~~~~~\text{when } \E \text{ intersects with the 45-degree line.}\label{eq:F_intersect}
\end{align}

We are left with providing an expression for $\h_{\Tilde{\E}}(\q)$. When $\E$ intersects $\mathcal{H}_{45}$, 
\begin{align}
    \h_{\Tilde{\E}}(\q)=\inf_{p_1,p_2\in\mathbb{R}^2: p_1+p_2=\q} \h_{\E}(p_1)+\h_{\mathcal{H}_{45}}(p_2)=\inf_{c\in\mathbb{R}} \h_\E\left(\q-c\begin{bmatrix}
        1\\
        -1
    \end{bmatrix}\right),\label{eq:h_tildeE}
\end{align}
where the first equality follows from \citet[Corollary 16.4.1]{roc97}, and the second follows from the fact that $\h_{\mathcal{H}_{45}}(p_2)$ is bounded from above only along the direction $p_2=c[1\,\,-1]^\intercal$ for any scalar $c\in\mathbb{R}$, in which case $\h_{\mathcal{H}_{45}}(p_2)=0$. 
Importantly, the last expression in Eq.~\eqref{eq:h_tildeE} is always well defined, regardless of whether $\E$ intersects with $\mathcal{H}_{45}$ or not; the infimum equals a bounded scalar in the case of intersection and $-\infty$ otherwise.\footnote{To see this, let $\q(c)\equiv\q-c[1~~-1]^\intercal$ and note that $\h_\E(\q(c))=\|\q(c)\|_E\cdot\h_\E\big(\frac{\q(c)}{\|\q(c)\|_E}\big)$. When $\E$ intersects with $\mathcal{H}_{45}$, $\h_\E\big(\frac{\q(c)}{\|\q(c)\|_E}\big)$ is bounded and nonnegative along the sequences $c\to\infty$ and $c\to-\infty$, but when $\E$ and $\mathcal{H}_{45}$ are disjoint, it takes negative value along one of these sequences, yielding $\inf_c \h_\E(\q(c))=-\infty$.}\label{ftnt:bound_inf_E_tilde}

\vspace{-.5cm}
\subsection{Support Function-Based Characterization of the FA-Frontier}
\label{sec:supp-func-FAfrontier}
We next show that the FA-frontier put forward by \citetalias{lia:lu:mu:oku24} and reproduced in our Eq.~\eqref{eq:FA_frontier} can be characterized using the support function of the feasible set $\E$ and that of an auxiliary set that we introduce in this subsection.

Given an algorithm $a^*\in\mathcal{A}(\sX)$ that induces the risk pair $e^*=[e_r^*, e_b^*]^\intercal\in\E$, let
\begin{align}
\cC(e^*)=\left\{e\in\mathbb{R}^2: e_r\leq e_r^*, e_b\leq e_b^*, |e_r-e_b|\leq |e_r^*-e_b^*|\right\},\label{eq:Cstar}
\end{align}
{
\begin{figure}\captionsetup[subfigure]{font=footnotesize}
\centering
\includegraphics[width=0.9\linewidth]{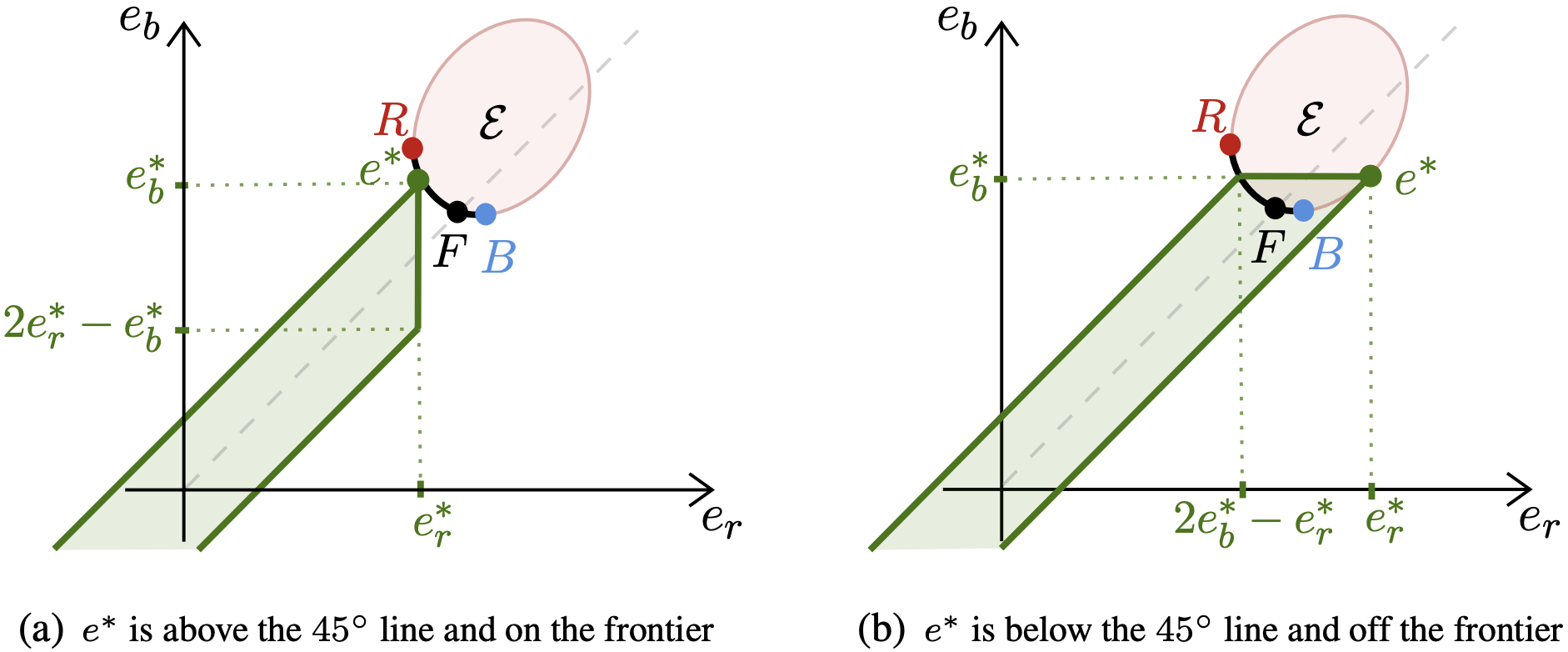}
\caption{\footnotesize{The set $\cC(e^*)$, which collects all improvements relative to $e^*$, is the region shaded in green. Its shape depends on whether $e^*$ lies above or below the $45^\circ$ line. Panel (a) shows an example of the case where $e^*$ lies on the frontier and there exists a hyperplane that properly separates $\cC(e^*)$ and $\E$, whereas in panel (b) $e^*$ is not on the frontier and no hyperplane can separate $\cC(e^*)$ and $\E$. 
}}
\label{fig2}
\end{figure}
}
\noindent denote the set of risk allocations $e\in\mathbb{R}^2$---whether or not they are feasible---that are both weakly more accurate and weakly fairer than $e^*$. 
The set $\cC(e^*)$, depicted in the two panels of Figure \ref{fig2} as the shaded green regions corresponding to two different values of $e^*$, is a closed and convex subset of $\mathbb{R}^2$.
Panel (a) depicts a case where $e^*\in\sF$, whereas panel (b) depicts a case where $e^*\notin\sF$. The key insight from the figure, which we prove can be used to characterize the FA-frontier using the support function of $\E$ and that of $\cC(e^*)$, is that for any $e^*\in\sF$, the sets $\E$ and $\cC(e^*)$ can be properly separated \citep[see, e.g.,][p.12, for a definition of proper separation]{sch93}, while in the case where $e^*\notin\sF$ they cannot.
\begin{prop}
\label{prop:sh}
\textit{
Under Assumptions~\ref{asm:moments}-\ref{asm:smooth-X}, $e^*\in\sF$ if and only if there exists a hyperplane that properly separates $\cC(e^*)$ and $\E$, i.e., there exists $\q\in\bS^1$ such that $\h_{\sC}(\q)=-\h_\E(-\q)$. 
Let $\tilde{\bS}^1\equiv\bS^1\setminus \{\q\in\bS^1:\q_1+\q_2<0\}$ and $[\,\,\cdot\,\,]_-\equiv-\min\{\,\cdot\,, 0\}$.
We then have that
\begin{align}
    \sF &= \left\{e^*\in\E:~\left[\max_{\q\in\tilde{\bS}^1}(-\h_{\cC(e^*)}(\q)-\h_\E(-\q))\right]_-=0\right\}.\label{eq:FA_frontier_supp_func}
\end{align}
}
\end{prop}
Remarkably, this characterization of the FA-frontier does not require knowledge of whether $\E$ intersects the $45^o$ line, or whether there is group balance or skew.
As we show in Sections~\ref{sec:set_estimation}-\ref{sec:test}, the characterization of the FA-frontier in Eq.~\eqref{eq:FA_frontier_supp_func} is very helpful for estimation and inference, and for testing whether there exists an LDA to a given algorithm.

\section{Support Function Estimator and Its Asymptotic Distribution}
\label{sec:estimation}

In practice, the policymaker does not have perfect knowledge of $\bP$. Hence, $\h_{\E}(\q)$ can only be estimated from a finite sample. Let a sample of size $n$, $\{(\Y_i,\G_i,\X_i)\}_{i=1}^n$, drawn independently and identically from $\bP$, be available.
Recall Eq.~\eqref{eqn:sf-new}: $\h_\E(\q)=\mathbb{E}\bigl[(\cM\q)^\intercal\bLL_0+(\cM\q)^\intercal(\bLL_1-\bLL_0)\mathds{1}\{k(\btheta(\X),\cM\q)>0\}\bigr]$ with $k(\btheta,\cM\q)\equiv(\cM\q)^\intercal\bigl(\btheta_1(\X)-\btheta_0(\X)\bigr)$, so that $\btheta(\X)\equiv[\theta_1^r(\X) ~ \theta_0^r(\X) ~ \theta_1^b(\X) ~ \theta_0^b(\X)]^\intercal$ enters the expression for $\h_\E(\q)$ only through $(\theta_1^r(\X)-\theta_0^r(\X))$ and $(\theta_1^b(\X)- \theta_0^b(\X))$.
We propose estimating $\h_{\E}(\q)$ by first estimating the finite dimensional parameters $\cM$ by sample averages and the nuisance functions $\Delta\btheta(\X)\equiv[(\theta_1^r(\X)-\theta_0^r(\X)) ~~ (\theta_1^b(\X)- \theta_0^b(\X))]^\intercal$ by flexible machine learning methods (allowing the complexity of the parameter space containing the estimator to grow with sample size), and then plugging their estimators, denoted $\hcM$ and $\widehat{\Delta\btheta}(\X)$, into the sample analogue of Eq.~\eqref{eqn:sf-new}.
Following the literature on debiased machine learning \cite[e.g., ][]{new94,
che:che:dem:duf:han:whi:rob18, sem:che21, che:esc:ich:new:rob22, ich:new22}, we show that, at the population $\cM$, the moment in Eq.~\eqref{eqn:sf-new} is Neyman-orthogonal and hence ``insensitive'' to the errors in the first-stage estimation of $\Delta\btheta$ \citep{ney79, ney95}. We use sample splitting to relax the otherwise needed Donsker condition that limits the complexity of the relevant parameter space \cite[e.g., ][]{bic82, rob:li:tch:van08, rob:li:muk:tch:van17}, and we account for the estimation error of $\cM$. 

Recall that $\bLL_d\equiv[\LL_d^r~~\LL_d^b]^\intercal$, with $\LL_d^g\equiv \ell(d,\Y)\mathds{1}\{\G=g\}$, and $\theta_d^g(\X)=\mathbb{E}[\LL_d^g|\X]$. Let $\Delta\LL^g\equiv\LL_1^g-\LL_0^g$, so that $\Delta\theta^g(\X)=\mathbb{E}[\Delta\LL^g|\X]$. To learn the nuisance parameter $\Delta\theta^g(\X)$, the effective label that, given $\X$, we train machine learners to predict is $\Delta\LL^g$. Let $\Theta$ denote the convex nuisance parameter space (a subset of a vector space with $L^2(\bP)$ norm) to which $\Delta\btheta=[\Delta\theta^g(\X)~~\Delta\theta^b(\X)]^\intercal$ belongs and $\Delta\bvartheta\equiv[\Delta\vartheta^r(\X)~~\Delta\vartheta^b(\X)]^\intercal$ be a generic element from $\Theta$ (to simplify notation, we drop the dependence of $\btheta$ on $\X$ unless explicitly needed).
For the $i$-th observation and a given $2\times2$ diagonal matrix $\mathring{\cM}$, we note that
\begin{align}
k\bigl(\bvartheta(\X_i),\mathring{\cM}\q\bigr)=\q^\intercal\mathring{\cM}\Delta\bvartheta(\X_i),\label{eq:k_Delta}
\end{align}
and using Eq.~\eqref{eq:k_Delta} to recognize that we estimate $\Delta\btheta$ while keeping the notation as close as possible to that in Eq.~\eqref{eqn:sf-new}, we define the mapping $\zeta_i(\mathring{\cM}\q;\,\cdot\,) : \Theta \to \mathbb{R}$ as
\begin{align} 
\zeta_i(\mathring{\cM}\q;\bvartheta)&\equiv(\mathring{\cM}\q)^\intercal\bLL_{0_i}+(\mathring{\cM}\q)^\intercal(\bLL_{1_i}-\bLL_{0_i})\cdot\mathds{1}\bigl\{k\bigl(\bvartheta(\X_i),\mathring{\cM}\q\bigr)>0\bigr\}.\label{eqn:influence-part}
\end{align}
By Proposition \ref{prop:sf}, $\h_{\E}(\q)=\mathbb{E}[\zeta_i(\cM\q;\btheta)]$. We next show that, when $\mathring{\cM}$ is fixed at the population $\cM$, the score function $\zeta_i(\cM\q;\bvartheta)$ is Neyman-orthogonal at $\Delta\bvartheta=\Delta\btheta$. 
\begin{prop}
\label{prop:orthogonality}
\textit{Let Assumptions~\ref{asm:moments}-\ref{asm:smooth-X} hold and $\sup_{\Delta\bvartheta\in\Theta}\|\Delta\bvartheta-\Delta\btheta\|_{L^2(\bP)}<\infty$. Then the map $\Delta\bvartheta\mapsto\mathbb{E}[\zeta_i(\cM\q;\bvartheta)]$ satisfies the Neyman orthogonality condition uniformly in $\q\in\bS^1$, i.e., for any $\Delta\bvartheta\in\Theta$ and scalar $t\in(0,1)$,
\vspace{-.25cm}
\begin{align*}
    \lim_{t\to0}\sup_{\q\in\bS^1}\biggl|\frac{1}{t}\left(\mathbb{E}\bigl[\zeta_i\bigl(\cM\q;\btheta+t(\bvartheta-\btheta)\bigr)\bigr]-\mathbb{E}\bigl[\zeta_i(\cM\q;\btheta)\bigr]\right)\biggr|=0.
\end{align*}}
\end{prop}
Intuitively, Proposition \ref{prop:orthogonality} shows that, under its maintained assumptions, the first-order mistake in the sign of $k(\btheta(\X),\cM\q)$ due to the estimation error in $\Delta\btheta(\X)$ is negligible.

We next show that, provided $n^{1/4}$-consistent first-stage estimators are available, the estimated support function using sample splitting and cross-fitting, as described in Definition \ref{def:cross-fitting} below, converges to a Gaussian process uniformly in $\q\in\bS^1$. The proof requires showing the residual in estimating the indicator functions is bounded by the quadratic rate of convergence of the nuisance parameter, which we establish under this assumption:
\begin{asm}[(Nuisance Parameter Structure)]
\label{asm:nuisance-structure}
\textit{
There is a known partition of $\X$, $$\X=(\X_1,\X_2,\X_{[3:d_\X]}),$$ where $\X_1,\X_2\in\mathbb{R}$ and $\X_{[3:d_\X]}\in\mathbb{R}^{(d_X-2)}$ are such that $(\X_1,\X_2)$ has a bounded support, the density of $|k(\btheta,\cM\q)|$ conditional on $\X_{[3:d_\X]}$ is uniformly bounded in $\q\in\bS^1$, and
\begin{align*}   
\Delta\theta^g(\X)=\alpha^g\X_1+\beta^g\X_2+\eta^g(\X_{[3:d_\X]}), 
\end{align*}
for some $\alpha^g, \beta^g\in\mathbb{R}$ satisfying 
$\alpha^b\cdot\beta^r\neq\alpha^r\cdot\beta^b$ and $\eta^g\in H$ for convex $H\subseteq\Theta$, $g\in\{r, b\}$.
}
\end{asm}

Assumption \ref{asm:nuisance-structure} requires 
$\Delta\btheta(\X)$ 
to be linearizable in a known set of covariates $(\X_1,\X_2)$ and that $\alpha^b\cdot\beta^r\neq\alpha^r\cdot\beta^b$. 
This assures that for each $\q\in\bS^1$, $k(\btheta,\cM\q)$ depends on at least one of $\X_1$ or $\X_2$, so that we can employ a proof technique similar to that in \citet[Lemma 4.1]{sem23} to show that the bias induced by errors in estimating the sign of $k(\btheta,\cM\q)$ is bounded by the desired quadratic $L^2$-rate of the nuisance estimator, accommodating a menu of flexible machine learners in the first stage. For example, one can adapt the method in \citet[pp.~935-936]{rob88} to estimate $\Delta\theta^g$.
Example machine learning methods that, under suitable conditions, are $n^{1/4}$-consistent in the $L^2$-norm include $\ell_1$-penalized methods, boosting, regression trees and forests, and neural nets \citep[see, e.g.,][and references therein, for a discussion of machine learners compatible with the rate requirement]{che:che:dem:duf:han:whi:rob18}. Assumption \ref{asm:nuisance-structure} can be eliminated at the price of using a more restrictive class of machine learning methods, by deriving rate bounds that depend on the squared $L^\infty$-rate of the first-stage estimation \citep[e.g., the Lasso, for which an $L^\infty$-rate is established for approximately sparse models, see][]{bel:che:fer:han17, sem23}. Alternatively, Assumption \ref{asm:nuisance-structure} can be eliminated by smoothing the indicators \citep[e.g.,][]{che:aus:syr23, par24} at the cost of introducing an additional tuning parameter that controls the degree of smoothing.\footnote{An earlier version of this paper \citep{liu:mol24v1} obtains $\sqrt{n}$-Gaussianty of the second-stage estimator without imposing Assumption \ref{asm:nuisance-structure}, but under the classical Donsker condition that limits $\Theta$'s complexity.} 
Importantly, Assumption \ref{asm:nuisance-structure} implies the margin condition in Assumption \ref{asm:smooth-X} with $m=1$ whenever $(\X_1,\X_2)$ is continuously distributed with a bounded density and $\X_{[3:d_\X]}$ can be either continuous or discrete; see Assumption \ref{asm:condition_smoothX1X2} and Proposition \ref{prop:no_kinks_condition}.

We next define cross-fitting, adapting Definition 3.2 in \citet{che:che:dem:duf:han:whi:rob18}:
\begin{defn}[Cross-Fitting]
\label{def:cross-fitting}
\textit{(i) Randomly partition the size-$n$ sample with observations indexed by $i\in[n]\equiv\{1,...,n\}$ to $K\geq2$ subsamples, each of size $n/K$ (assumed to be an integer), where $K$ is a fixed integer. (ii) For each partition $k\in[K]\equiv\{1,...,K\}$ with observations indexed by the set $I_k\subset[n]$, estimate $\Delta\btheta$ by $\widehat{\Delta\btheta}_k\equiv[(\widehat{\Delta\theta}^r)_k~~(\widehat{\Delta\theta}^b)_k]^\intercal$, where each $(\widehat{\Delta\theta}^g)_k=(\widehat{\alpha}^g)_k\X_1+(\widehat{\beta}^g)_k\X_2+(\widehat{\eta}^g)_k(\X_{[3:d_\X]})$ is estimated
using only observations from $I_k^c\equiv[n]\backslash I_k$. For $i\in I_k$, let $\widehat{\Delta\btheta}(\X_i)\equiv\widehat{\Delta\btheta}_k(\X_i)$. (iii) Let $\hcM=\diag(1/\widehat{\mu}_r,1/\widehat{\mu}_b)$, with $\widehat{\mu}_g\equiv\frac{1}{n}\sum_{i=1}\mathds{1}\{G_i=g\}$, and construct the second-stage estimator as 
\begin{align}
\widehat{\h}_{\E}(\q; \widehat{\btheta})&=\frac{1}{K}\sum_{k\in[K]}\left(\frac{1}{n/K}\sum_{i\in I_k}\zeta_i(\hcM\q;\widehat{\btheta}_k)\right)\equiv\frac{1}{n}\sum_{i=1}^n\zeta_i(\hcM\q;\widehat{\btheta}).\label{eqn:cross-fit-second-stage}
\end{align}
}
\end{defn}
In Eq.~\eqref{eqn:cross-fit-second-stage}, to simplify notation and keep it as close as possible to that in Eq.~\eqref{eqn:sf-new} and $\h_{\E}(\q)=\mathbb{E}[\zeta_i(\cM\q;\btheta)]$, we use the shorthand notation $\widehat{\Delta\btheta}(\X_i)\equiv\widehat{\Delta\btheta}_k(\X_i)$ for $i\in I_k$, suppress the dependence of $\widehat{\h}_{\E}(\q; \widehat{\btheta})$ on $\hcM$, and adapt Eqs.~\eqref{eq:k_Delta}-\eqref{eqn:influence-part} to let 
\begin{align}
    \zeta_i(\hcM\q;\widehat{\btheta})&=
    (\hcM\q)^\intercal\bLL_{0_i}+(\hcM\q)^\intercal(\bLL_{1_i}-\bLL_{0_i})\cdot\mathds{1}\bigl\{k\bigl(\widehat{\btheta}(\X_i),\hcM\q\bigr)>0\bigr\},\label{eq:zeta_estimated}\\
    k\bigl(\widehat{\btheta}(\X_i),\hcM\q\bigr)&=\q^\intercal\hcM\widehat{\Delta\btheta}(\X_i)\label{eq:k_estimated}
\end{align}

Our main asymptotic result shows that $\widehat{\h}_{\E}(\q; \widehat{\btheta})$ converges to a Gaussian process uniformly in the direction $\q\in\bS^1$, where the score $\zeta_i(\cM\q;\bvartheta)$ in Eq.~\eqref{eqn:influence-part}  evaluated at $\Delta\bvartheta=\Delta\btheta$ is the influence function that governs the part of the asymptotic distribution of $\widehat{\h}_{\E}(\q; \widehat{\btheta})$ due to $\widehat{\Delta\btheta}$, and the the remaining part is attributed to estimating $\cM$:
\begin{theorem}
    \label{thm:gaussian}
\textit{
    Let Assumptions \ref{asm:moments}-\ref{asm:smooth-X}-\ref{asm:nuisance-structure} hold and $\{(\Y_i,\G_i,\X_i)\}_{i=1}^n$ be a random sample from $\bP$.
    Define a shrinking neighborhood around $\Delta\btheta$ as
    \begin{align*}
        &\Theta_n\equiv\bigl\{\Delta\bvartheta\in\Theta: 
        \forall g\in\{0,1\}, \Delta\vartheta^g(\X)=\tilde{\alpha}^g\X_1+\tilde{\beta}^g\X_2+\tilde{\eta}^g(\X_{3:d_\X}),\\
        &\hspace{3cm}\max\{|\tilde{\alpha}^g-\alpha^g|, |\tilde{\beta}^g-\beta^g|, \|\tilde{\eta}^g-\eta^g\|_{L^2(\bP)}\}=o(n^{-1/4})\bigr\}.
    \end{align*}
    Let $\widehat{\Delta\btheta}_k\in \Theta_n$ with probability approaching 1, $\forall k\in[K]$. 
    Then, for $\widehat{\h}_{\E}(\q; \widehat{\btheta})$ in Eq.~\eqref{eqn:cross-fit-second-stage},
    \begin{align*}
    \sqrt{n}\biggl(\widehat{\h}_{\E}(\q; \widehat{\btheta})-\h_{\E}(\q)\biggr)=\mathbb{G}[ \zeta_i^*(\cM\q;\btheta)]+o_p(1) \quad\text{ in }\,\, \ell^\infty(\bS^1),
    \end{align*}
    where $\mathbb{G}[\zeta_i^*(\cM\q;\btheta)]$ is a Gaussian process in $\ell^\infty(\bS^1)$ indexed by
    \begin{align*}
        \zeta_i^*(\cM\q;\btheta)\equiv\zeta_i(\cM\q;\btheta)+(\cM_i^*\q)^\intercal\cM^{-1}\ss_\E(\q), ~~\text{for}~\cM_i^*\equiv\diag\left(\frac{\mathds{1}\{\G_i=r\}}{-\mu_r^2}, \frac{\mathds{1}\{\G_i=b\}}{-\mu_b^2}\right)
    \end{align*}
    with $\zeta_i(\cM\q;\btheta)$ defined in Eq.~\eqref{eqn:influence-part} and the covariance function of $\mathbb{G}[\zeta_i^*(\cM\q;\btheta)]$ equal to
\begin{align}
\Omega(\q, \Tilde{\q})=\mathbb{E}[\zeta_i^*(\cM\q;\btheta)\zeta_i^*(\cM\Tilde{\q};\btheta)]-\mathbb{E}[\zeta_i^*(\cM\q;\btheta)]\mathbb{E}[\zeta_i^*(\cM\Tilde{\q};\btheta)].\label{eq:omega}
\end{align}
If $Var(\bLL_d|\X)$ is positive definite, then $Var(\mathbb{G}[\zeta_i^*(\cM\q;\btheta)])>0$ for each $\q\in\bS^1$.
}
\end{theorem}

\section{Estimation and Inference for the Frontier}
\label{sec:set_estimation}

\subsection{Estimation and Inference for the Feasible Set}
\label{subsec:estimate_feasible_set}
We propose an estimator of the set $\E$ based on $\hn_\E(\q)$ and Eq.~\eqref{eq:E_through_h}, given by
\begin{align}
    \widehat{\E}\equiv\bigcap_{\q\in\bS^1}\left\{z\in\mathbb{R}^2:\q^\intercal z\leq \hn_\E(\q;\widehat{\btheta})\right\},\label{eq:Lambda_hat}
\end{align}
which is convex, almost surely compact, and non-empty with probability approaching one if $\E$ has a non-empty interior. As the Hausdorff distance between two non-empty convex and compact sets $A,B\in\mathbb R^d$, denoted $\mathbf{d}_H(A,B)$, equals the uniform distance between their support functions \citep[][p.~101]{mol:mol18}, when $\widehat{\E}$ is non-empty we have
\begin{align*}
    \mathbf{d}_H(\widehat{\E},\E)=\sup_{\q\in\bS^1}\left|\hn_\E(\q;\widehat{\btheta})-\h_{\E}(\q; \btheta)\right|,
\end{align*}
By Theorem \ref{thm:gaussian} and the continuous mapping theorem, $\mathbf{d}_H(\widehat{\E},\E)\xrightarrow[]{p} 0$ and $\sqrt{n}\mathbf{d}_H(\widehat{\E},\E)\xrightarrow[]{d}\sup_{\q\in\bS^1}|\mathbb{G}[  \zeta_i^*(\cM\q;\btheta)]|$.
Hence, asymptotically valid tests of hypotheses about $\E$, and confidence sets covering it, can be obtained as in \citet[Section 2]{ber:mol08}.

\subsection{Estimation and Inference for the FA-Frontier}
\label{subsec:FA_frontier}
As shown in \citetalias{lia:lu:mu:oku24} (Theorem 1), when $\big(\bP, \mathcal{A}(\sX)\big)$ is group-balanced, $\sF$ is the curve connecting $\R$ and $\B$ and coincides with the Pareto frontier (panel (a) in Figure \ref{fig:F}).
However, when $\big(\bP, \mathcal{A}(\sX)\big)$ is $g$-skewed, $\sF$ is the curve connecting $\F$ with the feasible point that minimizes the risk for group $g$ (panels (b)-(e) in Figure \ref{fig:F}).
A further challenge is that, while it is simple to express $\F$ through $\ss_\E$ when $\E$ is fully contained in one of the two half-spaces defined by the 45-degree line $\mathcal{H}_{45}$, as shown in Eqs.~\eqref{eq:F_above}-\eqref{eq:F_below} (panels (b) and (d) in Figure \ref{fig:F}), characterizing $\F$ is not as straightforward when $g$-skew occurs but $\E\cap\mathcal{H}_{45}\neq\emptyset$.

We therefore leverage the characterization of the FA-frontier $\sF$ in Proposition \ref{prop:sh}, whereby $\sF = \left\{e\in\E:~\left[\max_{\q\in\tilde\bS^1}(-\h_{\cC(e)}(\q)-\h_\E(-\q))\right]_-=0\right\}$, and the definition of the set $\cC(\cdot)$ in Eq.~\eqref{eq:Cstar} to sidestep these difficulties. We propose to estimate $\sF$ using
\begin{align}
\small
    \widehat{\sF} = \left\{e\in\bB_C:~\left[\max_{\q\in\bS^1}(\q^\intercal e-\widehat{\h}_\E(\q;\widehat{\btheta}))\right]_{+}+\left[\max_{\q\in\tilde\bS^1}(-\h_{\cC(e)}(\q)-\hn_\E(-\q;\widehat{\btheta}))\right]_-\le\frac{\kappa_n}{\sqrt{n}}\right\},\label{eq:Fhat}
\end{align}
where $[\,\,\cdot\,\,]_+\equiv\max\{\,\cdot\,, 0\}$, $\kappa_n=o(\sqrt{n})$ is a sequence that diverges to infinity, and $\bB_C\equiv\{e\in\mathbb{R}^2:\Vert e\Vert_E \le C\}$, with $\E\subset \bB_C$ by Assumption~\ref{asm:moments} 
and $C<\infty$ a constant pinned down by $c_1,c_2$ defined in Assumption~\ref{asm:moments}. 
The first maximization problem in Eq.~\eqref{eq:Fhat} is the sample analog to the requirement that $e\in\E$; the second is the sample analog to the requirement that $\E$ and $\cC(e)$ can be properly separated.
To implement this estimator, we derive a closed-form expression for the support function of $\cC(e)$ in direction $\q=[\q_1,~\q_2]^\intercal$. Only two points are ``active'' for evaluating $\h_{\cC(e)}(\q)$: $[\min\{e_r,2e_b-e_r\},\,\,e_b]^\intercal$ and $[e_r,\,\,\min\{e_b, 2e_r-e_b\}]^\intercal$, which correspond to the points $e$ and $[e_r,\,\,2e_r-e_b]^\intercal$ (respectively, $[2e_b-e_r,\,\,e_b]^\intercal$ and $e$) when $e$ is above (respectively, below) the 45-degree line as shown in Figure \ref{fig2}-panel (a) (respectively, panel (b)). By Proposition~\ref{prop:sh}, it is without loss of generality to focus on $\q\in\tilde{\bS}^1$. Hence, the support function of $\cC(e)$ at any $\q\in\tilde{\bS}^1$ equals
\begin{align}
    \h_{\cC(e)}(\q)=\max\biggl\{\q_1\min\{e_r,2e_b-e_r\}+\q_2e_b, \,\,\,\q_1 e_r + \q_2\min\{e_b, 2e_r-e_b\}\biggr \}.\label{eqn:sf-C}
\end{align}

We next propose a test statistic for the null hypothesis $\HH_0:~e\in\sF$, against the alternative $\HH_A:~e\notin\sF$. 
Our test statistic is given by
\begin{align}
    T_n^\sF(e)\equiv
    \sqrt{n}\left(\left[\max_{\q\in\bS^1}(\q^\intercal e-\widehat{\h}_\E(\q;\widehat{\btheta}))\right]_{+} +\left[\max_{\q\in\tilde\bS^1}(-\h_{\cC(e)}(\q)-\hn_\E(-\q;\widehat{\btheta}))\right]_-\right).\label{eq:Tn_sF}
\end{align}
As shown in the proof of Proposition~\ref{prop:consistency_coverage_sF}, if Eq.~\eqref{eq:no_kinks_cond} is satisfied and consequently $\E$ has no kinks, the large sample distribution of $T_n^\sF(e)$, denoted $\psi^{\sF}(e)$, simplifies to
\begin{align}
    \psi^{\sF}(e)= \left|\mathbb{G}\left[{\zeta_i^*(\cM\q^*_{\bS^1}(e);\btheta)}\right]\right|,\label{psi_FA_no_kinks}
\end{align}
where $\q^*_{\bS^1}(e)\equiv\arg\max_{\q\in\bS^1}\q^\intercal e -\h_\E(\q)$.
The quantiles of this distribution can be estimated by standard methods.
We build a confidence set for the elements of $\sF$ by test inversion:
\begin{align}
    \mathcal{CS}_n(\sF) = \left\{e\in\bB_C:T_n^\sF(e)\le c_{1-\alpha}^{\sF}(e)\right\},\label{eq:CS_sF}
\end{align}
where for any $\beta\in(0,1)$, $c_\beta^{\sF}$ is the $\beta$-quantile of $\psi^{\sF}$.
When Eq.~\eqref{eq:no_kinks_cond} is not assumed and $\E$ has kinks, the expression for $\psi^{\sF}(e)$  is more complex and given in Eq.~\eqref{psi-FA}. 
In this case, restrictions that would guarantee uniform continuity and strict increasing properties for $\psi^{\sF}(e)$ are harder to verify, and we follow \citet[p.625]{and:shi13} to replace the confidence set in Eq.~\eqref{eq:CS_sF} with $\mathcal{CS}_n(\sF) = \left\{e\in\bB_C:T_n^\sF(e)\le c_{1-\alpha+\varsigma}^{\sF}(e)+\varsigma\right\}$, for $\varsigma>0$ an arbitrarily small constant.
In this case, the critical value can be approximated through bootstrap methods, as in Procedure~\ref{proc:bootstrap} in Section~\ref{sec:test-LDA}.

We next establish consistency of our estimator $\widehat{\sF}$ and validity of the confidence set, building respectively on \citet{che:hon:tam07} and \citet{fan:san19}.
\begin{prop}\label{prop:consistency_coverage_sF}
\textit{
    Let the assumptions of Theorem~\ref{thm:gaussian} hold and let $\kappa_n\to\infty$ with $\kappa_n=o(\sqrt{n})$. Then, as $n\to\infty$, 
    \begin{align}
        \mathbf{d}_H(\widehat{\sF},\sF)&\xrightarrow[]{p} 0\label{eq:consistency_F_hat}\\
        \liminf_{n\to\infty}\mathbb{P}\left(e\in\mathcal{CS}_n(\sF)\right)&\ge 1-\alpha~~~\text{for all}~e\in\sF.\label{eq:validity_CS(F)}
    \end{align}
    }
\end{prop}

\subsection{Estimation and Inference for the Pareto Frontier}
\label{subsec:pareto_frontier}
The Pareto Frontier ($\sPF$) is the lower boundary of the feasible set $\E$ connecting points $R$ and $B$ (see Figure \ref{fig1}).
Denoting $\bQ\equiv\left\{\q\in\bS^1:\q=[\cos\gamma~~\sin\gamma]^\intercal,~\gamma\in[\pi,3/2\pi]\right\}$, we can express $\sPF$ through the support set $\ss_\E(\cdot)$, or equivalently through the support function:
\begin{subequations}
\begin{align}
    \sPF&\equiv\left\{\ss_\E(\q):\q\in\bQ\right\}\label{eq:PFs}\\
    &=\left\{ e\in\mathbb{R}^2: \left[\max_{\q\in\bS^1}(\q^\intercal e - \h_\E(\q))\right]_+= 0,\,\left[\max_{\q\in\bQ}(\q^\intercal e - \h_\E(\q))\right]_-=0\right\},\label{eq:PFh}
\end{align}
\end{subequations}
where the first condition in Eq.~\eqref{eq:PFh} enforces that $e\in\E$ and the second that it belongs to the supporting hyperplane of $\E$ in a direction $\q\in\bQ$.
Knowing the directions that determine the support points comprising $\sPF$ simplifies the expression for it in Eq.~\eqref{eq:PFh} relative to Eq.~\eqref{eq:FA_frontier_supp_func}.
It also allows for an estimator, put forward in Proposition~\ref{prop:PFconsistency} below, based on the support set characterization in Eq.~\eqref{eq:PFs} that is free from the tuning parameter $\kappa_n$, which instead is needed for the estimator of $\sF$ in Eq.~\eqref{eq:Fhat}.
Let $\bzeta_{\ss,i}(\hcM\q;{\widehat{\btheta}})\equiv\hcM\bLL_{0_i}+\hcM(\bLL_{1_i}-\bLL_{0_i})\cdot\mathds{1}\bigl\{k\bigl({\widehat{\btheta}(\X_i)},\hcM\q\bigr)>0\bigr\}$, with $k\bigl(\widehat{\btheta}(\X_i),\hcM\q\bigr)=\q^\intercal\hcM\widehat{\Delta\btheta}(\X_i)$ as in Eq.~\eqref{eq:k_estimated} and\footnote{As in Eq.~\eqref{eqn:cross-fit-second-stage}, this expression is a shorthand for $\widehat{\ss}_\E(\q;\widehat{\btheta})\equiv\frac{1}{K}\sum_{k\in[K]}\left(\frac{1}{n/K}\sum_{i\in I_k}\bzeta_{\ss,i}({\hcM}\q;\widehat{\btheta}_k)\right)$.}
\begin{align}
    \widehat{\ss}_\E(\q;\widehat{\btheta})=\frac{1}{n}\sum_{i=1}^n\bzeta_{\ss,i}({\hcM}\q;\widehat{\btheta}).\label{eq:hat_ss}
\end{align}
Proposition \ref{prop:PFconsistency} delivers an estimator for $\sPF$ based on $\widehat{\ss}_\E(\q;\widehat{\btheta})$ in Eq.~\eqref{eq:hat_ss} and establishes its Hausdorff-consistency.
\begin{prop}\label{prop:PFconsistency}
\textit{
    Let the assumptions of Theorem~\ref{thm:gaussian} hold. Let $\widehat{\sPF}\equiv\left\{\widehat{\ss}_\E(\q;\widehat{\btheta}):\q\in\bQ\right\}$. Then, as $n\to\infty$, 
    \begin{align}
        \max_{\q\in\bQ}\Vert \widehat{\ss}_\E(\q;\widehat{\btheta})-\ss_\E(\q)\Vert_E &\xrightarrow[]{p} 0,\label{eq:consistency_S_hat}\\
        \mathbf{d}_H(\widehat{\sPF},\sPF)&\xrightarrow[]{p} 0.\label{eq:consistency_PF_hat}
    \end{align}
}
\end{prop}

We next provide a method to test, for a given $e\in\mathbb{R}^2$,
\begin{align}
    \HH_0:~e\in\sPF~~~\text{against}~~~\HH_A:~e\notin\sPF,\label{eq:H0_PF}
\end{align}
using the characterization of $\sPF$ in Eq.~\eqref{eq:PFh}.
We first propose a test statistic and derive its asymptotic distribution in Proposition~\ref{prop:PFlimit_distr} below, building on results in \citet{kai16}. 
\begin{remark}\label{remark:est-supp-set}
We use different characterizations of $\sPF$ for estimation and for inference due to difficulties with DML estimation of $\ss_\E(\q)$ that we explain here. Using Eq.~\eqref{eq:ss} we can write the first (second) coordinate of $\ss_\E(\q)$ as $\mathbb{E}\big[(\cM\v)^\intercal\bLL_0+(\cM\v)^\intercal(\bLL_1-\bLL_0)\mathds{1}\{k(\btheta(\X),\cM\q)>0\}\big]$ for $\v=[1~~0]^\intercal$ ($\v=[0~~1]^\intercal$). As $k(\btheta(\X),\cM\v)=\mathbb{E}[(\cM\v)^\intercal(\bLL_1-\bLL_0)|\X]$, the proof of Theorem~\ref{thm:gaussian} shows that, if $\v=\q$, the first-stage estimation error in the sign of $k(\btheta(\X),\cM\q)$ is controlled by the size of the error itself (because in case of sign disagreement, $|k(\btheta(\X),\cM\q)|\le|k(\widehat{\btheta}(\X),\cM\q)-k(\btheta(\X),\cM\q)|$), with $k(\widehat{\btheta}(\X),\cM\q)$ as in Eq.~\eqref{eq:k_estimated}. However, when $\v\neq\q$, which necessarily occurs for at least one coordinate of $\ss_\E(\q)$, sign errors are not controlled for.  
Consequently, we switch to the moment inequality characterization in Eq.~\eqref{eq:PFh} that involves $\h_\E(\q)$ only. If one uses a parametric estimator for $\Delta\btheta(\X)$, the asymptotic distribution of $\widehat{\ss}_\E(\cdot;\widehat{\btheta})$ can be obtained and inference is simplified, as a special case of the treatment in \citet{liu:mol24v1}, who work with a sieve nonparametric estimator of $\Delta\btheta(\X)$ under Donsker conditions.
\end{remark}

\begin{prop}\label{prop:PFlimit_distr}
\textit{
    Let $Q_{\mathbb{T}}^*(e)\equiv\arg\max_{\q\in\mathbb{T}}\q^\intercal e -\h_\E(\q),~\mathbb{T}\in\{\bS^1,\bQ\}$.
    Then, under the null in Eq.~\eqref{eq:H0_PF} and the assumptions of Theorem~\ref{thm:gaussian},
    \begin{align}
        T_n^\sPF(e)&\equiv \sqrt{n}\left(\left[\max_{\q\in\bS^1}\q^\intercal e-\widehat{\h}_\E(\q;\widehat{\btheta})\right]_{+} + \left[\max_{\q\in\bQ}\q^\intercal e-\widehat{\h}_\E(\q;\widehat{\btheta})\right]_-\right)\label{eq:Tn_PF}\\
        &\xrightarrow[]{d}\left[\sup_{\q\in Q^*_{\bS^1}(e)}\mathbb{G}[-{\zeta_i^*(\cM\q;\btheta)}]\right]_+ + \left[ \sup_{\q\in Q^*_{\bQ}(e)}\mathbb{G}[-{\zeta_i^*(\cM\q;\btheta)}]\right]_-.\label{eq:Tn_PF_limit}
    \end{align}
    If $Var(\bLL_d|\X)$ is positive definite for each $d\in\{0,1\},~\X-a.s.$, the limit law in Eq.~\eqref{eq:Tn_PF_limit} is absolutely continuous with respect to Lebesgue measure on $\mathbb{R}_{++}$.
}
\end{prop}
If Eq.~\eqref{eq:no_kinks_cond} is satisfied and the set $\E$ has no kinks (see Remark~\ref{remark:kinks}), $Q_{\bS^1}^*(e)$ and $Q_{\bQ}^*(e)$ are singletons that can be consistently estimated through standard methods, and the limit distribution in Eq.~\eqref{eq:Tn_PF_limit} coincides with that in Eq.~\eqref{psi_FA_no_kinks}.
If instead kinks are not ruled out for $\q\in\bQ$, $Q_{\bS^1}^*(e)$ and $Q_{\bQ}^*(e)$ may not be singletons. In this case we can consistently estimate these sets \citep[Lemma D.5]{kai16} as
\begin{align}
    \widehat{Q}_{\mathbb{T}}^*(e)&\equiv\left\{\q\in\mathbb{T}:\q^\intercal e - \hn_\E(\q;\widehat{\btheta})\geq\sup_{\tilde{\q}\in\mathbb{T}} \tilde{\q}^\intercal e - \hn_\E(\tilde{\q};\widehat{\btheta})-\kappa_n/\sqrt{n}\right\},\label{eq:QstarSet}
\end{align}
where $\mathbb{T}\in\{\bS^1,\bQ\}$ and $\kappa_n=o(\sqrt{n})$ is a sequence that diverges to infinity. 
We use Procedure~\ref{proc:bootstrap}-Step 1 to obtain a valid bootstrap-based approximation to the Gaussian process in Theorem \ref{thm:gaussian}. Denote $\mathbb{P}^*$ this bootstrap distribution conditional on $\{(\Y_i,\G_i,\X_i)\}_{i=1}^n$. Let
\begin{align}
    \small\hat{c}^\sPF_{1-\alpha}(e)\equiv
    \inf\Bigg\{c:
    \mathbb{P}^*\left(\left[\sup_{\q\in Q^*_{\bS^1}(e)}\mathbb{G}[-{\zeta_i^*(\cM\q;\btheta)}]\right]_+ + \left[ \sup_{\q\in Q^*_{\bQ}(e)}\mathbb{G}[-{\zeta_i^*(\cM\q;\btheta)}]\right]_->c\right)=\alpha\Bigg\}\label{eq:hat_c_PF}
\end{align} 
When $Var(\bLL_d|\X)$ is positive definite, it follows that if $\alpha\in(0,0.5)$, 
\begin{align*}
\limsup_{n\to\infty}\mathbb{P}\left(T_n^\sPF(e)>\hat{c}^\sPF_{1-\alpha}(e)\right)
\begin{cases}
    \le\alpha & \text{if}~e\in\sPF \\
    =1 & \text{if}~e\notin\sPF
\end{cases}
\end{align*}
\citep[Corollary 3.2]{kai16}. A confidence set that covers each $e\in\sPF$ with asymptotic probability at least equal to $(1-\alpha)$ can be obtained by test inversion.
\begin{remark}
    The result in Proposition~\ref{prop:PFlimit_distr} can be adapted to testing whether $e^*\in\sPF$ when $e^*$ needs to be estimated, by adjusting the covariance function of the limit Gaussian process in Theorem~\ref{thm:gaussian} and using the test statistic in Eq.~\eqref{eq:Tn_PF}. If the limit law in Eq.~\eqref{eq:Tn_PF_limit} is not guaranteed to be absolutely continuous on $\mathbb{R}_{++}$, one can use infinitesimal adjustments to the critical value to maintain asymptotic validity, as in \citet{kai16}.
\end{remark}

\subsection{Algorithms Yielding a Risk Allocation on the Frontier}
\label{subsec:algorithms:frontier}
An algorithm designer or a regulator may wonder if one can characterize the algorithm yielding a specific point on the FA-frontier $\sF$ or the Pareto frontier $\sPF$.
It turns out that, using our support function approach, the answer to this question is affirmative, and particularly simple for points in $\sPF$.

Suppose we have two data sets: one for training the algorithm and the other for evaluating it.
We denote the training sample as $\{(\tilde{\Y}_i,\tilde{\G}_i,\tilde{\X}_i)\}_{i=1}^{n_1}$ and the evaluation sample as $\{(\Y_j,\G_j,\X_j)\}_{j=1}^{n_2}$, with $\frac{n_1}{n_2}\to c$ for some positive constant $c$ and both samples drawn from the same distribution $\bP$.
Use the training data to estimate $\Delta\btheta$ via machine learning and $\cM$ by sample averages as described in Section~\ref{sec:estimation} and denote the resulting estimators $\widehat{\Delta\btheta}_{n_1}$ and $\hcM_{n_1}$, where the subscript $n_1$ indicates that it is based on the training sample.
Let 
\begin{align}
    \widehat{a}_{n_1}(\X_j;\q)=\mathds{1}\bigl\{k\bigl(\widehat{\btheta}_{n_1}(\X_j),\hcM_{n_1}\q\bigr)>0\bigr\},\label{eq:a_hat}
\end{align}
with $k\bigl(\widehat{\btheta}_{n_1}(\X_j),\hcM_{n_1}\q\bigr)$ as in Eq.~\eqref{eq:k_estimated}.
Note that $\widehat{a}_{n_1}(\X_j;\q)$ only takes as input the covariates, $\X_j$, and does not depend on group identity $\G_j$. Consequently, individuals with the same covariates are assigned the same treatment irrespective of their group. Nonetheless, information on group identity contained in the training data is used in the prediction model (to separately estimate $\Delta\theta^g$, $g\in\{r,b\})$, thereby offering a compromise between a utilitarian perspective as in \citet{man:mul:ven23} which advocates for group-aware decisions, and proponents of group-blind decision making.

{Then, for any $q\in\bS^1$, algorithm $\widehat{a}_{n_1}(\,\cdot\,;\q)$ leads to a loss for each observation $i$ in the evaluation sample equal to $\ell(0,\Y_j)+(\ell(1,\Y_j)-\ell(0,\Y_j))\cdot\mathds{1}\bigl\{k\bigl(\widehat{\btheta}_{n_1}(\X_j),\q\bigr)>0\bigr\}$.
Therefore, the average losses for the $r$ group and for the $b$ group in the evaluation sample equal
\begin{align}
\begin{bmatrix}
\frac{\sum_{j:\G_j=r}\ell(0,\Y_j)+(\ell(1,\Y_j)-\ell(0,\Y_j))\cdot\mathds{1}\bigl\{k\bigl(\widehat{\btheta}_{n_1}(\X_j),\q\bigr)>0\bigr\}}{\sum_{j=1}^{n_2}\mathds{1}\{\G_j=r\}}\\
\frac{\sum_{j:\G_j=b}\ell(0,\Y_j)+(\ell(1,\Y_j)-\ell(0,\Y_j))\cdot\mathds{1}\bigl\{k\bigl(\widehat{\btheta}_{n_1}(\X_j),\q\bigr)>0\bigr\}}{\sum_{j=1}^{n_2}\mathds{1}\{\G_j=b\}}
\end{bmatrix}
    =\widehat{\ss}_{\E,n_2}(\q;\widehat{\btheta}_{n_1})\label{eq:hat_ss_alg}
\end{align}
}
By the same argument as in the proof of Proposition~\ref{prop:PFconsistency}, it follows that 
\begin{align*}
    \max_{q\in\bQ}\left\Vert\widehat{\ss}_{\E,n_2}(\q;\widehat{\btheta}_{n_1})-\ss_\E(\q)\right\Vert \xrightarrow[]{p} 0~~~\text{as}~ n_1,n_2\to\infty.
\end{align*}
Hence, the algorithm in Eq.~\eqref{eq:a_hat} with $q\in\bQ$ gives consistent estimators of points on $\sPF$.

Algorithms that return consistent estimators of points on $\sF$ (other than those on $\sPF$) are harder to obtain, as one needs to determine which direction $\q$ corresponds to a point $e\in\sF$ and plug that direction in the algorithm in Eq.~\eqref{eq:a_hat}.
Suppose Eq.~\eqref{eq:no_kinks_cond} is satisfied and $\E$ has no kinks. 
Use the same DML construction 
in Section~\ref{sec:estimation} applied to the training sample to obtain an estimator of $\h_\E(\q)$, denoted $\hn_{\E,n_1}(\q;\widehat{\btheta}_{n_1})$, as in Eq.~\eqref{eqn:cross-fit-second-stage}, and an estimator of $\sF$, denoted $\widehat{\sF}_{n_1}$, as in Eq.~\eqref{eq:Fhat}.
{Select $\widehat{e}_{n_1}\in\widehat{\sF}_{n_1}$ through a projection method detailed in the proof of Proposition~\ref{prop:algorithm_consistency} so that $\Vert\widehat{e}_{n_1}-e\Vert=o_p(1)$ for some $e\in\sF$ and let}
\begin{align}
    \widehat{\q}^*_{n_1}(\widehat{e}_{n_1})=\arg\max_{\q\in\bS^1}\q^\intercal \widehat{e}_{n_1}-\hn_{\E,n_1}(\q;\widehat{\btheta}_{n_1}). \label{eq:hat_q_star}
\end{align}
Let $\widehat{\ss}_{\E,n_2}(\widehat{\q}^*_{n_1}(\widehat{e}_{n_1});\widehat{\btheta}_{n_1})$ be as in Eq.~\eqref{eq:hat_ss_alg} but with $\widehat{\q}^*_{n_1}(\widehat{e}_{n_1})$ replacing $q$.
The next proposition establishes that $\widehat{\ss}_{\E,n_2}(\widehat{\q}^*_{n_1}(\widehat{e}_{n_1});\widehat{\btheta}_{n_1})$ is a consistent estimator of $\ss_\E(\q^*_{\bS^1}(e))\in\sF$.
\begin{prop}\label{prop:algorithm_consistency}
\textit{Let the assumptions of Theorem~\ref{thm:gaussian} hold. Then, as $n_1,n_2\to\infty$,
\begin{align*}
    \left\Vert\widehat{\ss}_{\E,n_2}(\widehat{\q}^*_{n_1}(\widehat{e}_{n_1});\widehat{\btheta}_{n_1})-\ss_\E(\q^*_{\bS^1}(e))\right\Vert \xrightarrow[]{p} 0.
\end{align*}
}    
\end{prop}

Regardless of whether one aims at obtaining points in $\sPF$ or the entire $\sF$, the direction $\q$ can be interpreted as the vector of weights that the agent choosing the algorithm puts on each group's risk.
In other words, one may think of the agent as evaluating group risks according to the welfare loss function $U(e;\q)\equiv\q_1 e_r(a)+\q_2 e_b(a)$. 
For example, the more the agent cares about group $r$, the closer $\q$ is to $\mathfrak{u}_1$.
We note that $\widehat{a}_{n_1}(\X;\q)$ is an empirical success rule, and leave its statistical decision theory analysis to future research.

\section{Hypothesis Testing}
\label{sec:test}
\noindent In this Section we propose hypothesis tests to answer the following policy questions: 

(1) Should the policymaker consider banning group identity as an input to the algorithm? 

(2) Is there a less discriminatory alternative (LDA) to an existing algorithm?

\noindent We show how to express the first policy question in terms of restrictions on $\h_\E(\q)$, and we leverage Proposition \ref{prop:sh} to do the same for the second policy question.
We then establish asymptotic validity of the corresponding testing procedures.

\subsection{How to Test Whether Group Identity Should be Banned}
\label{sec:test-group-balance}
\citetalias{lia:lu:mu:oku24} (Proposition 6) show that using $\X$ only instead of $(\X,\G)$ as algorithmic input uniformly worsens the frontier if $\R$ and $\B$, obtained when only $\X$ is used as input, are strictly separated by the 45-degree line.\footnote{“Uniformly worsening the frontier” means that, under the preference relations defined in Eq.~\eqref{pref}, every point on the frontier $\sF\big(\bP,\mathcal{A}(\sX)\big)$ is dominated by a point on the frontier $\sF\big(\bP,\mathcal{A}(\sX\times\{r,b\})\big)$.} We therefore aim at testing the null hypothesis that $\R$ and $\B$ lie weakly on the same side of the 45 degree line, i.e., that the difference in the two coordinates of $\R$ has the same sign as that of $\B$, against the alternative that they are strictly separated by the 45-degree line:
\begin{align}
\HH_0&: \bigg((\mathfrak{u_1}-\mathfrak{u_2})^\intercal \R\bigg)\bigg((\mathfrak{u_1}-\mathfrak{u_2})^\intercal \B\bigg)\geq0,\label{eq:null-weak-GB}\\
\HH_A&: \bigg((\mathfrak{u_1}-\mathfrak{u_2})^\intercal \R\bigg)\bigg((\mathfrak{u_1}-\mathfrak{u_2})^\intercal \B\bigg)<0.\notag
\end{align}
If the null in Eq.~(\ref{eq:null-weak-GB}) is rejected, the policymaker should not ban group identity $\G$ as algorithm's input. A Type-I error amounts to the case where one concludes that there is strict group-balance, while instead weak group-skew holds. 
As a consequence, one does not ban $G$ as input to the algorithm, thinking that banning $G$ is uniformly welfare-reducing, while instead depending on the preferences of the designer it might not be the case.
Another interpretation of this test amounts to determining, based on whether $\HH_0$ is rejected or not, if one can justify implementing algorithms that lead to Pareto-dominated risks based on the designer's preference over fairness and accuracy. When $\HH_0$ holds true, this justification is possible, as the frontier includes an upward-sloping segment (e.g., Panels (b)-(c) of Figure \ref{fig1}). On the other hand, when $\HH_A$ holds true, such justification is untenable, as the frontier coincides with the Pareto frontier (e.g., Panel (a) of Figure \ref{fig1}).
Hence, a Type-I error can also be interpreted as a case where one concludes that only Pareto-optimal risks should be implemented, while fairness considerations may justify Pareto-dominated risks.

{As discussed in Remark~\ref{remark:est-supp-set}, carrying out inference if we use DML to directly estimate both coordinates of the points $\R$ and $\B$ is difficult.} Yet, using Eqs.~\eqref{eq:define_SupportSet} and \eqref{eq:coord}, we can represent these points through moment equalities and inequalities that involve $\h_\E(\q)$ only:
\begin{align}
\R: \Bigg\{\begin{array}{lr}
        \h_\E(\mathfrak{u}_1)-\mathfrak{u}_1^\intercal \R=0,\\
        \h_\E(\q)\,-\,\q^\intercal \R~\geq 0, ~\forall \q\in\bS^1,
        \\
        \end{array} 
        ~~~~~~~~\,\B: \Bigg\{\begin{array}{lr}
        \h_\E(\mathfrak{u}_2)-\mathfrak{u}_2^\intercal \B=0,\\
        \h_\E(\q)\,-\,\q^\intercal \B~\geq 0, ~\forall \q\in\bS^1,
        \\
        \end{array}\label{eq:moment-ineq}
\end{align}
where in Eq.~\eqref{eq:moment-ineq}, the equality constraint for $\R$ restricts it to have horizontal coordinate equal to that of $\ss_\E(\mathfrak{u}_1)$, as per Eq.~\eqref{eq:define_SupportSet}, and the continuum of inequality constraints indexed by $\q\in\bS^1$ restricts $\R$ to be an element of $\E$. The moment constraints that define $\B$ are interpreted similarly. Because the support set $\ss_\E(\cdot)$ in any direction is a singleton by Proposition \ref{prop:ss}, the moments in Eq.~\eqref{eq:moment-ineq} yield points $\R$ and $\B$ that coincide with Eq.~\eqref{eq:coord}. 

We test the null in Eq.~\eqref{eq:null-weak-GB} at a given significance level $\alpha\in(0,1)$ based on our procedure to test, for given $e\in\mathbb{R}^2$, whether $e\in\sPF$:
\vspace{-.25cm}
\begin{proc}[Testing Weak Group-Skew]
\label{proc:test-GB}
\textit{
\begin{enumerate}
    \item Build a $(1-\alpha)$-level confidence set for $(\R, \B)$ by
    \begin{align}
        \mathcal{CS}_n(\R,\B)\equiv\big\{(\Tilde{\R},\Tilde{\B})\in\bB_C\times\bB_C: T_n(\Tilde{\R},\Tilde{\B})\leq \hat{c}_{1-\alpha}(\Tilde{\R},\Tilde{\B})\big\},\label{eq:CS-R}
    \end{align}
    where for the support function estimator $\widehat{\h}_{\E}(\cdot\,; \widehat{\btheta})$ in Theorem \ref{thm:gaussian}, $T_n(\Tilde{\R},\Tilde{\B})$ adapts the test statistic in Eq.~\eqref{eq:Tn_PF}:
    \begin{align*}
        T_n(\Tilde{\R},\Tilde{\B})\equiv\hspace{-.6cm}\sum_{\substack{(e, \mathfrak{u}_j)\\\in\{(\Tilde{\R}, \mathfrak{u}_1),(\Tilde{\B}, \mathfrak{u}_2) \}}}\hspace{-.8cm}\sqrt{n}\left(\left[\max_{\q\in\bS^1}\q^\intercal e-\widehat{\h}_\E(\q;\widehat{\btheta})\right]_{+} + \left[\mathfrak{u}_j^\intercal e-\widehat{\h}_\E(\mathfrak{u}_j;\widehat{\btheta})\right]_-\right).
    \end{align*}
    The critical value $\hat{c}_{1-\alpha}(\Tilde{\R},\Tilde{\B})$ is obtained similarly to $\hat{c}^\sPF_{1-\alpha}(e)$ in Eq.~\eqref{eq:hat_c_PF}.
    \item Reject $\HH_0$ in Eq. \eqref{eq:null-weak-GB} if
    \vspace{-.3cm}
    \begin{align}
        \varphi_n^{\texttt{skew}}\equiv\mathds{1}\left\{\sup_{(\Tilde{\R}, \Tilde{\B})\in \mathcal{CS}_n(\R,\B)}\bigg((\mathfrak{u_1}-\mathfrak{u_2})^\intercal \Tilde{\R}\bigg)\bigg((\mathfrak{u_1}-\mathfrak{u_2})^\intercal \Tilde{\B}\bigg)<0\right\}=1.\label{eq:rule-weak-GB}
    \end{align}
\end{enumerate}
}
\end{proc}

We note that the test in Eq.~\eqref{eq:rule-weak-GB} may be conservative as it is based on projection.
\begin{prop}
\label{prop:weak-GB}
    \textit{Let the assumptions in Theorem \ref{thm:gaussian} hold. Then
    \vspace{-.3cm}
    \begin{align}
        \limsup_{n\to\infty}\mathbb{E}\left[\varphi_n^{\texttt{skew}}\right]\leq \alpha.\label{eq:testRB-validity}
    \end{align} 
}

    \vspace{-.3cm}
\end{prop}

\subsection{How to Test for the Existence of an LDA}\label{sec:test-LDA}
Given an algorithm $a^*\in\mathcal{A}(\sX)$ that induces the risk pair $e^*=(e_r^*, e_b^*)\in\E$, call another algorithm that yields a feasible risk pair $e=(e_r,e_b)\in\E$ an LDA if it is at least as accurate as $a^*$ for both groups and at least as fair, with one of these inequalities strict. It follows from the characterization in Proposition~\ref{prop:sh} and Eq.~\eqref{eq:Cstar} that no LDA to $a^*$ exists if and only if $\E$ can be properly separated from
\begin{align*}
\cC(e^*)\equiv\left\{e\in\mathbb{R}^2: e_r\leq e_r^*, e_b\leq e_b^*, |e_r-e_b|\leq |e_r^*-e_b^*|\right\}.
\end{align*}
Recall that the closed form expression for the support function of $\sC\equiv\cC(e^*)$ in direction $\q=[\q_1,~\q_2]^\intercal\in\tilde{\bS}^1$ is given in Eq.~\eqref{eqn:sf-C}.
We then test the null hypothesis
\begin{align}
\HH_0: \max_{\q\in\tilde{\bS}^1}(-\h_{\cC^*}(\q)-\h_\E(-\q))=0\label{eqn:null-LDA}
\end{align}
against the alternative that $\HH_0$ is false. Rejecting the null in Eq.~\eqref{eqn:null-LDA} means that $e^*\notin\sF$ and there exists an LDA. We propose estimating $\h_{\sC}(\q)$ for $\q\in\tilde{\bS}^1$ by
\begin{align*}
\hn_{\sC}(\q)&=\max\biggl\{\q_1\min\{\widehat{e}_r^*,2\widehat{e}_b^*-\widehat{e}_r^*\}+\q_2\widehat{e}_b^*, \,\,\,\q_1 \widehat{e}_r^* + \q_2\min\{\widehat{e}_b^*, 2\widehat{e}_r^*-\widehat{e}_b^*\}\biggr \},
\end{align*}
where for $g\in\{r,b\}$, $e_g^*$ is estimated by sample means,
\begin{align}
\widehat{e}_g^*=\frac{1}{n}\sum_{i=1}^n \frac{Z_i^g}{\widehat{\mu}_g}, \quad \text{where }Z_i^g\equiv\mathds{1}\{\G_i=g\}\bigl(a^*(\X_i)\ell(1,\Y_i)+(1-a^*(\X_i))\ell(0,\Y_i)\bigr)\label{ehat}
\end{align}
and $\widehat{\mu}_g\equiv\frac{1}{n}\sum_{i=1}\mathds{1}\{G_i=g\}$. We propose the following test statistic:
\begin{align}
T_n^{\texttt{LDA}}\equiv\sqrt{n}\left(\left[\max_{\q\in\bS^1}(\q^\intercal \widehat{e}^*-\widehat{\h}_\E(\q;\widehat{\btheta}))\right]_{+} +\left[\max_{q\in\tilde{\bS}^1}\left(-\hn_{\sC}(\q)-\hn_\E(-\q)\right)\right]_-\right).\label{eq:Tn_LDA}
\end{align}
$T_n^{\texttt{LDA}}$ differs from $T_n^\sF$ in Eq.~\eqref{eq:Tn_sF} only in that $\widehat{e}^*$ is estimated in the former.

\begin{prop}
\label{prop:lda}
\textit{
Let the assumptions in Theorem \ref{thm:gaussian} hold. Then, for any pre-specified significance level $\alpha\in(0,1)$, the test below has asymptotically correct size control:
$$
\text{Reject the null in Eq.~\eqref{eqn:null-LDA} if} \quad T_n^{\texttt{LDA}} > c_{1-\alpha+\varsigma}^{\texttt{LDA}}+\varsigma,
$$
where $\varsigma>0$ is an arbitrarily small positive constant and for any $\beta\in(0,1)$ the critical value $c_\beta^{\texttt{LDA}}$ is the $\beta$-quantile of $\psi^{\texttt{LDA}}$, for $\psi^{\texttt{LDA}}$ a random variable defined in Eq.~\eqref{psi-inf}.
}
\end{prop}
When Eq.~\eqref{eq:no_kinks_cond} holds and $Var(\bLL_d|\X)$ is positive definite for $d\in\{0,1\},~\X$-a.s., one can take $\varsigma=0$ in Proposition~\ref{prop:lda} and the expression for $\psi^{\texttt{LDA}}$ simplifies to that in Eq.~\eqref{psi-inf_no_kink}. 
When kinks might be present and the limit distribution is not guaranteed to be continuous and strictly increasing, we take $\varsigma>0$ as in \citet{and:shi13}.
The derivation of $\psi^{\texttt{LDA}}$ uses the fact that $T_n^{\texttt{LDA}}$ is formed by compositions of the $\max$ and $\min$ functions in Eq.~\eqref{eqn:sf-C} and the $\max$ function in Eq.~\eqref{eqn:null-LDA}---all of which are Hadamard directionally differentiable, as shown in \citet{fan:san19} and \citet{car:cue:alb20}; since compositions preserve directional differentiability \citep[Proposition 3.6]{Shapiro90}, an extension to the functional Delta method can be applied to Theorem \ref{thm:gaussian}. However, standard bootstraps are inconsistent due to the lack of full differentiability \citep[see Section 3.2 of][]{fan:san19}.  As such, we leverage results in \citet{fan:san19} to approximate the distribution of $\psi^{\texttt{LDA}}$ and its quantiles via a modified multiplier bootstrap procedure detailed below, similar to that in \citet{sem23}, where we denote $\phi$ any generic Hadamard directionally differentiable function, $\widehat{he^*}=\widehat{he^*}(\widehat{\btheta})\equiv[\hn_\E(\q;\widehat{\btheta}),~\widehat{e}_r^*,~\widehat{e}_b^*]^\intercal$ the vector of estimators, and $he^*\equiv[\h_\E(\q),~e_r^*,~e_b^*]^\intercal$ the vector of truths.

\begin{proc}[Bootstrap for the Quantiles of $\sqrt{n}\{\phi\big(\widehat{he^*}\big)-\phi\big(he^*\big)\}$]
\label{proc:bootstrap}

    \begin{enumerate}
    \textit{   \item Draw $\{W_i\}_{i=1}^n$ i.i.d. from the exponential distribution with mean $1$ independent of the sample $\{(\Y_i,\G_i,\X_i)\}_{i=1}^n$ and construct the bootstrap analogue of $\widehat{he^*}$:
        \vspace{-.1cm}
        \begin{align}
        \widetilde{he^*}=\widetilde{he^*}(\widehat{\btheta})\equiv[\widetilde{\h}_\E(\q;\widehat{\btheta}),~\widetilde{e}_r^*,~\widetilde{e}_b^*]^\intercal,\label{eq:bs-anolog-he}
        \end{align}
}

    \vspace{-.1cm}
\noindent\textit{where $\widetilde{\h}_\E(\q;\widehat{\btheta})\equiv\frac{1}{n}\sum_{i=1}^n\frac{W_i}{\overline{W}}\zeta_i(\tcM\q;\widehat{\btheta})$  for $\overline{W}\equiv\frac{1}{n}\sum_{i=1}^n W_i$, $\tcM\equiv\diag(1/\widetilde{\mu}_r, 1/\widetilde{\mu}_b)$, $\widetilde{\mu}_g\equiv\frac{1}{n}\sum_{i=1}^n\frac{W_i}{\overline{W}}\mathds{1}\{\G_i=g\}$, and $\widetilde{e}_g^*\equiv\frac{1}{n}\sum_{i=1}^n \frac{W_i}{\overline{W}}\frac{Z_i^g}{\widetilde{\mu}_g}$.
    \item Numerically approximate $\phi'_{he^*}(\cdot)$, the directional derivative of $\phi(\cdot)$ at $he^*$, by
    \begin{align*}
        \widehat{\phi'}_{he^*}(	\Ddot{he})=\frac{1}{s_n}\left(\phi\left(\widehat{he^*}+s_n(\Ddot{he})\right)-\phi\left(\widehat{he^*}\right)\right),
    \end{align*}
    where $\Ddot{he}\in\ell^\infty(\bS^1)\times\mathbb{R}^2$ is a candidate direction at which we evaluate $\phi'_{he^*}(\cdot)$ and $s_n$ is a vanishing sequence of step sizes such that $\sqrt{n}s_n\to\infty$.
    \item Obtain $\widehat{\phi'}_{he^*}\big(\sqrt{n}\{\widetilde{he^*}-\widehat{he^*}\}\big)$ and  
    \begin{align*}  \widehat{c}_\beta\equiv\inf\bigg\{c:\bP\bigg(\left.\widehat{\phi'}_{he^*}\big(\sqrt{n}\{\widetilde{he^*}-\widehat{he^*}\}\big)\leq c \,\,\right|\, \{(\Y_i,\G_i,\X_i)\}_{i=1}^n\bigg)\geq \beta\bigg\},
\end{align*}
as estimators, respectively, for the limit distribution of $\sqrt{n}\{\phi\big(\widehat{he^*}\big)-\phi\big(he^*\big)\}$ and its $\beta$-quantile, denoted as $c_\beta$.}
    \end{enumerate}
\end{proc}
\noindent The consistency of the bootstrap outlined in Procedure \ref{proc:bootstrap} is stated in the following result.

\begin{prop}
    \label{prop:bootstrap}
\textit{Under the assumptions of Theorem \ref{thm:gaussian},
\begin{align*}
\sup_{f\in\mathcal{BL}_1}\left|\mathbb{E}\big[f\big(\widehat{\phi'}_{he^*}\big(\sqrt{n}\{\widetilde{he^*}-\widehat{he^*}\}\big)\big)\,\big|\,\{(\Y_i,\G_i,\X_i)\}_{i=1}^n\big]-\mathbb{E}\big[f\big(\phi_{he^*}'(\mathbb{G}_{he^*})\big)\big]\right|=o_p(1),
\end{align*}
where $\mathcal{BL}_1$ is the set of $1$-Lipschitz functions $f:\mathbb{R}\to\mathbb{R}$ such that $|f|_\infty\leq1$ and $\mathbb{G}_{he^*}$ is the Gaussian limit process of $\sqrt{n}(\widehat{he^*}-he^*)$ given in Eq. \eqref{eq:e-h-joint} in Appendix \ref{appn:A}. If the cdf of $\phi_{he^*}'(\mathbb{G}_{he^*})$ is continuous and
increasing at its $\beta$-quantile, denoted $c_\beta$, then $\widehat{c}_\beta=c_\beta+o_p(1)$.
}
\end{prop}

The proof of Proposition \ref{prop:lda} shows that $\psi^{\texttt{LDA}}=\phi_{he^*}'(\mathbb{G}_{he^*})$ for a particular $\phi_{he^*}'(\cdot)$ that is the composition of the directional derivatives of the $\min$, $\max$, and $\inf$ functions that constitute $T_n^{\texttt{LDA}}$. The expression of $\phi_{he^*}'(\cdot)$, given in Eq. \eqref{psi-inf}, is complex and hence we recommend the numerical approximation approach in Step 2 of Procedure \ref{proc:bootstrap}. Under Proposition \ref{prop:bootstrap}, we can estimate $c_\beta^{\texttt{LDA}}$, the $\beta$-quantile of $\psi^{\texttt{LDA}}$, by going through the steps in Procedure \ref{proc:bootstrap}, where we replace $\phi(\cdot)$ by the composition of the above mentioned directionally differentiable functions accordingly. The same bootstrap procedure can be used to consistently estimate the critical values put forward to carry out inference in Section~\ref{sec:distance_F}. The use of the infinitesimal constant $\varsigma$ in Proposition~\ref{prop:lda} and Procedure~\ref{proc:test-F} accounts for the possibility that the cdf of $\phi_{he^*}'(\mathbb{G}_{he^*})$ may not be continuous and increasing at $c_{1-\alpha}$.

\subsubsection{Algorithms Yielding an LDA}
\label{subsec:algorithm:LDA}
{
\begin{figure}\captionsetup[subfigure]{font=footnotesize}
\centering
\includegraphics[width=0.75\linewidth]{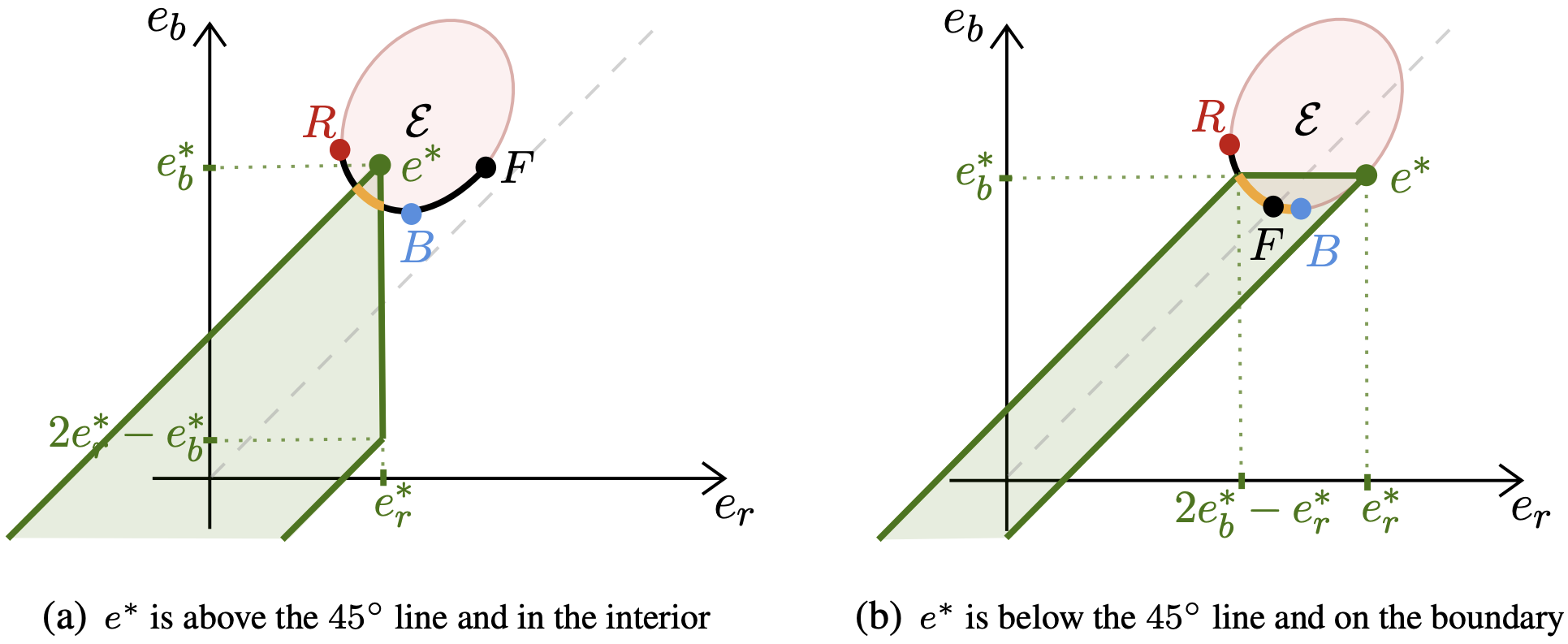}
\caption{\footnotesize{The set $\sF^*$, marked in orange, is the portion of the FA-frontier yielding risk pairs preferred to $e^*$.}
}
\label{fig:frontier-star}
\end{figure}
}
An algorithm designer or a regulator may wonder if one can characterize the algorithms yielding LDAs to a given algorithm $e^*$.
It turns out that it is possible to do so, by combining the algorithm that we put forward in Eq.~\eqref{eq:a_hat} with a careful use of our characterization of $\sF$ in Proposition~\ref{prop:sh}.
We illustrate the idea for the case that Eq.~\eqref{eq:no_kinks_cond} holds and $\E$ has no kinks.
For a given risk pair $e^*$ induced by an algorithm $a^*\in\mathcal{A}$ such that $e^*\notin\sF$, by definition the set of risk pairs $e\in\sF$ such that $e>_{FA}e^*$ are:
\begin{align}
    \sF^*\equiv\sF\cap\cC(e^*)
    =\Bigg\{e\in\E:~&\left[\max_{\q\in\tilde{\bS}^1}(-\h_{\cC(e)}(\q)-\h_\E(-\q))\right]_-=0,\notag\\
    &\left[\max_{\q\in\{\mathfrak{u}_1-\mathfrak{u}_2,\mathfrak{u}_2-\mathfrak{u}_1,-\mathfrak{u}_1,-\mathfrak{u}_2\}}\left(\q^\intercal e-\h_{\cC(e^*)}(\q)\right)\right]_+=0 \Bigg\}.\label{eq:F_star}
\end{align}
The set $\sF^*$ is depicted in Figure~\ref{fig:frontier-star} as the orange portion of $\sF$ that intersects with $\cC(e^*)$.
Using the same notation as in Section~\ref{subsec:algorithms:frontier}, we can use a training sample $\{(\tilde{\Y}_i,\tilde{\G}_i,\tilde{\X}_i)\}_{i=1}^{n_1}$ and a construction that mimics the procedure to build the estimator of $\sF$ in Eq.~\eqref{eq:Fhat} to obtain a consistent estimator $\widehat{\sF}^*_{n_1}$ of $\sF^*$ (the consistency of this estimator can be established through the same steps as in the proof of Proposition~\ref{prop:consistency_coverage_sF}).
{Select $\widehat{e}_{n_1}\in\widehat{\sF}_{n_1}$ through a projection method detailed in the proof of Proposition~\ref{prop:algorithm_consistency} and denote by $e\in\sF^*$ the point to which it converges.}
Let $\widehat{\q}^*_{n_1}(\widehat{e}_{n_1})$ be the consistent estimator of $\q^*_{\bS^1}(e)$ defined in Eq.~\eqref{eq:hat_q_star}.
Let $\widehat{\ss}_{\E,n_2}(\widehat{\q}^*_{n_1}(\widehat{e}_{n_1});\widehat{\btheta}_{n_1})$ be defined as in Eq.~\eqref{eq:hat_ss_alg} but with $\widehat{\q}^*_{n_1}(\widehat{e}_{n_1})$ replacing $q$. Then $\widehat{\ss}_{\E,n_2}(\widehat{\q}^*_{n_1}(\widehat{e}_{n_1});\widehat{\btheta}_{n_1})$ is a consistent estimator of $\ss_\E(\q^*_{\bS^1}(e))\in\sF^*$, by the same argument used to establish Proposition~\ref{prop:algorithm_consistency}.

\section{Distance to the Fairest Point}\label{sec:distance_F}
In this section, we propose a method to build a confidence interval for the distance between the risk $e^*$ induced by a given algorithm and the fairest point $\F$ on the frontier, denoted $\rho(e^*,\F)$, where $\rho$ is a Hadamard directionally differentiable distance function (e.g., the Euclidean distance, Manhattan distance, Chebyshev distance, etc.). 
This can inform a decision maker of the relative merits in promoting equity and achieving business efficiency of different algorithms by comparing the confidence intervals on their distance to $\F$. 

Recall $\mathcal{H}_{45}\equiv\{e\in\mathbb{R}^2:e_r=e_b\}$ denotes the 45-degree line; let $\mathcal{H}_{45}^+\equiv\{e\in\mathbb{R}^2:e_r<e_b\}$ and $\mathcal{H}_{45}^-\equiv\{e\in\mathbb{R}^2:e_r>e_b\}$ denote, respectively, the open halfspace above and below the 45-degree line. As shown in Section \ref{sec:fairest}, the coordinates of $\F$ depend on whether $\E$ intersects with $\mathcal{H}_{45}$.
When $\Tilde{\E}\equiv\E\cap\mathcal{H}_{45}\neq\emptyset$, as shown in Eqs.~\eqref{eq:F_intersect}-\eqref{eq:h_tildeE} we have 
$\F=\mathfrak{u}\cdot\h_{\Tilde{\E}}\left(\mathfrak{u}_1\right)$, with $\h_{\Tilde{\E}}\left(\mathfrak{u}_1\right)=\inf_{c\in\mathbb{R}}\h_\E\left(\mathfrak{u}_1(c)\right)$, $\mathfrak{u}_1(c)\equiv\mathfrak{u}_1-c[1~~-1]^\intercal$, and $\mathfrak{u}\equiv(\mathfrak{u}_1+\mathfrak{u}_2)$. Note that $\h_{\Tilde{\E}}\left(\cdot\right)$ is a Hadamard directionally differentiable function of $\h_\E(\cdot)$, and its composition with the distance function $\rho$ is again directionally differentiable. Hence, we use our DML estimator, Theorem \ref{thm:gaussian}, and the results in \citet{fan:san19} to directly obtain the limit distribution of an estimator for $\rho(e^*, \F)$. 
However, when $\E\subset \mathcal{H}_{45}^+$ (respectively, $\E\subset \mathcal{H}_{45}^-$), $\F$ is given by the support set of $\E$ in direction $(\mathfrak{u}_2-\mathfrak{u}_1)$ (respectively, $(\mathfrak{u}_1-\mathfrak{u}_2)$); see Eqs.~\eqref{eq:F_above}-\eqref{eq:F_below}. {As discussed in Remark~\ref{remark:est-supp-set}, carrying out inference if we use DML to directly estimate both coordinates of a support point is difficult,} and therefore we use Eq.~\eqref{eq:define_SupportSet} to represent $\F$ through moments that involve $\h_\E(\cdot)$ only; see Eqs.~\eqref{eq:MMI-F-above}-\eqref{eq:MMI-F-below} below. 

Observe that $\E\subset \mathcal{H}_{45}^+$ if and only if $\F \in \mathcal{H}_{45}^+$, and $\E\subset \mathcal{H}_{45}^-$ if and only if $\F \in \mathcal{H}_{45}^-$, as illustrated in Panels (b) and (d) of Figure \ref{fig:F}. 
Hence, we partition the parameter space $\bB_C$, to which $\F$ belongs Assumption~\ref{asm:moments}, 
into three sets: $\bB_C^+\equiv\bB_C\cap\mathcal{H}_{45}^+$, $\bB_C^-\equiv\bB_C\cap\mathcal{H}_{45}^-$, and $\bB_C^{45}\equiv\bB_C\cap\mathcal{H}_{45}$.
We then have the following expressions for $\rho(e^*, \F)$:
\begin{enumerate}
\item If $\F\in\bB_C^{45}$, $\rho(e^*, \F)=\rho(e^*,\mathfrak{u}\cdot\h_{\Tilde{\E}}\left(\mathfrak{u}_1\right))$.
\item If $\F\in\bB_C^+$, $\rho(e^*, \F)=\rho(\Tilde{e}, \Tilde{\F})$ for $(\Tilde{e}, \Tilde{\F})$ satisfying:  
\vspace{-.25cm}
    \begin{align}
\left\{\begin{array}{lr}
        \Tilde{e}-e^*=0,\\
        \h_{\E}\big((\mathfrak{u}_2-\mathfrak{u}_1)/\sqrt{2}\big)-\big((\mathfrak{u}_2-\mathfrak{u}_1)/\sqrt{2}\big)^\intercal\Tilde{\F}=0,\\
        \h_\E(\q)\,-\,\q^\intercal \Tilde{\F}~\geq 0, ~\forall \q\in\bS^1.
        \end{array}\right.\label{eq:MMI-F-above}
\end{align}
\item If $\F\in\bB_C^-$, $\rho(e^*, \F)=\rho(\Tilde{e}, \Tilde{\F})$ for $(\Tilde{e}, \Tilde{\F})$ satisfying:
\vspace{-.25cm}
\begin{align}
\left\{\begin{array}{lr} 
          \Tilde{e}-e^*=0,\\
          \h_{\E}\big((\mathfrak{u}_1-\mathfrak{u}_2)/\sqrt{2}\big)-\big((\mathfrak{u}_1-\mathfrak{u}_2)/\sqrt{2}\big)^\intercal\Tilde{\F}=0,\\
        \h_\E(\q)\,-\,\q^\intercal \Tilde{\F}~\geq 0, ~\forall \q\in\bS^1. 
        \end{array}\right.\label{eq:MMI-F-below}
\end{align}
\end{enumerate}
In each of Eqs. \eqref{eq:MMI-F-above}-\eqref{eq:MMI-F-below}, the first condition pins down $\Tilde{e}$ to equal $e^*$; the second and third conditions restrict, respectively, $\tilde{\F}$ to be a point on the supporting hyperplane of $\E$ in the appropriate direction and $\tilde{\F}\in\E$, which together restricts $\tilde{\F}$ to be the support set of $\E$ in this direction.
We propose the following testing procedure:
\begin{proc}[Confidence Interval for  $\rho(e^*, \F)$]
\label{proc:test-F}
\textit{
    \begin{enumerate}
    \item Construct estimators $\widehat{e}^*$ as in Eq.~\eqref{ehat} and  $\widehat{\h}_\E\big(\cdot;\widehat{\btheta}\big)$ as in Theorem \ref{thm:gaussian}.
    \item Emulate the construction in Step 1 of Procedure \ref{proc:test-GB} to obtain two $(1-\alpha)$-level confidence sets for $(e^*, \F)$: let  $\mathcal{CS}_n^+(e^*, \F)$ denote the one for the moments in Eq.~\eqref{eq:MMI-F-above} (in the analog of Eq.~\eqref{eq:CS-R}, replace $\bB_C$ with $\cl{\bB_C^+}$), and $\mathcal{CS}_n^-(e^*, \F)$ the one for the moments in Eq.~\eqref{eq:MMI-F-below} (in the analog of Eq.~\eqref{eq:CS-R}, replace $\bB_C$ with $\cl{\bB_C^-}$).\footnote{For a given set $\bB$, we denote by $\cl{\bB}$ its closure.}
    \item For a given $(\Tilde{e}, \Tilde{\F})\in\bB_C\times\bB_C^{45}$ and $\widehat{\h}_{\Tilde{\E}}\left(\mathfrak{u}_1;
    \widehat{\btheta}\right)\equiv\inf_{c\in\mathbb{R}}\widehat{\h}_\E\big(\mathfrak{u}_1(c);
    \widehat{\btheta}\big)$, let:
    \begin{align}
    T_n^{45}(\rho(\Tilde{e},\Tilde{\F}))\equiv\sqrt{n}\left|\rho\bigg(\widehat{e}^*,\mathfrak{u}\cdot\widehat{\h}_{\Tilde{\E}}(\mathfrak{u}_1;
    \widehat{\btheta})\bigg)-\rho(\Tilde{e}, \Tilde{\F})\right|.\label{eq:test-stat-delta}
    \end{align}
    Let $\psi^{45}$ denote the random variable to which $T_n^{45}(\rho(\Tilde{e},\Tilde{\F}))$ converges in distribution for $\Tilde{e}=e^*$ and $\Tilde{\F}=\F_{45}\equiv\mathfrak{u}\cdot\h_{\Tilde{\E}}\left(\mathfrak{u}_1\right)$.
    Let $c_\beta^{45}$ denote the $\beta$-quantile of $\psi^{45}$ and $\varsigma>0$ an infinitesimal uniformity factor.  Use test inversion to construct the confidence set:
    \begin{align*}
        \mathcal{CS}_n^{45}(\rho(e^*, \F))=\left\{\rho(\tilde{e},\tilde{\F}): (\Tilde{e}, \Tilde{\F})\in\bB_C\times\bB_C^{45}, T_n^{45} \le c_{1-\alpha+\varsigma}^{45}+\varsigma\right\}.
    \end{align*}
    The expression for $\psi^{45}$ is complex and given in Eq.~\eqref{eq:psi-delta}, with a simpler expression provided in Eq.~\eqref{eq:psi-delta_no_kink} for the case that Eq.~\eqref{eq:no_kinks_cond} is satisfied and $\E$ has no kinks.
\item Obtain a confidence interval for $\rho(e^*,\F)$ as $$
\mathcal{CS}_n^{\rho(e^*,\F)}\equiv\left\{\rho(\tilde{e}, \tilde{\F}): (\tilde{e}, \tilde{\F})\in\mathcal{CS}_n^+(e^*,\F)\bigcup\mathcal{CS}_n^-(e^*,\F)\right\}\bigcup\left\{\mathcal{CS}_n^{45}(\rho(e^*,\F))\right\}.
$$
\end{enumerate}}
\end{proc}
Intuitively, this construction inverts a test that jointly assesses the location of $\E$ relative to $\mathcal{H}_{45}$ \textit{and} the value of $\rho(e^*,\F)$. For example, if $\F\in\mathcal{H}_{45}^-$, then both $\mathcal{CS}_n^+(e^*,\F)$ and $\mathcal{CS}_n^{45}(\rho(e^*,\F))$ are empty with probability approaching one (recall from Section \ref{sec:fairest} that $\inf_{c\in\mathbb{R}}\h_\E\left(\mathfrak{u}_1(c)\right)$ is unbounded when $\F\notin\mathcal{H}^{45}$).
Our next result shows that Procedure \ref{proc:test-F} delivers an asymptotically valid confidence interval. 
\begin{prop}
    \label{prop:dist-F}
    \textit{Let $\rho$ be a Hadamard directionally differentiable distance function and the assumptions in Theorem~\ref{thm:gaussian} hold.
    Then the confidence interval constructed following Procedure \ref{proc:test-F} asymptotically covers the true $\rho(e^*,\F)$ with probability at least $1-\alpha$.}
\end{prop}
As we show in the proof, $\psi^{45}$ is again a composition of Hadamard directionally differentiable functions that define $T_n^{\rho(\Tilde{e}, \Tilde{\F})}$ in Eq.~\eqref{eq:test-stat-delta}. We can therefore employ the same bootstrap method detailed in Procedure \ref{proc:bootstrap} to consistently estimate the quantiles $c_\beta^{45}$ for $\beta\in(0,1)$, where we replace $\phi(\cdot)$ by the composition of the directionally differentiable functions, including $\inf$ and $\rho$, that define $T_n^{45}$, whose exact expression we relegate to Appendix \ref{appn:A}.

\section{Monte Carlo Experiments and Empirical Illustration}\label{sec:MC_and_empirical}
\subsection{Monte Carlo Simulations}\label{subsec:MC}
We evaluate the finite sample properties of the tests introduced in Sections \ref{sec:test}-\ref{sec:distance_F} using two distinct data generating processes (DGPs).
For both DGPs, the covariates $\X\equiv[\X_1,\dots,\X_{20}]^\intercal\in\mathbb{R}^{20}$ and group identity $\G$ are drawn from the following distributions:
\begin{align*}
    \X_2 \stackrel{d}{\sim} Unif(0,1),~\X_3\stackrel{d}{\sim}Beta(2,2),~G \stackrel{d}{\sim} Bern\left(0.6\right);\\
    \text{ for } j \in \{1,4,...,20\},~\X_j \stackrel{d}{\sim} \mathcal{N}(0,1) \text{ truncated to } [-3,3]. 
\end{align*}
We consider two different ways of generating the outcome $\Y$:
\begin{enumerate}
    \item Group-balanced DGP: $\Y \,|\, \G, \X\,\,\stackrel{d}{\sim} Bern\left(\frac{\G}{1+e^{-(\X_1+\X_2+0.5\X_3)}}+\frac{(1-\G)}{1+e^{-(-\X_1-0.5\X_2+\X_4)}}\right)$;
    \item $r$-skewed DGP: $\Y \,|\, \G, \X\,\,\stackrel{d}{\sim} Bern\left(\frac{\G}{1+e^{-2(\X_1+\X_2+\X_3)}}+\frac{(1-\G)}{1+e^{-0.7(\X_1+0.5\X_2+0.6\X_4)}}\right)$,
\end{enumerate}
where the group-balanced DGP is such that $\X$ is informative about $\Y$ in opposite directions for group $r$ ($\G=1$) and group $b$ ($\G=0$), but its predictive power is similar across groups. In contrast, the $r$-skewed DGP is such that $\X$ is systematically more informative about the group-$r$ outcome. We take the loss function to be the classification error, $\ell(d,y)=\mathds{1}\{d\neq y\}$. 
We also construct a status quo algorithm $a^*$, that we fix in the simulations, by training once a logistic regression on a sample of size $10,000$, with $5,000$ observations from the balanced DGP and $5,000$ observations from the $r$-skewed DGP.\footnote{We train the logistic regression on a mixture of group-balanced and $r$-skewed data so that we test for existence of an LDA to the same algorithm $a^*$ in both DGPs. DGP-specific logistic regressions trained on $10,000$ observations drawn form that DGP for each case yield group risks $e^*$ very close to the ones plotted in Figure~\ref{fig:simDGP}.}
{
\begin{figure}\captionsetup[subfigure]{font=footnotesize}
\centering
\includegraphics[width=0.9\linewidth]{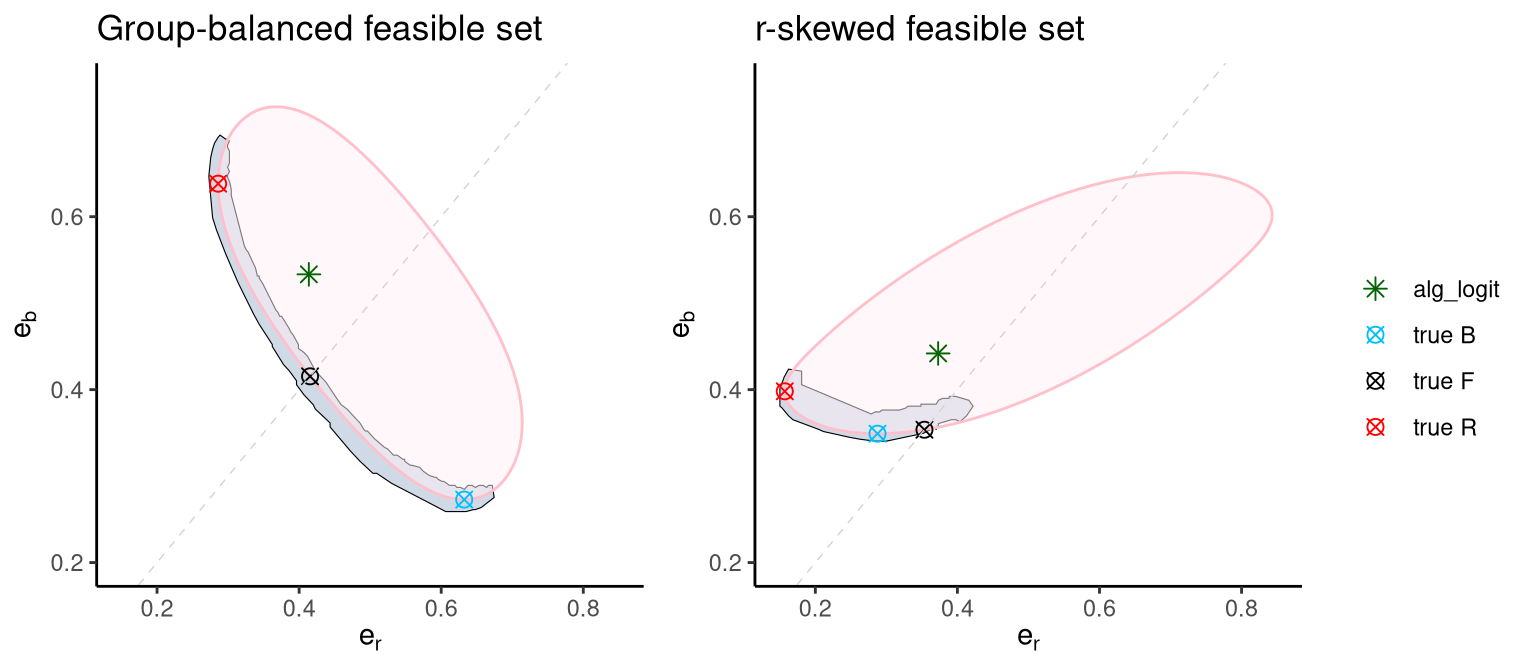}
\caption{\footnotesize{True feasible set $\E$ for group-balanced DGP (left) and $r$-skewed DGP (right), based on evaluating $\ss_\E(\q)$ in Eq.~\eqref{eq:ss} at true $\btheta$ in $500$ directions, with expectations approximated by averaging over $10^7$ observations drawn from the respective DGP. True $\R$ and $\B$ obtained similarly, using Eq.~\eqref{eq:coord}. True $\F$ based on evaluating  Eq.~\eqref{eq:h_tildeE} at true $\btheta$ and minimizing over $c$ via stochastic gradient descent. True risk $e^*$ (green asterisk) induced by the logit algorithm based Eq.~\eqref{eq:expression:e_g}. Shaded region: $95\%$ confidence set for $\sF$ built using $10,000$ observations drawn from the respective DGP and test inversion detailed in Section \ref{subsec:FA_frontier} with $\Delta\btheta$ estimated by logit lasso.
}}\label{fig:simDGP}
\end{figure}
}

Figure \ref{fig:simDGP} depicts the population feasible set $\E$ corresponding to each DGP (pink region), along with a $95\%$ confidence set for the frontier $\sF$ (shaded grey region), constructed using a random sample of $10,000$ observations drawn from the respective DGPs, and the risk $e^*$ induced by the status-quo algorithm $a^*$ (green asterisk).
Throughout Section \ref{subsec:MC}, we fit the nuisance parameter $\Delta\btheta$ using logit lasso and $5$-fold sample splitting. 
To assess the finite sample properties of the LDA test in Section \ref{sec:test-LDA} and the distance-to-$\F$ test in Section \ref{sec:distance_F}, we test whether $e^*$ is on the frontier and whether it is at a specific distance from $\F$.

Table \ref{tbl:mc-results} reports the simulation results. Overall, the Monte Carlo exercise suggests good finite sample properties for our proposed tests, especially as the sample size increases.
The top panel corresponds to the test for weak group skew (Section~\ref{sec:test-group-balance}), where the third column reports the frequency with which the $95\%$ confidence set for the vector $(\R, \B)$ (Eq.~\ref{eq:CS-R}) based on the available sample fails to cover the true value of $(\R,\B)$.
While the $r$-skewed DGP exhibits some over-rejection at $n=1{,}000$, which is a small sample size relative to the complexity of DML estimation with 20 covariates and estimating the support function across directions, this over-rejection quickly disappears as sample size increases. On the other hand, the weak group skew test (fourth column) rejects the null of weak group skew in the balanced DGP and fails to reject it in the $r$-skewed DGP essentially with probabilities 1 and 0, respectively.
This is not surprising because the simulation DGPs are far from the boundary of the null, and because the test is conservative due to the projection step.
{
\begin{table}
\scriptsize
\vspace{-.5cm}
\caption{\vspace{-.25cm}Rejection Rates (1000 Monte Carlo Simulations, $\alpha=0.05$)\label{tbl:mc-results}}
\vspace{.1cm}
\centering
\resizebox{.8\columnwidth}{!}{
\begin{tabular}
{@{\extracolsep{3pt}}l@{}c@{}@{}c@{}@{}c@{}@{}c@{}@{}c@{}@{}} 

\multicolumn{6}{c}{\cellcolor{pink!20}\textbf{Test \textit{$\HH_0:$ Weak Group Skew}}}\\
\cline{1-6}
\multicolumn{1}{c}{$n$} & \multicolumn{1}{c}{DGP} & \multicolumn{3}{c}{$(\R,\B)\notin\mathcal{CS}_n(\R,\B)$} & \multicolumn{1}{c}{\hspace{-.5cm}$\HH_0$ rejected} \\
\cline{1-6}
\multicolumn{1}{c}{\multirow{2}{*}{$1,000$}} & balance & \multicolumn{3}{c}{ $0.035$} & \multicolumn{1}{c}{$\hspace{-.5cm}0.999$} \\ 
\cline{2-6}
& $r$-skew & \multicolumn{3}{c}{ $0.12$} &  \multicolumn{1}{c}{$\hspace{-.5cm}0$}\\
\cline{1-6}
\multicolumn{1}{c}{\multirow{2}{*}{$5,000$}} & balance & \multicolumn{3}{c}{$0.012$} &  \multicolumn{1}{c}{$\hspace{-.5cm}1$}\\ 
\cline{2-6}
& $r$-skew & \multicolumn{3}{c}{$0.021$} & \multicolumn{1}{c}{$\hspace{-.5cm}0$} \\
\cline{1-6}
\multicolumn{1}{c}{\multirow{2}{*}{$10,000$}} & balance & \multicolumn{3}{c}{$0.01$} &  \multicolumn{1}{c}{$\hspace{-.5cm}1$}\\ 
\cline{2-6}
& $r$-skew & \multicolumn{3}{c}{$0.012$} & \multicolumn{1}{c}{$\hspace{-.5cm}0$} \\
\cline{1-6}
\multicolumn{6}{c}{\cellcolor{pink!20}\textbf{Test \textit{$\HH_0:$ There Is No LDA to $\tilde{e}$}}}\\
\cline{1-6}
\multicolumn{1}{c}{$n$} & \multicolumn{1}{c}{DGP} & \multicolumn{1}{c}{$\hspace{.25cm}\tilde{e}=\R$} & \multicolumn{1}{c}{$\hspace{.25cm}\tilde{e}=\B$} & \multicolumn{1}{c}{$\tilde{e}=${\tiny$(\R+\B)/2$}} & \multicolumn{1}{c}{$\hspace{-.4cm}\tilde{e}=e^*$} \\
\cline{1-6}
\multicolumn{1}{c}{\multirow{2}{*}{$1,000$}} & balance & $\hspace{.25cm}0.046$ & $\hspace{.25cm}0.104$ & $0.145$ & $\hspace{-.5cm}0.397$\\ 
\cline{2-6}
& $r$-skew & $\hspace{.25cm}0.051$ & $\hspace{.25cm}0.236$ & $0.073$ & $\hspace{-.5cm}0.162$\\
\cline{1-6}
\multicolumn{1}{c}{\multirow{2}{*}{$5,000$}} & balance & $\hspace{.25cm}0.026$ & $\hspace{.25cm}0.037$ & $0.993$ & $\hspace{-.5cm}1$\\ 
\cline{2-6}
& $r$-skew & $\hspace{.25cm}0.018$ & $\hspace{.25cm}0.043$ & $0.059$ & $\hspace{-.5cm}1$ \\
\cline{1-6}
\multicolumn{1}{c}{\multirow{2}{*}{$10,000$}} & balance & $\hspace{.25cm}0.026$ & $\hspace{.25cm}0.029$ & $1$ & $\hspace{-.5cm}1$ \\ 
\cline{2-6}
& $r$-skew & $\hspace{.25cm}0.007$ & $\hspace{.25cm}0.039$ & $0.264$ & $\hspace{-.5cm}1$ \\
\cline{1-6}
\multicolumn{6}{c}{\cellcolor{pink!20}\textbf{Test \textit{$\HH_0: \rho(\tilde{e}, \F)=\delta$ for Constant $\delta$ Equal to the True Distance}}}\\
\cline{1-6}
\multicolumn{1}{c}{$n$} & \multicolumn{1}{c}{DGP} & \multicolumn{1}{c}{$\hspace{.25cm}\tilde{e}=\R$} & \multicolumn{1}{c}{$\hspace{.25cm}\tilde{e}=\B$} & \multicolumn{1}{c}{$\tilde{e}=${\tiny$(\R+\B)/2$}} & \multicolumn{1}{c}{$\hspace{-.4cm}\tilde{e}=e^*$} \\
\cline{1-6}
\multicolumn{1}{c}{\multirow{2}{*}{$1,000$}} & balance & $\hspace{.25cm}0.242$ & $\hspace{.25cm}0.273$ & $0.251$ & $\hspace{-.5cm}0.159$\\ 
\cline{2-6}
& $r$-skew & $\hspace{.25cm}0.063$ & $\hspace{.25cm}0.048$ & $0.052$ & $\hspace{-.5cm}0.204$\\
\cline{1-6}
\multicolumn{1}{c}{\multirow{2}{*}{$5,000$}} & balance & $\hspace{.25cm}0.081$ & $\hspace{.25cm}0.097$ & $0.088$ & $\hspace{-.5cm}0.051$\\ 
\cline{2-6}
& $r$-skew & $\hspace{.25cm}0.019$ & $\hspace{.25cm}0.027$ & $ 0.017$ & $\hspace{-.5cm}0.077$ \\
\cline{1-6}
\multicolumn{1}{c}{\multirow{2}{*}{$10,000$}} & balance & $\hspace{.25cm}0.069$ & $\hspace{.25cm}0.078$ & $0.073$ & $\hspace{-.5cm}0.046$\\ 
\cline{2-6}
& $r$-skew & $\hspace{.25cm}0.012$ & $\hspace{.25cm}0.024$ & $0.013$ & $\hspace{-.5cm}0.034$ \\
\cline{1-6}
\end{tabular}
}
\begin{tablenotes}
\footnotesize
\item Population values for the balanced DGP are $\R=[0.286, 0.638]^\intercal,\B=[0.632, 0.273]^\intercal, \F=[0.415, 0.415]^\intercal$, $e^*=[0.414, 0.533]^\intercal$; for the $r$-skewed DGP are $\R=[0.157, 0.398]^\intercal,\B=[0.288, 0.349]^\intercal,\F=[0.354, 0.354]^\intercal$, $e^*=[0.373, 0.442]^\intercal$. We take $\rho$ to be the squared Euclidean distance.
\end{tablenotes}
\end{table}
}

The middle panel reports the LDA test results (Section~\ref{sec:test-LDA}). The third (fourth) column shows the frequency with which the population point $\R$ ($\B$), which by definition belongs to $\sF$, is rejected by the test of the null that it belongs to $\sF$. The fifth (sixth) column reports the frequency with which $(\R+\B)/2$ (the risk $e^*$ associated with the logit algorithm $a^*$), which by construction does not belong to $\sF$, is rejected by the test as an element of $\sF$.
While at small sample size ($n=1,000$), the test exhibits some over-rejection for the point $\B$, this quickly disappears as sample size grows. At small and medium sample sizes ($n=1,000$ and $n=5,000$), the rejection probability for the false null that $(\R+\B)/2\in\sF$ is low for the $r$-skewed DGP, while it is high at all sample sizes for the balanced DGP. This is justifiable in light of Figure~\ref{fig:simDGP}, which shows that in the $r$-skewed DGP the chord between $\R$ and $\B$ is close to the frontier and lies largely inside its 95\% confidence set.
For $n=10,000$, the test detects the false nulls in columns five and six with substantially higher probability. 

The bottom panel reports the results for the distance-to-$\F$ test (Section~\ref{sec:distance_F}). This test is more delicate because its implementation requires solving an optimization problem to estimate $\F$. Here we observe more substantial over-rejection for $n=1,000$, but the distortion gets markedly reduced as sample size increases.

\subsection{Empirical Illustration}\label{subsec:empirical}
We revisit the analysis in \citet*[OPVM henceforth]{obe:pow:vog:mul19}, who analyze properties of the algorithm used by a research hospital to determine if a patient should be automatically enrolled in a high-risk care management program.
The algorithm used by the research hospital aims at predicting patients' health needs based on total medical expenditures (the label on which the algorithm is trained). It produces a health risk score and automatically enrolls a patient in the high-risk care management program if that patient's risk score exceeds the $97^{\text{th}}$ percentile of all predicted scores. 
We reassess this algorithm through our testing procedures.
To do so, we use the synthetic data made available by \citet*{li:lin:obe19} at \href{https://gitlab.com/labsysmed/dissecting-bias}{GitLab} to replicate all analyses in \citetalias{obe:pow:vog:mul19}.
The data include $48,784$ patient observations, of which $5,582$ self report as Black and the others self report as White ($g\in\{\textsf{bl},\textsf{wh}\}$). The data include $149$ covariates, such as age, gender, comorbidity and medication variables, costs, and biomarkers, and provide information about each patient's number of active chronic conditions in the subsequent year, viewed as the true measure of health needs of a patient.
Following \citetalias{lia:lu:mu:oku24}, we let $\ell(d,\Y)=\mathds{1}\{\Y\neq d\}$ be the classification loss and, unless explicitly stated otherwise, $\Y_i=\mathds{1}\{$patient $i$ has 6 or more chronic conditions$\}$, with the choice of 6 driven by the fact that it is the $97^{\text{th}}$ percentile of active chronic condition numbers across patients in the sample.
We use random forests with $5,000$ trees to estimate $\Delta\btheta$ as described in Section~\ref{sec:estimation}.
\begin{figure}[t]\captionsetup[subfigure]{font=footnotesize}
    \centering
    \includegraphics[width=0.9\linewidth]{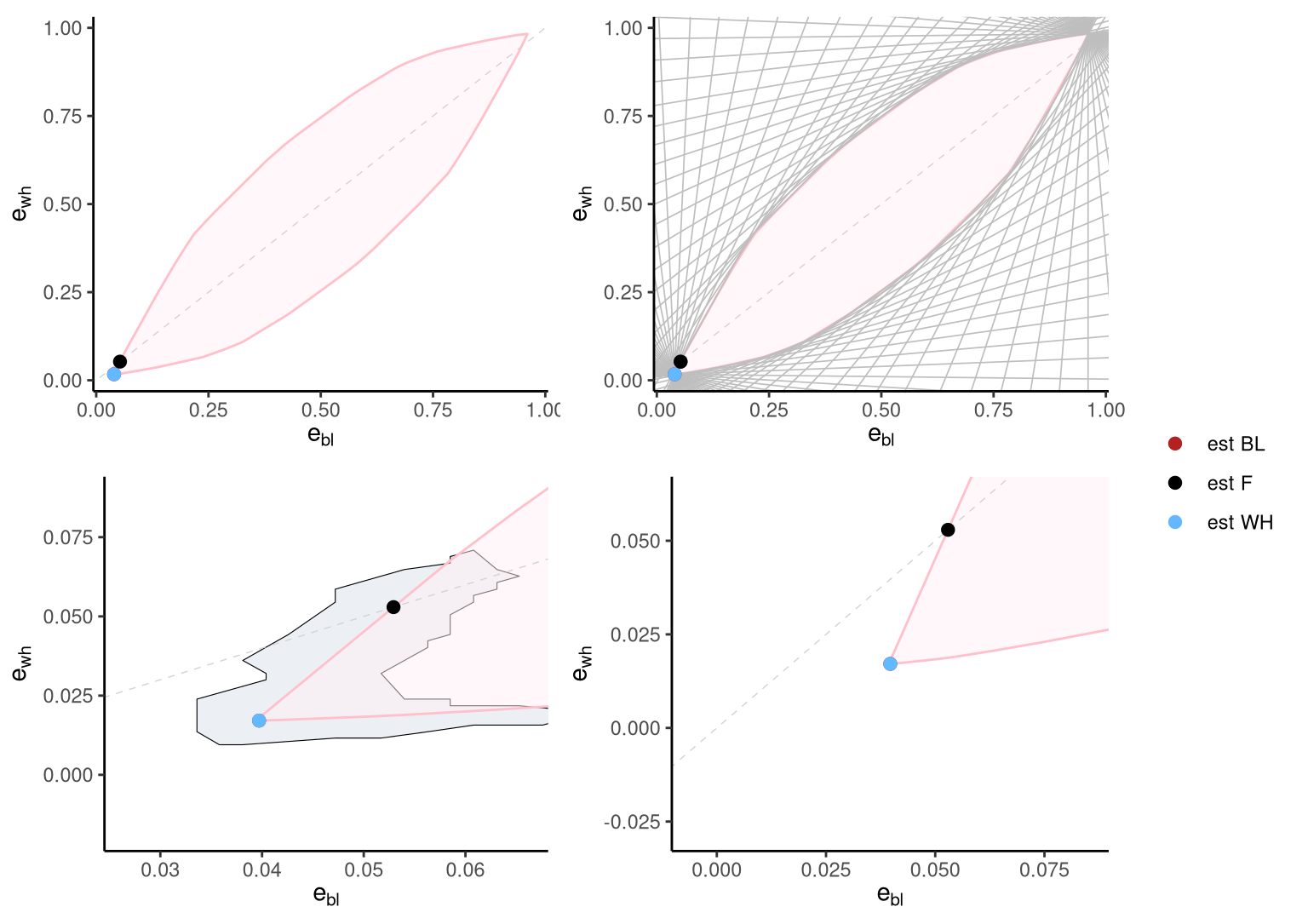}
    \caption{Top-left panel: $\widehat{\E}$; top-right panel: $\widehat{\E}$ along with one hundred supporting hyperplanes; bottom-left panel: zoom-in to $\widehat{\sF}$ and the $95\%$ confidence set around this frontier; bottom-right panel: further zoom-in to the best group-specific points $\textsf{BL}$ and $\textsf{WH}$, and the fairest point $\F$.}
    \label{fig:application:estimated:sets}
\end{figure}

\subsubsection{Feasible Set Estimation and Inference for the Frontier}
\label{subsec:emp-feasible-set-frontier-est}
We report in Figure~\ref{fig:application:estimated:sets} an estimate of the feasible set $\E$ based on Eq. \eqref{eq:Lambda_hat} using $1,000$ directions (top-left panel), along with 100 estimated supporting hyperplanes (top-right panel), where across all panels the horizontal (vertical) axis is the risk for Blacks (Whites), denoted as $e_{\textsf{bl}}$ ($e_{\textsf{wh}}$). We zoom in to show the estimated FA-frontier $\widehat\sF$ and its $95\%$ confidence set (bottom-left panel), and zoom in further to show the best risk achievable for the Black patients (bottom-right panel, red point labeled $\textsf{BL}$), which coincides with and is overlaid by the best risk for the White patients (bottom-right panel, blue point labeled $\textsf{WH}$).

\subsubsection{Hypothesis Testing}\label{subsec:hyp-test}
We next report the weak group skew test results in Figure~\ref{fig:tests}-Panel (a), where both $\textsf{BL}$ and $\textsf{WH}$ are below the $45^\circ$ degree line and the feasible set is $\textsf{wh}$-skewed; the test fails to reject the null of weak group skew, suggesting that implementing Pareto-dominated algorithms could be justified based on the designer's preferences over fairness and accuracy.
In Figure~\ref{fig:tests}-Panel (b) we plot $\widehat{\sF}$, along with its $95\%$ confidence set. 
Figure~\ref{fig:tests}-Panel (b) also plots the estimated group risks associated with the original algorithm used by the hospital (a green asterisk) and three alternative algorithms that \citetalias{obe:pow:vog:mul19} experiment with to assess whether different label choices yield decision rules that are more accurate and fairer than the algorithm currently used by the hospital. One of these algorithms is trained to predict total cost (hollow diamond with a cross in Figure~\ref{fig:tests}-Panel (b)), one to predict avoidable costs (filled diamond), and the other one to predict the number of active chronic conditions (hollow diamond). All our tests take into account the finite sample estimation error of the group risks induced by these four algorithms. The top panel of Table \ref{tbl:LDA-results} reports the values of the LDA test statistics and the associated critical values that we compute for the four algorithms considered by \citetalias{obe:pow:vog:mul19}. The hypothesis that the original algorithm yields group risks on the frontier is rejected, and so is the same hypothesis for group risks associated with the algorithm trained to predict total costs.
On the other hand, we fail to reject that the algorithms trained to predict avoidable costs and the number of active chronic conditions yield group risks on the FA-frontier.
Regarding the three algorithms proposed by \citetalias{obe:pow:vog:mul19}, if one were to plot the set $\sC$ corresponding to the original algorithm, it would be immediate to see that one cannot reject the hypothesis that all other algorithms considered by \citetalias{obe:pow:vog:mul19} improve upon it, both in terms of fairness and accuracy.

\begin{figure}[t]\captionsetup[subfigure]{font=footnotesize}
\centering
\subfigure[Candidate values for $(\textsf{BL}, \textsf{WH})$]{
\includegraphics[width=0.46\linewidth]{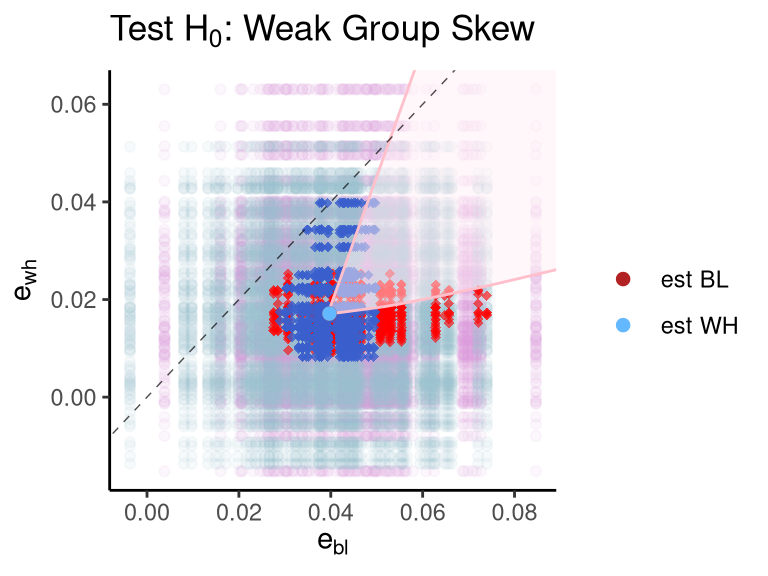}
}
\hspace{1em}
\subfigure[$3$ experimental algorithms]
{\includegraphics[width=0.46\linewidth]{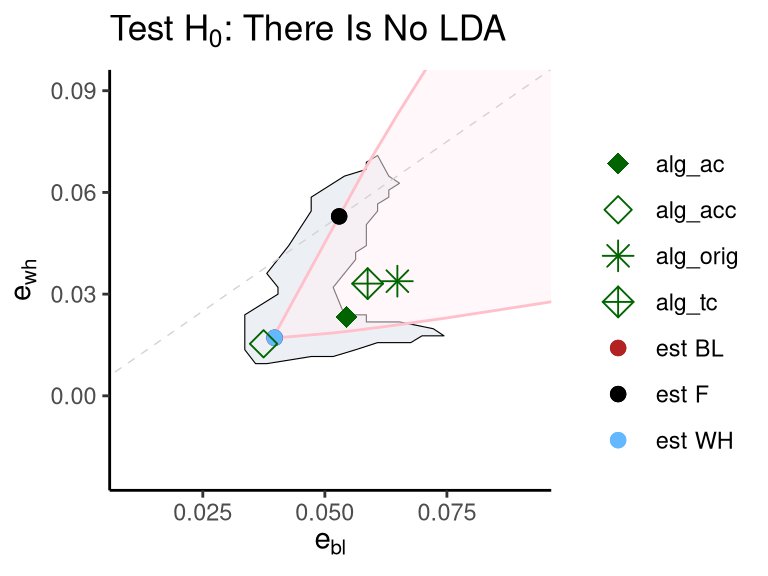}
}
    \caption{Panel (a): Plum-colored (respectively, light-blue colored) circles correspond to candidate values for the best point for Blacks $\textsf{BL}$ (Whites $\textsf{WH}$) sampled from a normal distribution centered at the estimated $\textsf{BL}$ ($\textsf{WH}$), and red (blue) diamonds correspond to non-rejected values.
    Panel (b): $\widehat{\sF}$ along with its $95\%$ confidence set and the estimated group risks for four algorithms considered by \citetalias{obe:pow:vog:mul19}: the original algorithm used by the hospital (asterisk); one that predicts total cost (hollow diamond with a cross); one that predicts avoidable costs (filled diamond); and one that predicts the number of active chronic conditions (hollow diamond).}
    \label{fig:tests}
\end{figure}

The bottom panel of Table~\ref{tbl:LDA-results} shows that the closest risk to $\F$ is the one associated with the algorithm predicting total costs. However, all confidence intervals overlap.

{
\begin{table}[t]
\scriptsize
\caption{\vspace{-.25cm}Results for the LDA Test and Confidence Sets for the Distance to $\F$ (for $\alpha=0.05$)\label{tbl:LDA-results}}
\vspace{.1cm}
\centering
\resizebox{0.83\columnwidth}{!}{
\begin{tabular}
{@{\extracolsep{3pt}}c@{}c@{}@{}c@{}@{}c@{}@{}c@{}@{}c@{}@{}} 

\multicolumn{5}{c}{\cellcolor{pink!20}\textbf{Test \textit{$\HH_0:$ There Is No LDA}}}\\
\cline{1-5}
\multicolumn{1}{c}{} & 
 \multicolumn{1}{c}{Original} & \multicolumn{1}{c}{Total Costs} & \multicolumn{1}{c}{Avoid. Costs} & \multicolumn{1}{c}{Act. Chr. Cond.} \\
\cline{1-5}
Estimated Risks & $ (0.065, 0.034)$ & $(0.059, 0.033)$ & $(0.054, 0.023)$ & $(0.037, 0.015)$ \\
Test Statistic & $3.767$ & $2.753$ & $1.428$ & $0.541$\\
Critical Value & $1.885$ & $1.948$ & $1.732$ & $1.642$ \\
Conclusion & Rejected & Rejected & Not Rejected & Not Rejected \\ 
\cline{1-5}
\multicolumn{5}{c}{\cellcolor{pink!20}\textbf{Distance to $\F=(0.049, 0.049)$}}\\
\cline{1-5}
Estimated Distance & $0.0005$ & $0.0004$ & $0.0007$ & $0.0013$ \\
Confidence Set & $(0.000, 0.001)$ & $(0.000, 0.001)$ & $(0.000, 0.002)$ & $(0.000, 0.003)$\\
\cline{1-5}
\end{tabular}
}
\begin{tablenotes}
\footnotesize
\item Top panel: LDA test statistics and $0.05$-level critical values associated with the original algorithm and the three experimental algorithms (predicting, respectively, total costs; avoidable costs; number of active chronic conditions) analyzed by \citetalias{obe:pow:vog:mul19}. Bottom panel: estimated squared-Euclidean distance to the $\F$ point and corresponding confidence set for this distance.
\end{tablenotes}

\end{table}
}

\subsubsection{Performance of Algorithms Constructed to Be on a Constrained Frontier}\label{subsec:build-alg}
Finally, we compare various outcomes associated with the algorithms considered by \citetalias{obe:pow:vog:mul19} with those of decision rules resulting from the algorithms that we propose in Section~\ref{subsec:algorithms:frontier}. To do so, we randomly split the sample into two halves, and use one half (training data) to estimate the nuisance parameter $\Delta\btheta$, and implement Eq.~\eqref{eq:a_hat} on the other half (evaluation data) for the following choices of $\q$:
\begin{enumerate}
    \item[(i)] $\q=[-1~~~0]^\intercal$, yielding an algorithm that (asymptotically) achieves the best point on $\sF$ for Black patients, which we call Rawlsian because $e_\textsf{bl}>e_\textsf{wh}$ at $\textsf{BL}$.
    \item[(ii)] $\q=[0~~-1]^\intercal$, yielding an algorithm that (asymptotically) achieves the best point on $\sF$ for White patients, which we call Majority as Whites are the majority of the sample.
    \item[(iii)] $\widehat{\q}(\widehat\F)$, yielding an algorithm based on a direction estimated using the entire sample that (asymptotically) achieves $\F$, the fairest point on $\sF$, as established in Proposition~\ref{prop:algorithm_consistency}, which we call Egalitarian.
    \item[(iv)] $\q=-\frac{\sqrt{2}}{2}[1~~1]^\intercal$, yielding an algorithm that weighs both groups equally, which we call Utilitarian as it can be expressed as a generalization of the Utilitarian rule.
\end{enumerate}
Throughout, we recognize that to compare algorithms we should enforce a global capacity constraint on the total percentage of patients assigned to the high-risk case management program, so that variation across algorithms is not confounded with a possibly more generous care program.
To enforce the capacity constraint, we need to require the algorithms to satisfy $\int a(x)d\bP_\X\le \bar{a}$ for some known constant $\bar{a}$.
For example, $\bar{a}=0.03$ when comparing with the algorithm currently used by the hospital which assigns only patients with risk score above the $97^\text{th}$ percentile to the high-risk care program.
Recall from the discussion following Proposition~\ref{prop:sf} that without capacity constraints, we have
\begin{align}
    \h_\E(\q)&=\mathbb{E}\left[\q_1\tfrac{\theta_0^r(\X)}{\mu_r}+\q_2\tfrac{\theta_0^b(\X)}{\mu_b}\right]+\max_{a\in\mathcal{A}(\sX)}\mathbb{E}\left[a(\X)k\left(\btheta(\X),\cM\q\right)\right].\label{eq:h:a_based}
\end{align}
When $\mathcal{A}(\sX)$ is constrained to only include algorithms such that $\int a(x)d\bP_\X\le \bar{a}$, maximization in Eq.~\eqref{eq:h:a_based} is achieved by setting
\begin{align}
    a^\text{opt}(\X;\q)=\mathds{1}\{k\left(\btheta(\X),\cM\q\right)>\max(0,\mathrm{quant}_{k(\btheta(\X),\cM\q)}(1-\bar{a}))\},\label{eq:a_constrained}
\end{align}
for $\mathrm{quant}_{k(\btheta(\X),\cM\q)}(\alpha)$ the $\alpha$-quantile of $k(\btheta(\X),\cM\q)$ and $k(\btheta(\X),\cM\q)$ as in Eq.~\eqref{eq:k_Delta}.\footnote{In this case, the \emph{constrained} feasible set is $\E^\texttt{co}_{\bar{a}}\equiv \left\{\bigl(\er(a),\eb(a)\bigr)\in\mathbb{R}^2:a\in\mathcal{A}(\sX)~\text{and}~\int a(x)d\bP_\X\le \bar{a}\right\}$, a convex subset of $\E$, and the algorithms in Eq.~\eqref{eq:a_constrained} is on its frontier.}
In words, for our example with $\bar{a}=0.03$, high-risk care management is assigned to those patients with positive values of $k(\btheta(\X),\cM\q)$ that exceed its $97^\text{th}$ percentile.

The results of our first exercise are reported in Figure~\ref{fig:application:build:alg}.
The top panel replicates Figure 1-(a) in \citetalias{obe:pow:vog:mul19}, and shows that at each percentile of the original algorithm's risk score (the horizontal axis), Black patients have a substantially larger number of active chronic conditions (the vertical axis) than White patients.\footnote{The top panel of Figure \ref{fig:application:build:alg} is not identical to Figure 1-(a) of \citetalias{obe:pow:vog:mul19} because it is plotted using the synthetic data, instead of the real data, and because the code used to generate Figure 1-(a) of \citetalias{obe:pow:vog:mul19} had an error that was later corrected after the publication of \citetalias{obe:pow:vog:mul19}; see the documentation of \citet{li:lin:obe19} for more details.}
Even among patients automatically enrolled in the high-risk care management program (those with a risk score above the $97^{\text{th}}$ percentile), Black patients appear to be in worse health than White patients.

We ask whether the algorithms that we propose are able to select for treatment patients that, in fact, exhibit substantially worse health outcomes one year later.
The bottom four panels in Figure~\ref{fig:application:build:alg} plot, for each of the algorithms described at the beginning of this section, the number of active chronic conditions for (a) Black patients that our algorithm does not assign to the high-risk care program (dash-dotted, light-purple line); (b) White patients that our algorithm does not assign to the high-risk care program (solid salmon line); (c) Black patients that our algorithm assigns to the high-risk care program (dash-dotted, dark-purple line); (b) White patients that our algorithm assigns to the high-risk care program (solid orange line).
For comparability  with the top panel, on the horizontal axis we continue to report the risk score produced by the algorithm used by the hospital.

\begin{figure}\captionsetup[subfigure]{font=footnotesize}
    \centering
    \includegraphics[width=0.95\linewidth]{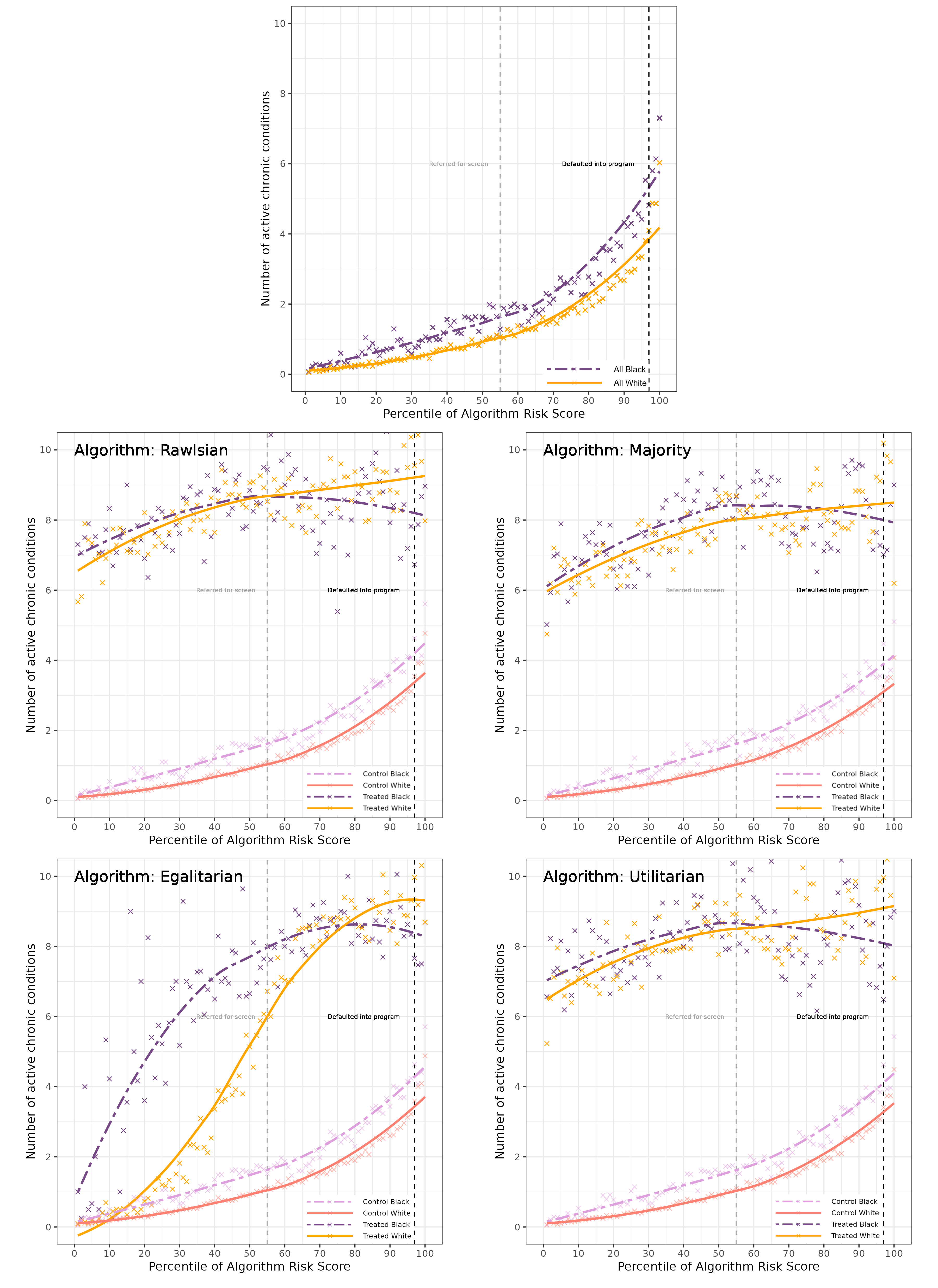}
    \caption{Average number of active chronic conditions within each risk-score percentile bin by treatment group under the alternative algorithms on the FA frontier subject to $3\%$ capacity constraint, averaged across $20$ replications of the $50$-$50$ split.}
    \label{fig:application:build:alg}
\end{figure}

The main takeaway from Figure~\ref{fig:application:build:alg} is that our four proposed algorithms, which use group identity to estimate $\Delta\btheta$ but not for treatment assignment, are successful at selecting for treatment patients who, one year later, experience substantially worse health outcomes.
Moreover, with the Ralwsian, Majority, and Utilitarian algorithms, Black and White patients assigned to treatment have similar numbers of chronic conditions.
Only the Egalitarian algorithm shows notable disparities for patients with hospital risk scores below the $70^{\text{th}}$ percentile.
We think this result might be due to two factors: estimation of $\F$ is challenging and hence the direction $\widehat{\q}(\widehat{\F})$ might be imprecisely estimated; and the feasible set is $\textsf{wh}$-skewed and hence the Egalitarian algorithm leads to a Pareto-dominated outcome.

Our last exercise aims at further assessing the extent to which our proposed algorithms may reduce the substantial disparities between Black and White patients in the current program screening practices documented by \citetalias{obe:pow:vog:mul19}.
To illustrate the potential for improvement over the hospital's algorithm, \citetalias{obe:pow:vog:mul19} simulate ``a counterfactual world with no gap in health conditional on risk'' (p.3).
They construct an infeasible, \emph{couterfactual} algorithm that uses group identity and patients' ex-post active chronic health conditions to find the sickest Black patient with health risk score just below a threshold (the ``inframarginal Black patient'') and the healthiest White patient with health risk score just above the same threshold (the ``supramarginal White patient''). If the number of chronic health conditions of the inframarginal Black patient is larger than that of the supramarginal White patient, they iteratively swap them until the number of chronic health conditions of the inframarginal Black patient equals that of the supramarginal White patient.
\citetalias{obe:pow:vog:mul19} find that at all risk thresholds $\alpha$ above the median, the counterfactual algorithm increases the fraction of Black patients treated. 
Table~\ref{tbl:frac-bl-trt}, columns 2-3, show that across the various thresholds for automatic enrollment in the program that we consider, the fraction of Black patients would rise between 6 and 41 percentage points.\footnote{In columns 2-3, the fractions reported are based on a denominator that equals the total number of patients with risk scores above a certain percentile of risk scores, e.g., the $55^{\text{th}}$ percentile, with this threshold viewed as the threshold above which the patient is automatically enrolled in the high-risk care program, and a numerator that equals the number of Black patients with risk score above that percentile.}

We ask how much of this gap could be filled if one were to use our algorithms, which are feasible and do not rely on ex-post knowledge of the number of active chronic conditions nor on using group identity for assignment to the high-risk care program.
The results, reported in Table~\ref{tbl:frac-bl-trt}, columns 4-7,\footnote{For each algorithm, the fraction of Black patients treated at each capacity threshold is computed as the following ratio.
The denominator equals the total number of patients treated under a given algorithm, e.g., the Rawlsian algorithm, subject to the constraint that at each threshold $\bar{a}\in[0.55,0.97]$ at most $100(1-\bar{a})\%$ of the evaluation sample is treated, and for each threshold $\bar{a}$ the outcome $Y$ equals the indicator of whether the number of active chronic conditions exceeds the $\bar{a}$-quantile of the distribution of the number of active chronic conditions.}
show that each of our four algorithm yields an increase in the fraction of Black patients treated at all capacity thresholds. The Rawlsian and the Utilitarian algorithms yield the largest increases, ranging between 4 and 18 percentage points across the various capacity thresholds.
This shows that these two feasible, easy to implement algorithms can close between $26\%-69\%$ of the gap between the algorithm that the hospital uses and the counterfactual, infeasible algorithm simulated by \citetalias{obe:pow:vog:mul19}.

{
\begin{table}
\scriptsize
\vspace{-.5cm}
\caption{\vspace{-.25cm}Fraction of Black Patients Treated among All Treated \label{tbl:frac-bl-trt}}
\vspace{.1cm}
\centering
\resizebox{0.95\columnwidth}{!}{
\begin{tabular}{c|c@{\hspace{-.5cm}}c@{\hspace{-.5cm}}|cccc}
\cline{1-7}
\multicolumn{1}{c}{\cellcolor{pink!20}} & \multicolumn{2}{c}{\cellcolor{pink!20}\hspace{-.1cm}\textbf{Algorithms from \citeauthor{obe:pow:vog:mul19}}} & \multicolumn{4}{c}{\cellcolor{pink!20}\textbf{Algorithms on the FA-Frontier}}\\ 
\cline{1-7}
Capacity Threshold & $\hspace{.5cm}$Original & Counterfactual & Rawlsian & Majority & Egalitarian & Utilitarian \\
\cline{1-7}
$55$ & $\hspace{.5cm}0.120$ & $0.184$ & $0.164$ & $0.164$ & $0.160$ & $0.164$\\
$69$ & $\hspace{.5cm}0.128$ & $0.255$ & $0.185$ & $0.184$ & $0.163$ & $0.185$\\
$82$ & $\hspace{.5cm}0.138$ & $0.327$ & $0.206$ & $0.207$ & $0.160$ & $0.206$\\
$89$ & $\hspace{.5cm}0.151$ & $0.407$ & $0.218$ & $0.223$ & $0.166$ & $0.219$\\ 
$94$ & $\hspace{.5cm}0.167$ & $0.498$ & $0.286$ & $0.262$ & $0.190$ & $0.273$\\
$97$ & $\hspace{.5cm}0.184$ & $0.592$ & $0.368$ & $0.300$ & $0.253$ & $0.338$ \\
\cline{1-7}
\end{tabular}
}
\begin{tablenotes}
\footnotesize
\item The distribution of the number of active chronic conditions is such that the $55^\text{th}$ to the $68^\text{th}$ percentiles all correspond to $1$ active chronic condition, the $69^\text{th}$-$81^\text{st}$ correspond to $2$, the $82^\text{nd}$-$88^\text{th}$ correspond to $3$, the $89^\text{th}$-$92^\text{nd}$ correspond to $4$, the $94^\text{th}$-$95^\text{th}$ correspond to $5$, and the $96^\text{th}$-$97^\text{th}$ correspond to $6$. 
\end{tablenotes}
\end{table}
}

We conclude by noting that results based on random forests trained using the \texttt{grf} package are not guaranteed to reproduce exactly across platforms, even with the same seed (all simulations and estimation are in \texttt{R}). This is a known feature of \texttt{grf} (see the \href{https://grf-labs.github.io/grf/REFERENCE.html\#forests-predict-different-values-depending-on-the-platform-even-though-the-seed-is-the-same}{reference manual}).
Similarly, tests based on optimization via stochastic gradient descent implemented by the \texttt{torch} package are not exactly reproducible, even after seed setting, which affects the $\F$ estimator and thus the distance-to-$\F$ test (see the \href{https://github.com/mlverse/torch/issues/1311}{repository discussion}).
To assess sensitivity of our results to these features, we repeat the entire empirical exercise 20 times with different seeds. Appendix \ref{appn:C} reports results and includes a robustness check using logit lasso for estimating the nuisance function.\footnote{For the simulations in Section~\ref{subsec:MC}, we expect cross-platform variability to be negligible, as results are averaged over 1000 Monte Carlo replications.}
While the results exhibit some nontrivial variation across seeds (although the qualitative results are unchanged), we view this as expected, considering that the nonparametric estimation step involves 149 covariates against a total sample size of $48,784$ and a minority group of size $5,582$.

\section{Conclusion}\label{sec:conclude}
We provide a consistent nonparametric estimator for a theoretical fairness-accuracy frontier proposed by \citetalias{lia:lu:mu:oku24} and algorithms that attain points on this frontier. We obtain the estimator through judicious use of the separating hyperplane theorem and the support function of the (convex) feasible set of expected losses associated with all possible algorithms, a portion of whose boundary coincides with the FA-frontier. We provide a DML estimator of the support function and show it converges to a tight Gaussian process as sample size increases. We formulate important policy-relevant hypotheses that have received much attention in the fairness literature as restrictions on the support function and construct valid test statistics. We provide an estimator for the distance between a given algorithm and the fairest point on the frontier.
We carry out a Monte Carlo exercise that illustrates the good finite sample properties of our method, and demonstrate its practical relevance by revisiting the empirical analysis in \citetalias{obe:pow:vog:mul19}. Our results show that the algorithm that a research hospital employs to screen patients for high-risk care is not on the frontier. Our proposed algorithms substantially improve over this status quo in terms of both fairness and accuracy.

\bibliographystyle{ecta-fullname} 
\bibliography{draft_LM_arXiv.bib}  

\begin{appendix}

\section{Proofs of Main Results}\label{appn:A}
\subsection{Proofs for Section \ref{sec:sup_fun}}
\begin{proof}[Proof of Proposition \ref{prop:sf}]
As shown in the discussion leading to Eq.~\eqref{eq:E_as_Aumann}, $\E=\left\{\mathbb{E}[\cM\vartheta(\X)]: \vartheta(\X)\in \linseg(\X) \right\}=\mathbf{E}\left[\cM\linseg(\X)\right]$, where $\linseg(\X) \equiv \conv\left(\{\btheta_0(\X), \btheta_1(\X)\}\right)$, with $\btheta_d(\X)$ defined in (\ref{eq:def_theta_d}) for $d\in\{0,1\}$. Hence, $\cM\linseg(\X)$ is a random compact interval \citep[Example 1.11]{mol:mol18}, and $\mathbf{E}\left[\cM\linseg(\X)\right]$ is its \emph{Aumann expectation} \citep[Def. 3.1]{mol:mol18}, which is well defined because $\linseg(\X)$ is an integrable random convex set owing to $\mathbb{E}[|\theta_d^g(\X)|]\leq \mathbb{E}[|\theta_d^g(\X)|^2]^{1/2} \leq \mathbb{E}\big[(\LL_d^g)^2\big]^{1/2}<{\sqrt{c_2}}<\infty$ for any $d\in\{0,1\},g\in\{r,b\}$ by Assumption~\ref{asm:moments}. 
It follows that $\h_{\E}(\q)=\h_{\mathbf{E}\left[\cM\linseg(\X)\right]}(\q)=\mathbb{E}[\h_{\cM\linseg(\X)}(\q)]$ \citep[][Theorem 3.11]{mol:mol18}. Observe that
\begin{align}
   \h_{\cM\linseg(\X)}(\q)\equiv\max_{\vartheta(\X) \in \linseg(\X)} (\cM\q)^\intercal\vartheta(\X)=\max\{(\cM\q)^\intercal\btheta_0(\X),(\cM\q)^\intercal\btheta_1(\X)\}\label{eqn:supp-func-lineseg},
\end{align}
where the last equality in Eq.~\eqref{eqn:supp-func-lineseg} is well-known in the literature \citep[e.g.,][p.~105]{roc97}.
Taking the expectation with respect to $\bP(\X)$ yields the first line of Eq.~\eqref{eqn:sf-new}, which we can re-write as
\begin{align*}
\h_\E(\q)=\mathbb{E}\bigg[(\cM\q)^\intercal\btheta_0(\X)+(\cM\q)^\intercal\big(\btheta_1(\X)-\btheta_0(\X)\big)\mathds{1}\big\{(\cM\q)^\intercal\big(\btheta_1(\X)-\btheta_0(\X)\big)>0\big\}\bigg].
\end{align*}
The law of iterated expectations yields the expression in the second line of Eq.~\eqref{eqn:sf-new}.
\end{proof}

\begin{proof}[Proof of Proposition \ref{prop:ss}] 
The following proof closely follows the argument from \citet[Lemma 3]{cha:che:mol:sch18}. Take any $\|\delta\|_E\to 0$,
\vspace{-.25cm}
\begin{align*}
&\frac{1}{\|\delta\|_E}\Biggl(\mathbb{E}\big[\big(\cM(\q+\delta)\big)^\intercal\btheta_0 +k\big(\btheta,\cM(\q+\delta)\big)\mathds{1}\big\{k\big(\btheta,\cM(\q+\delta)\big)>0\big\}\big]\\
&\hspace{6cm}-\mathbb{E}\big[(\cM\q)^\intercal\btheta_0 +k(\btheta,\cM\q)\mathds{1}\{k(\btheta,\cM\q)>0\}\big]\Biggr)
\\
=&\frac{\delta^\intercal}{\|\delta\|_E}\mathbb{E}\big[\cM\btheta_0 +\cM(\btheta_1-\btheta_0)\mathds{1}\{k(\btheta,\cM\q)>0\}\big]+\frac{1}{\|\delta\|_E}\mathbb{E}[R(\q,\delta)],
\end{align*}
where
\vspace{-0.5cm}
\begin{align*}
R(\q,\delta) &\equiv\big(\cM(\q+\delta)\big)^\intercal(\btheta_1-\btheta_0)\mathds{1}\big\{k(\btheta,\cM\q)\leq 0< k\big(\btheta,\cM(\q+\delta)\big)\big\}\\
&\hspace{.5cm}-\big(\cM(\q+\delta)\big)^\intercal(\btheta_1-\btheta_0)\mathds{1}\big\{k(\btheta,\cM\q) > 0 \geq k\big(\btheta,\cM(\q+\delta)\big)\big\}\\
\implies
&\sup_{q\in\bS^1} \mathbb{E}[|R(\q,\delta)|]
\lesssim \|\delta\|_E\cdot
\bigl\|\btheta_1-\btheta_0\bigr\|_{L^2(\bP)}\cdot\sup_{q\in\bS^1} \biggl
\|\mathds{1}\bigl\{\left|k(\btheta,\cM\q)\right|<\left|k(\btheta,\cM\delta)\right|\bigr\}\biggr\|_{L^2(\bP)}\\
&\lesssim\|\delta\|_E\cdot\sup_{\q\in\bS^1} \mathbb{P}\bigl(\left|k(\btheta,\cM\q)\right|<\left|k(\btheta,\cM\delta)\right|\bigr)^{1/2}\lesssim \|\delta\|_E\big(\|\delta\|_E^{m/2}+\|\delta\|_E^{1/2}\big)^{1/2}
\end{align*}
where the first inequality follows from Hölder's inequality and that, for each $\X$,  conditional on the event  $\left|k(\btheta,\cM\q)\right|<\left|k(\btheta,\cM\delta)\right|$,
we have $|\big(\cM(\q+\delta)\big)^\intercal(\btheta_1-\btheta_0)|\leq 2|\cM\delta^\intercal(\btheta_1-\btheta_0)|\leq2\|\cM\delta\|_E\|\btheta_1-\btheta_0\|_E\lesssim\|\delta\|_E\|\btheta_1-\btheta_0\|_E$;
the second inequality follows from $\bigl\|\btheta_1-\btheta_0\bigr\|_{L^2(\bP)}<\infty$ by Assumption~\ref{asm:moments}.  
The last inequality follows from
\begin{align}
    \sup_{\q\in\bS^1} \mathbb{P}\bigl(\left|k(\btheta,\cM\q)\right|<\left|k(\btheta,\cM\delta)\right|\bigr)&\leq \sup_{\q\in\bS^1} \mathbb{P}\bigl(\left|k(\btheta,\cM\q)\right|<\|\delta\|_E^{1/2}\bigr)+\mathbb{P}\bigl(\left|k(\btheta,\cM\delta)\right|\geq\|\delta\|_E^{1/2}\bigr)\notag\\
    &\lesssim\|\delta\|_E^{m/2}+\|\delta\|_E^{1/2},\label{eq:markov-ineq}
\end{align}
where the last line follows because $\sup_{\q\in\bS^1} \mathbb{P}\bigl(\left|k(\btheta,\cM\q)\right|\bigr)\lesssim\|\delta\|_E^{m/2}$ by Assumption \ref{asm:smooth-X} and $\mathbb{P}\bigl(\left|k(\btheta,\cM\delta)\right|\geq\|\delta\|_E^{1/2}\bigr)\leq\frac{\|\cM\delta\|_E\mathbb{E}[\|\btheta\|_E]}{\|\delta\|_E^{1/2}}\lesssim\|\delta\|_E^{1/2}$ by Markov's inequality and Assumption~\ref{asm:moments}. Hence,
$\sup_{q\in\bS^1}\frac{1}{\|\delta\|_E}\mathbb{E}[R(\q,\delta)]\leq \sup_{q\in\bS^1}\frac{1}{\|\delta\|_E}\mathbb{E}[|R(\q,\delta)|]\lesssim(\|\delta\|_E^{m/2}+\|\delta\|_E^{1/2})^{1/2} \to 0$ and Eq. \eqref{eq:ss} follows from applying the law of iterated expectations.

Next, the claim that $\nabla_\q\h_\E(\q)=\ss_\E(\q)$ follows from \citet[Corollary 1.7.3]{sch93}. Finally, uniform continuity of $\ss_\E(\q)$ follows from continuity over compact $\bS^1$.
\end{proof}

\begin{proof}[Proof of Proposition \ref{prop:sh}]
By definition of $\sF$ and $\cC(e^*)$, $e^*\in\sF$ if and only if $\cC(e^*)\cap\E=\{e^*\}$. Suppose $e^*\in\sF$. Then $\text{relint}(\cC(e^*))\cap\text{relint}(\E)=\emptyset$. Since both $\cC(e^*)$ and $\E$ are nonempty convex sets, by \citet[Theorem 1.3.8]{sch93} $\cC(e^*)$ and $\E$ are properly separated. Hence, there exists $\q\in\bS^1$ and $z\in\mathbb{R}$ such that
$$
\forall\,\Tilde{e}\in \cC(e^*), \,\,\,\Tilde{e}^\intercal\q \leq z\quad\text{and}\quad\forall\, e\in\E,\,\,\, e^\intercal(-\q)\leq -z.
$$
Since $e^*\in \cC(e^*)\cap\E$, we have ${e^*}^\intercal\q=z$ and
$
\h_{\cC(e^*)}(\q)=-\h_\E(-\q)=z.
$
For the other direction, suppose there exists $\q\in\bS^1$ such that $\h_{\cC(e^*)}(\q)=-\h_\E(-\q)$. By definition, $e^*$ is feasible and $e^*\in \cC(e^*)\cap\E$.
If there exists $e'\in \cC(e^*)\cap\E$ but $e'\neq e^*$, then $\q^\intercal{e'}=\q^\intercal {e^*}=-\h_\E(-\q)$. This means that the support set of $\E$ in direction $-\q$ is not a singleton, contradicting Assumption \ref{asm:smooth-X}, under which $\E$ has a smooth boundary.
To obtain the characterization of $\sF$ in Eq.~\eqref{eq:FA_frontier_supp_func}, for $\bB^1\equiv\{\q:\Vert\q\Vert_E\le 1\}$, note that the optimization problem $\max_{\q\in\bB^1}(-\h_{\cC(e^*)}(\q)-\h_\E(-\q))$ is dual to the primal problem $\min_{\tilde{e}\in\cC(e^*),e\in\E}\Vert\tilde{e}-e\Vert_E$, which measures the distance between $\cC(e^*)$ and $\E$. When $e^*\in\E$, this distance is zero, and weak duality along with $\bS^1\subset\bB^1$ imply $\max_{\q\in\bS^1}(-\h_{\cC(e^*)}(\q)-\h_\E(-\q))\leq 0$.
Hence, when $e^*\notin\sF$ we must have $\sup_{\q\in\bS^1}(-\h_{\cC(e^*)}(\q)-\h_\E(-\q))< 0$.
Since $\h_{\cC(e^*)}(\q)=\sup_{e\in\cC(e^*)}\q^\intercal e$ is unbounded for any $\q$ such that $\q_1+\q_2<0$, it is without loss of generality to focus on $\q\in\tilde{\bS}^1$ and to use the criterion $\left[\max_{\q\in\tilde{\bS}^1}(-\h_{\cC(e^*)}(\q)-\h_\E(-\q))\right]_-=0$.
\end{proof}

\subsection{Proofs for Section \ref{sec:estimation}}
In the proofs that follow, recall 
\begin{align*}
    \zeta_i(\mathring{\cM}\q;\bvartheta)&\equiv(\mathring{\cM}\q)^\intercal\bLL_{0_i}+(\mathring{\cM}\q)^\intercal(\bLL_{1_i}-\bLL_{0_i})\cdot\mathds{1}\bigl\{k\bigl(\bvartheta(\X_i),\mathring{\cM}\q\bigr)>0\bigr\}.
\end{align*}
defined in Eq.~\eqref{eqn:influence-part}, where
$k\bigl(\bvartheta(\X_i),\mathring{\cM}\q\bigr)=\q^\intercal\mathring{\cM}\Delta\bvartheta(\X_i)$ is given in Eq.~\eqref{eq:k_Delta} for generic $\Delta\bvartheta(\X_i)\in\Theta$ and $\mathring{\cM}$. For $\widehat{\Delta\btheta}$ and $\widehat{\cM}$ estimated as per Definition \ref{def:cross-fitting}, recall
\begin{align*}
    \zeta_i(\hcM\q;\widehat{\btheta})=
    (\hcM\q)^\intercal\bLL_{0_i}+(\hcM\q)^\intercal(\bLL_{1_i}-\bLL_{0_i})\cdot\mathds{1}\bigl\{k\bigl(\widehat{\btheta}(\X_i),\hcM\q\bigr)>0\bigr\}
\end{align*}
for $k\bigl(\widehat{\btheta}(\X_i),\hcM\q\bigr)=\q^\intercal\hcM\widehat{\Delta\btheta}(\X_i)$ as in Eqs.~\eqref{eq:zeta_estimated}-\eqref{eq:k_estimated}.

\begin{proof}[Proof of Proposition \ref{prop:orthogonality}]
Apply the law of iterated expectations,
\begin{align}
    &\mathbb{E}\bigl[\zeta_i\bigl(\cM\q;\btheta+t(\bvartheta-\btheta)\bigr)\bigr]-\mathbb{E}\bigl[\zeta_i(\cM\q;\btheta)\bigr]\notag\\
    =&\mathbb{E}\small{\Bigg[\Biggl(\mathds{1}\Bigl\{(\cM\q)^\intercal(\Delta\btheta+t(\Delta\bvartheta-\Delta\btheta))\geq0\Bigr\}-\mathds{1}\Bigl\{(\cM\q)^\intercal\Delta\btheta\geq0\Bigr\}\Biggr)(\cM\q)^\intercal\Delta\btheta\Bigg]}.\label{eqn:gateux-diff}
\end{align}
In what follows, we show Eq.~\eqref{eqn:gateux-diff} is bounded in absolute value by $t(t^{m/2}+t^{1/2})^{1/2}$ for $m$ in Assumption \ref{asm:smooth-X}, uniformly in $\q\in\bS^1$. It then follows that, for all $\q\in\bS^1$, 
\begin{align*}
    \lim_{t\to0}\frac{1}{t}\left|\left(\mathbb{E}\bigl[\zeta_i\bigl(\cM\q;\btheta+t(\bvartheta-\btheta)\bigr)\bigr]-\mathbb{E}\bigl[\zeta_i(\cM\q;\btheta)\bigr]\right)\right|\lesssim \lim_{t\to0}\frac{1}{t}\cdot t(t^{m/2}+t^{1/2})^{1/2}=0.
\end{align*}
The term in Eq.~\eqref{eqn:gateux-diff} is non-zero if and only if the indicator involving $\Delta\btheta+t(\Delta\bvartheta-\Delta\btheta)$ equals $1$ and that involving $\Delta\btheta$ equals $0$, or vice versa; this happens on a subset of events
\begin{align*}
   \small\hspace{-.25cm}\biggl\{(\cM\q)^\intercal\Delta\btheta\leq0<(\cM\q)^\intercal(\Delta\btheta+t(\Delta\bvartheta-\Delta\btheta))\biggr\}\bigcup\biggl\{(\cM\q)^\intercal\Delta\btheta>0\geq (\cM\q)^\intercal(\Delta\btheta+t(\Delta\bvartheta-\Delta\btheta))\biggr\},
\end{align*}
which implies the event $\bigl\{|(\cM\q)^\intercal\Delta\btheta|<t|(\cM\q)^\intercal(\Delta\bvartheta-\Delta\btheta)|\bigr\}$. It then follows that
\begin{align*}
    |\text{Eq.}~\eqref{eqn:gateux-diff}|&\leq \mathbb{E}\left[\mathds{1}\biggl\{|(\cM\q)^\intercal\Delta\btheta|<t|(\cM\q)^\intercal(\Delta\bvartheta-\Delta\btheta)|\biggr\}\big|(\cM\q)^\intercal\Delta\btheta\big|\right]\\
    &\leq \mathbb{E}\left[\mathds{1}\biggl\{|(\cM\q)^\intercal\Delta\btheta|<t|(\cM\q)^\intercal(\Delta\bvartheta-\Delta\btheta)|\biggr\}\cdot t\big|(\cM\q)^\intercal(\Delta\bvartheta-\Delta\btheta)\big|\right] \\
    &\lesssim t(t^{m/2}+t^{1/2})^{1/2},
\end{align*}
where the last line follows from Hölder's inequality, Assumption \ref{asm:smooth-X}, and $\sup_{\Delta\bvartheta\in\Theta}\|\Delta\bvartheta-\Delta\btheta\|_{L^2(\bP)}<\infty$, using a similar argument to that of Eq. \eqref{eq:markov-ineq}.
\end{proof}

\begin{proof}[Proof of Theorem \ref{thm:gaussian}]
\textbf{Part 1.} We begin by showing that for any fixed $\q\in\bS^1$,
\begin{align}
     \sqrt{n}\biggl(\widehat{\h}_{\E}(\q; \widehat{\btheta})-\h_{\E}(\q; \btheta)\biggr)=\mathbb{G}_n[\zeta_i^*(\cM\q;\btheta)]+o_p(1),\label{eqn:emprical-process}
\end{align}
where for a generic measurable function $t\in\mathcal{T}$ of random variable $O_i$, $\mathbb{G}_n[t(O_i)]\equiv n^{-1/2}\sum_{i=1}^n \bigl(t(O_i)-\mathbb{E}[t(O_i)]\bigr)$ denotes the empirical process indexed by the function class $\mathcal{T}$. For $2\times2$ matrices $\check{\cM}$ and $\mathring{\cM}$, let
\begin{align}
    \bxi_{\ss,i}(\check{\cM};\mathring{\cM}\q;\bvartheta)\equiv \check{\cM}\bLL_0+\check{\cM}(\bLL_1-\bLL_0)\mathds{1}\{k(\bvartheta,\mathring{\cM}\q)>0\}.\label{eq:bxi}
\end{align} 
To show Eq.~\eqref{eqn:emprical-process}, fix $\q\in\bS^1$ and decompose
\begin{align*}
\sqrt{n}\biggl(\widehat{\h}_{\E}(\q; \widehat{\btheta})-\h_{\E}(\q)\biggr)
    =\underbrace{\sqrt{n}\bigg(\frac{1}{n}\sum_{i=1}^n\q^\intercal\big(\hcM\cM^{-1}\big)\big(\bxi_{\ss,i}(\cM;\hcM\q;\widehat{\btheta})-\mathbb{E}[\bxi_{\ss,i}(\cM;\cM\q;\btheta)]\big)\bigg)}_{\equiv A}&\\
    +\underbrace{\sqrt{n}\q^\intercal\big(\hcM\cM^{-1}-\bI\big)\mathbb{E}[\bxi_{\ss,i}(\cM;\cM\q;\btheta)],}_{\equiv B}&
\end{align*}
where $\bI$ is the $2\times 2$ identity matrix, and
\begin{align}
    &A=\mathbb{G}_n[\zeta_i(\cM\q;\btheta)] +o_p(1)\notag\\
    &\hspace{1cm}+\bigg\{\sqrt{n}\biggl(\frac{1}{n}\sum_{i=1}^n(\cM\q)^\intercal\big(\bLL_{0_i}+(\bLL_{1_i}-\bLL_{0_i})\mathds{1}\{k(\widehat{\btheta},\hcM\q)>0\}\big)\notag\\
    &\hspace{3cm}-\mathbb{E}\left[(\cM\q)^\intercal\big(\bLL_{0_i}+(\bLL_{1_i}-\bLL_{0_i})\mathds{1}\{k(\widehat{\btheta},\hcM\q)>0\}\big)\right]\biggr)\notag\\
    &\hspace{1cm}\underbrace{\hspace{4.5cm}-\sqrt{n}\biggl(\frac{1}{n}\sum_{i=1}^n\zeta_i(\cM\q;\btheta)-\mathbb{E}[\zeta_i(\cM\q;\btheta)]\biggr)\biggr\}}_{\equiv R_1}\label{eqn:R1}\\
    &\hspace{1cm}+\underbrace{\sqrt{n}\biggl(\mathbb{E}\left[(\cM\q)^\intercal\big(\bLL_{0_i}+(\bLL_{1_i}-\bLL_{0_i})\mathds{1}\{k(\widehat{\btheta},\hcM\q)>0\}\big)\right]-\mathbb{E}[\zeta_i(\cM\q;\btheta)]\biggr)}_{\equiv R_2},\label{eqn:R2}
\end{align}
where the equality follows from $\|\hcM\cM-\bI\|_{\max}=O_p(n^{-1/2})$ under the theorem's assumptions and adding and subtracting terms and $R_1$ is an empirical process term indexed by $\Delta\zeta_i\big(\q;\widehat{\btheta},\hcM\big)\equiv(\cM\q)^\intercal(\bLL_{1_i}-\bLL_{0_i})\mathds{1}\big\{k\big(\widehat{\btheta},\hcM\q\big)>0\big\}-(\cM\q)^\intercal(\bLL_{1_i}-\bLL_{0_i})\mathds{1}\big\{k\big(\btheta,\cM\q\big)>0\big\}$. Recall that $k(\widehat{\btheta} ,\hcM\q)=\q^\intercal\hcM\widehat{\Delta\btheta}$, where $\widehat{\Delta\btheta}(\X_i)=\widehat{\Delta\btheta}_k(\X_i)$ for $i\in I_k$.
Fixing any $k\in[K]$ and conditional on the $I_k^c$ sample, $\Delta\widehat{\btheta}_k$ is non-stochastic and 
\begin{align}
    \mathbb{E}[R_1^2~|~ I_k^c]   &\lesssim\sup_{\Delta\bvartheta\in\Theta_n}\mathbb{E}\biggl[\Delta\zeta_i\big(\q;\widehat{\btheta},\hcM\big)^2\biggr]\notag\\
    &=\sup_{\Delta\bvartheta\in\Theta_n}\mathbb{E}\left[\big((\cM\q)^\intercal(\bLL_{1_i}-\bLL_{0_i})\big)^2\left|\mathds{1}\{\q^\intercal\hcM\Delta\bvartheta>0\}-\mathds{1}\{\q^\intercal\cM\Delta\btheta>0\}\right|\right]\notag\\
    &\lesssim\sup_{\Delta\bvartheta\in\Theta_n}\mathbb{E}\left[\big((\cM\q)^\intercal(\bLL_{1_i}-\bLL_{0_i})\big)^2\mathds{1}\big\{|\q^\intercal\cM\Delta\btheta|<|\q^\intercal(\hcM\Delta\bvartheta-\cM\Delta\btheta)|\big\}\right]\notag\\
    &\lesssim\sup_{\Delta\bvartheta\in\Theta_n}\mathbb{P}\left(|\q^\intercal\cM\Delta\btheta|<|\q^\intercal(\hcM\Delta\bvartheta-\cM\Delta\btheta)|\right)=o_p(1),\label{eq:R1op1}
\end{align}
where the first inequality follows from \citet[Proof of Theorem 3.1 and Lemma 6.1]{che:che:dem:duf:han:whi:rob18}, the third inequality follows by Assumption~\ref{asm:moments}, 
and the last equality follows from the definition of $\Theta_n$, $\|\hcM-\cM\|_{\max}=O_p(n^{-1/2})$, and a similar argument to that of Eq. \eqref{eq:markov-ineq}.
Hence $|R_1|=o_p(1)$. In addition, by Proposition \ref{prop:orthogonality}, $R_2$ encapsulates higher-order Gateaux derivatives and can be bounded by 
\begin{align*}
    |R_2|\lesssim\sqrt{n} \underbrace{\mathbb{E}\left[\mathds{1}\biggl\{|\q^\intercal\cM\Delta\btheta|<|\q^\intercal(\hcM\widehat{\Delta\btheta}-\cM\Delta\btheta)|\biggr\}|\q^\intercal\cM\Delta\btheta|\right]}_{R_3},
\end{align*}

\vspace{-.2cm}
\noindent and we show $|R_3|=o_p(n^{-1/2})$ by leveraging the proof technique from \citet[Lemma 4.1]{sem23} under Assumption \ref{asm:nuisance-structure}.  It then follows that $|R_2|=o_p(1)$.

Under Assumption \ref{asm:nuisance-structure}, for $g\in\{r,b\}$ and the $k$-th fold,
\begin{align*}
    &\bigg|\frac{(\widehat{\Delta\theta}^g)_k(\X)}{\widehat{\mu}_g}-\frac{\Delta\theta^g(\X)}{\mu_g}\bigg|\\
    \lesssim &\bigg|\frac{(\widehat{\alpha}^g)_k}{\widehat{\mu}_g}-\frac{\alpha^g}{\mu_g}\bigg|+\bigg|\frac{(\widehat{\beta}^g)_k}{\widehat{\mu}_g}-\frac{\beta^g}{\mu_g}\bigg|+\bigg|\frac{(\widehat{\eta}^g)_k(\X_{[3:d_\X]})}{\widehat{\mu}_g}-\frac{\eta^g(\X_{[3:d_\X]})}{\mu_g}\bigg|\equiv\delta_k^g(\X_{[3:d_\X]}),
\end{align*}
where $\lesssim$ follows from the assumption that $(\X_1, X_2)$ has bounded support. 
Let $\delta_k(\X_{[3:d_\X]})\equiv\sum_{g\in\{r,b\}}\delta_k^g(\X_{[3:d_\X]})$.
Denote the distribution of 
$|\q^\intercal\cM\Delta\btheta(\X)|$ conditional on $\X_{[3:d_\X]}$ by $\bP_{\bigl(|\q^\intercal\cM\Delta\btheta|~\bigl|\X_{[3:d_\X]}\bigr)}$. Conditional on $\X$ and the sample $I_k^c$, $|\q^\intercal(\hcM\widehat{\Delta\btheta}_k-\cM\Delta\btheta)|\leq\|\hcM\widehat{\Delta\btheta}_k-\cM\Delta\btheta\|_E\lesssim\delta_k(\X_{[3:d_\X]})$, and
\begin{align*}
    \small|R_3|\lesssim\mathbb{E}_{\X_{[3:d_\X]}}\Bigg[\int_{-\delta_k(\X_{[3:d_\X]})}^{\delta_k(\X_{[3:d_\X]})}\delta_k(\X_{[3:d_\X]})\bP_{\bigl(|q^\intercal\cM\Delta\btheta|~\bigl|\X_{[3:d_\X]}\bigr)}\Biggr]\lesssim\mathbb{E}_{\X_{[3:d_\X]}}\big[\delta_k(\X_{[3:d_\X]})^2\big]=o_p(n^{-1/2}),
\end{align*}
where the second inequality follows from Assumption \ref{asm:nuisance-structure} that $\bP_{\bigl(|q^\intercal\cM\Delta\btheta|~\bigl|\X_{[3:d_\X]}\bigr)}$ is bounded, and the last equality follows from $\widehat{\Delta\btheta}_k\in\Theta_n$ as $n\to\infty$, $\|\hcM-\cM\|_{\max}=O_p(n^{-1/2})$, and repeated application of triangle inequality. Therefore,
\begin{align*}
A=\mathbb{G}_n[\zeta_i(\cM\q;\btheta)]+o_p(1).
\end{align*}
In addition,
\begin{align*}
    B=\sqrt{n}\big((\hcM-\cM)\q\big)^\intercal\cM^{-1}\ss_\E(\q)=\mathbb{G}_n[(\cM_i^*\q)^\intercal\cM^{-1}\ss_\E(\q)]+o_p(1)
\end{align*}
for $\cM_i^*\equiv\diag\left(\frac{\mathds{1}\{\G_i=r\}}{-\mu_r^2}, \frac{\mathds{1}\{\G_i=b\}}{-\mu_b^2}\right)$ by the Delta method. We hence conclude
\begin{align*}
    \sqrt{n}\biggl(\widehat{\h}_{\E}(\q; \widehat{\btheta})-\h_{\E}(\q; \btheta)\biggr)=\mathbb{G}_n[\zeta_i^*(\cM\q;\btheta)]+o_p(1)
\end{align*}
for $\zeta_i^*(\cM\q;\btheta)\equiv\zeta_i(\cM\q;\btheta)+(\cM_i^*\q)^\intercal\cM^{-1}\ss_\E(\q)$.

\noindent\textbf{Part 2.} Since the class of functions over the random variables $(\bLL_{1_i}, \bLL_{0_i}, \X_i)$,
\begin{align*}
    \mathcal{Z}\equiv\bigl\{\zeta_i^*(\cM\q;\btheta): \q\in\bS^1\bigr\},
\end{align*}
is composed of linear functions (linear in $\q\in\bS^1$) and their indicators, $\mathcal{Z}$ is VC-subgraph \citep[see, for example,][]{and94}. An envelope function of $\mathcal{Z}$ is
\begin{align}
    \sum_{d\in\{0,1\},g\in\{r,b\}}(|\LL_{d_i}^g|+C)/c_1.\label{eqn:envelope}
\end{align}
As $C$ and $c_1$ are constant and $\LL_{d_i}^g$ is square-integrable, Eq.~\eqref{eqn:envelope} is square-integrable under $\bP$. By \citet[Theorem 2.5.2]{van:wel13}, $\mathcal{Z}$ is $\bP$-Donsker, and hence
\begin{align*}
    \sqrt{n}\biggl(\widehat{\h}_{\E}(\q; \widehat{\btheta})-\h_{\E}(\q)\biggr)=\mathbb{G}[  \zeta_i^*(\cM\q;\btheta)]+o_p(1) \quad\text{ in }\,\, \ell^\infty(\bS^1).
\end{align*}
\noindent\textbf{Part 3.}
The fact that $Var\left(\mathbb{G}[\zeta_i(\q;\btheta)]\right)>0$ for each $\q\in\bS^1$ follows using the variance decomposition formula and the same argument as in \citet[proof of Theorem 4.3-(ii), p.808]{ber:mol08}.  
\end{proof}

\subsection{Proofs for Section \ref{sec:set_estimation}}

\begin{proof}[Proof of Proposition \ref{prop:consistency_coverage_sF}]
Recalling the definition of the set $Q_{\mathbb{T}}^*(e),~\mathbb{T}\in\{\bS^1,\tilde{\bS}^1\}$ in Proposition~\ref{prop:PFlimit_distr}, by the same argument as in the proof of Proposition~\ref{prop:PFlimit_distr} below,
\begin{align*}
    \left[\sqrt{n}\max_{\q\in\bS^1}(\q^\intercal e-\widehat{\h}_\E(\q;\widehat{\btheta}))\right]_{+}\xrightarrow[]{d}\left[\sup_{\q\in Q^*_{\bS^1}(e)}\mathbb{G}[-{\zeta_i^*(\cM\q;\btheta)}]\right]_+.
\end{align*}
Applying the argument in the proof of Proposition~\ref{prop:lda} to the second maximization problem in the definition of $T_n^\sF$ in Eq.~\eqref{eq:Tn_sF}, we have
\begin{align}
 \psi^\sF(e)=\left[\sup_{\q\in Q^*_{\bS^1}(e)}\mathbb{G}[-{\zeta_i^*(\cM\q;\btheta)}]\right]_+ + \left[\sup_{\q\in Q^*_{\tilde{\bS}^1}(e)}\mathbb{G}[-{\zeta_i^*(-\cM\q;\btheta)}]\right]_-. \label{psi-FA}
\end{align}
When Eq.~\eqref{eq:no_kinks_cond} holds, $\E$ has no kinks or flat faces (by Assumption~\ref{asm:smooth-X}); hence under the null:
\begin{align*}
    Q^*_{\bS^1}(e)=\arg\max_{\q\in\bS^1}\left(\q^\intercal e-\h_\E(\q)\right)=\left\{q^*_{\bS^1}(e)\right\}, ~Q^*_{\tilde{\bS}^1}(e)=\arg\max_{\q\in\tilde{\bS}^1}\left(-\h_{\cC(e)}(\q)-\h_\E(-\q)\right)\equiv\{\q^*_{\tilde{\bS}^1}(e)\}.
\end{align*}
Under the null $e\in\sF$, $(\q^*_{\bS^1}(e))^\intercal e-\h_\E(\q^*_{\bS^1}(e))=0$ and $e=\ss_\E(\q^*_{\bS^1}(e))$; since $e\in\sF$, $-(\q^*_{\tilde{\bS}^1}(e))^\intercal e-\h_\E(-\q^*_{\tilde{\bS}^1}(e))=0$, so that $e=\ss_\E(-\q^*_{\tilde{\bS}^1}(e))$. By Eq.~\eqref{eq:no_kinks_cond} kinks are absent, and hence $\q^*_{\bS^1}(e)=-\q^*_{\tilde{\bS}^1}(e)$. 
We therefore obtain the expression in Eq.~\eqref{psi_FA_no_kinks}, as
\begin{align*}
    \psi^\sF(e)=\left[\mathbb{G}[-{\zeta_i^*(\cM\q^*_{\bS^1}(e);\btheta)}]\right]_+ + \mathbb{G}\left[-{\zeta_i^*(\cM\q^*_{\bS^1}(e);\btheta)}\right]_-
    =\left|\mathbb{G}\left[{\zeta_i^*(\cM\q^*_{\bS^1}(e);\btheta)}\right]\right|.
\end{align*}
Eq.~\eqref{eq:validity_CS(F)} follows by standard arguments.

We establish Eq.~\eqref{eq:consistency_F_hat} by verifying Condition C.1 in \citet{che:hon:tam07}, under which Eq.~\eqref{eq:consistency_F_hat} follows from their Theorem 3.1-(1).
By definition, the parameter space $\bB_C$ is a compact (and convex) set.
Our criterion function is
\begin{align*}
    \mathsf{f}(e)\equiv\left[\max_{\q\in\bS^1}(\q^\intercal e-\h_\E(\q))\right]_{+}+\left[\max_{\q\in\Tilde{\bS}^1}(-\h_{\cC(e)}(\q)-\h_\E(-\q))\right]_-,
\end{align*}
and the population set is $\sF=\{e\in\bB_C:~\mathsf{f}(e)=0\}$. 
The criterion function $\mathsf{f}(e)$ is lower-semicontinuous by Berge's Maximum Theorem and composition with a continuous function.
The sample criterion function
\begin{align*}
    \widehat{\mathsf{f}}(e)\equiv
\left[\max_{\q\in\bS^1}(\q^\intercal e-\widehat{\h}_\E(\q;\widehat{\btheta}))\right]_{+}+\left[\max_{\q\in\tilde\bS^1}(-\h_{\cC(e)}(\q)-\hn_\E(-\q;\widehat{\btheta}))\right]_-
\end{align*}
takes values in $\mathbb{R}_+$ and is jointly measurable in the parameter $e\in\E$ and the data by standard arguments.
Finally, by Proposition~\ref{prop:lda}, $\sup_{e\in\bB_C}\left(\mathsf{f}(e)-\widehat{\mathsf{f}}(e)\right)_+=O_p(1/\sqrt{n})$ and  $\sup_{e\in\sF}~\widehat{\mathsf{f}}(e)=O_p(1/\sqrt{n})$.
\end{proof}

\begin{proof}[Proof of Proposition \ref{prop:PFconsistency}]
Recall that by Proposition \ref{prop:ss}, $\ss_\E(\q)=\mathbb{E}\bigl[\bzeta_{\ss}(\cM\q;\btheta)\bigr]$ for $\bzeta_{\ss}(\cM\q;\btheta)\equiv\cM\bLL_0+\cM(\bLL_1-\bLL_0)\mathds{1}\{k(\btheta,\cM\q)>0\}$.
Using the notation $\bxi_{\ss,i}(\check{\cM};\mathring{\cM}\q;\bvartheta)$ defined in Eq.~\eqref{eq:bxi}, we have:
\begin{subequations}
\begin{align}
    &\left\Vert \widehat{\ss}_\E(\q;\widehat{\btheta})-\ss_\E(\q)\right\Vert=   
    \left\Vert\frac{1}{n}\sum_{i=1}^n\bzeta_{\ss,i}(\hcM\q;\widehat{\btheta})-\ss_\E(\q)\right\Vert\notag\\
    \le&\left\Vert \hcM\cM^{-1}\right\Vert
    \Bigg\Vert \frac{1}{n}\sum_{i=1}^n\bxi_{\ss,i}(\cM;\hcM\q;\widehat{\btheta})-\mathbb{E}\left[\bxi_{\ss,i}(\cM;\cM\q;\btheta)\right]\Bigg\Vert\notag\\
    &\hspace{6cm}+\left\Vert \left(\hcM\cM^{-1}-\bI\right) \mathbb{E}[\bzeta_{\ss,i}(\cM\q;\btheta)]\right\Vert\notag\\
    \le&\Bigg\Vert \frac{1}{n}\sum_{i=1}^n\bxi_{\ss,i}(\cM;\cM\q;\btheta)-\mathbb{E}\left[\bxi_{\ss,i}(\cM;\cM\q;\btheta)\right]\Bigg\Vert\label{eq:consistency_hatS1}\\
    &+\Bigg\Vert \left(\frac{1}{n}\sum_{i=1}^n\bxi_{\ss,i}(\cM;\hcM\q;\widehat{\btheta})-\mathbb{E}\left[\bxi_{\ss,i}(\cM;\hcM\q;\widehat{\btheta})\right]\right)\notag\\
    &\hspace{1.5cm}-\left(\frac{1}{n}\sum_{i=1}^n\bxi_{\ss,i}(\cM;\cM\q;\btheta)-\mathbb{E}\left[\bxi_{\ss,i}(\cM;\cM\q;\btheta)\right]\right)\Bigg\Vert\label{eq:consistency_hatS2}\\
    &\hspace{1.75cm}+\left\Vert \frac{1}{n}\sum_{i=1}^n\mathbb{E}\left[\bxi_{\ss,i}(\cM;\hcM\q;\widehat{\btheta})\right]-\mathbb{E}\left[\bxi_{\ss,i}(\cM;\cM\q;\btheta)\right]  \right\Vert+o_p(1),\label{eq:consistency_hatS3}
\end{align}
\end{subequations}
where the $o_p(1)$ term in line~\eqref{eq:consistency_hatS3} follows as $\Vert \hcM\cM^{-1}-\bI\Vert_{\max}=O_p(n^{-1/2})$.
The term in Eq.~\eqref{eq:consistency_hatS1} equals $o_p(1)$ by the Law of Large Numbers.
Using that $\mathbb{E}\left[\cM(\bLL_1-\bLL_0)|\X\right]=\cM(\btheta_1(\X)-\btheta_0(\X))=[-\mathfrak{u}_1^\intercal\cM
\Delta\btheta(\X),~-\mathfrak{u}_2^\intercal\cM
\Delta\btheta(\X)]^\intercal$ and omitting dependence on $\X$ to shorten notation, we have that the first term of Eq.~\eqref{eq:consistency_hatS3} is upper bounded by
\begin{align*}
    &\sup_{\Delta\bvartheta\in\Theta_n}\Big\Vert\mathbb{E}\Bigl[\cM(\bLL_{1_i}-\bLL_{0_i})\left(\mathds{1}\bigl\{\q^\intercal\hcM\Delta\bvartheta>0\bigr\}-\mathds{1}\bigl\{\q^\intercal\cM\Delta\btheta>0\bigr\}\right)\Bigr]\Big\Vert\\
    =&\sup_{\Delta\bvartheta\in\Theta_n}\Big\Vert\mathbb{E}\Bigl[[-\mathfrak{u}_1^\intercal\cM
\Delta\btheta,~-\mathfrak{u}_2^\intercal\cM
\Delta\btheta]^\intercal\left(\mathds{1}\bigl\{\q^\intercal\hcM\Delta\bvartheta>0\bigr\}-\mathds{1}\bigl\{\q^\intercal\cM\Delta\btheta>0\bigr\}\right)\Bigr]\Big\Vert\\
    \lesssim &\sup_{\Delta\bvartheta\in\Theta_n}\max_{\v\in\{-\mathfrak{u}_1,-\mathfrak{u}_2\}}\Big|\mathbb{E}\Bigl[\v^\intercal\cM
\Delta\btheta\left(\mathds{1}\bigl\{\q^\intercal\hcM\Delta\bvartheta>0\bigr\}-\mathds{1}\bigl\{\q^\intercal\cM\Delta\btheta>0\bigr\}\right)\Bigr]\Big|\\
    \le &\sup_{\Delta\bvartheta\in\Theta_n}\max_{\v\in\{-\mathfrak{u}_1,-\mathfrak{u}_2\}}\mathbb{E}\Bigl[\big|\v^\intercal\cM\Delta\btheta\big|\mathds{1}\bigl\{\big|\q^\intercal\cM\Delta\btheta\big|<\big|\q^\intercal\hcM\Delta\bvartheta-\q^\intercal\cM\Delta\btheta\big|\bigr\}\Bigr]\\
    \lesssim &\sup_{\Delta\bvartheta\in\Theta_n}\mathbb{E}\left[\mathds{1}\bigl\{\big|\q^\intercal\cM\Delta\btheta\big|<\big|\q^\intercal\hcM\Delta\bvartheta-\q^\intercal\cM\Delta\btheta\big|\bigr\}\right]^{1/2}=o_p(1)
\end{align*}
where the last inequality follows by Cauchy-Schwartz and Assumption~\ref{asm:moments}, 
and the last equality follows from Assumption \ref{asm:smooth-X}, the fact that $\Delta\bvartheta\in\Theta_n$, and that under the proposition's assumptions $\Vert \hcM\cM^{-1}-\bI\Vert_{\max}=O_p(n^{-1/2})$; one can then show Eq.~\eqref{eq:consistency_hatS3} is $o_p(1)$ by the same argument as that used to establish Eq.~\eqref{eq:R1op1} in the proof of Theorem~\ref{thm:gaussian}.

We conclude by establishing Eq.~\eqref{eq:consistency_PF_hat}.
Using the definition of Hausdorff distance,
\begin{align*}
    \mathbf{d}_H(\widehat{\sPF},\sPF)=\max\left\{\sup_{\hat{e}\in\widehat{\sPF}}\inf_{e\in\sPF}\left\Vert\hat{e}-e\right\Vert, \sup_{e\in\sPF}\inf_{\hat{e}\in\widehat{\sPF}}\left\Vert\hat{e}-e\right\Vert\right\}
\end{align*}
Suppose by contradiction that there is $e^*\in\sPF$ such that for all $n\ge 1$, $\inf_{\hat{e}\in\widehat{\sPF}}\left\Vert\hat{e}-e^*\right\Vert>c$ for some constant $c>0$. By definition, there exists $\q^*\in\bQ$ such that $\ss_\E(\q^*)=e^*$. By Eq.~\eqref{eq:consistency_S_hat}, $\Vert\widehat{\ss}_\E(\q^*;\widehat{\btheta})-e^*\Vert=o_p(1)$, and by definition $\widehat{\ss}_\E(\q^*;\widehat{\btheta})\in\widehat{\sPF}$, yielding a contradiction. The same argument holds by swapping the role of $\widehat{\sPF}$ and $\sPF$.
\end{proof}

\begin{proof}[Proof of Proposition \ref{prop:PFlimit_distr}]
Our proof follows arguments in \citet[Theorem 3.4]{kai16}. 
We first observe that for $e\in\sPF$, $\max_{\q\in\bS^1}(\q^\intercal e - \h_\E(\q)) = 0$ and $\max_{\q\in\bQ}(\q^\intercal e - \h_\E(\q))=0$.
Hence, for $\mathbb{T}\in\{\bS^1,\bQ\}$ and $\phi_{e,\mathbb{T}}(f)\equiv\sup_{\q\in\mathbb{T}}(\q^\intercal e-f(\q))$, we can write
\begin{multline*}
    T_n^\sPF(e) =\max\left\{\sqrt{n}\left(\phi_{e,\bS^1}\big(\hn_\E(\q;\widehat{\btheta})\big)-\phi_{e,\bS^1}\big(\h_\E(\q)\big)\right),0\right\} \\
    - \min\left\{\sqrt{n}\left(\phi_{e,\bQ}\big(\hn_\E(\q;\widehat{\btheta})\big)-\phi_{e,\bQ}\big(\h_\E(\q)\right),0\right\}.
\end{multline*}
\citet[Lemma D.3]{kai16} establishes that $\phi_{e,\mathbb{T}}$ is Hadamard directionally differentiable at $\h_\E(\cdot)$ with Hadamard directional derivative equal to $\phi^\prime_{e,\mathbb{T}}(y)=\sup_{\q\in Q^*_\mathbb{T}(e)}-y(\q)$.  
By Theorem~\ref{thm:gaussian}, the assumptions in \citet[Lemma D.4]{kai16} are satisfied, and the result follows by the Continuous Mapping Theorem and as argued in \citet[proof of Theorem 3.4]{kai16}.
Absolute continuity of the limit law in Eq.~\eqref{eq:Tn_PF_limit} follows using the same argument as in \citet[proof of Theorem 4.3-(ii), p.808]{ber:mol08}.  
\end{proof}

\begin{proof}[Proof of Proposition~\ref{prop:algorithm_consistency}]
    Take a point $e^0$ in the set $\{\tilde{e}: \Vert \tilde{e}\Vert=2C~\text{and}~\tilde{e}_r+\tilde{e}_b\le 0\}$. Select $\widehat{e}_{n_1}$ as the metric projection of $e^0$ onto $\widehat{\sF}$ and let $e$ be the metric projection of $e^0$ onto $\sF$.
    By the consistency result in Proposition~\ref{prop:consistency_coverage_sF} and by \citet[proof of Theorem 1.7.19]{mol17}, $\Vert \widehat{e}_{n_1} - e\Vert \xrightarrow[]{p} 0$ as $n_1\to\infty$.
    Hence, given that by Theorem~\ref{thm:gaussian} the support function estimator converges to the population support function uniformly in $\q\in\bS^1$, and given that $\widehat{\q}^*_{n_1}(\widehat{e}_{n_1})$ is an extremum estimator and all the conditions for its consistency are satisfied, $\Vert \widehat{\q}^*_{n_1}(\widehat{e}_{n_1})-\q^*_{\bS^1}(e)\Vert \xrightarrow[]{p} 0$ as $n_1\to\infty$.
    Next, using the same argument as in the proof of Proposition~\ref{prop:PFconsistency} leading to Eqs.~\eqref{eq:consistency_hatS1}-\eqref{eq:consistency_hatS3}, 
    and $\bxi_{\ss,i}(\check{\cM};\mathring{\cM}\q;\bvartheta)$ defined in Eq.~\eqref{eq:bxi},
    \begin{subequations}
    \begin{align}
        \Big\Vert\widehat{\ss}_\E(\widehat{\q}^*_{n_1}(\widehat{e}_{n_1});\widehat{\btheta}_{n_1}) - \ss_\E(\q^*_{\bS^1}(e))\Big\Vert\hspace{7.5cm}& \notag\\
        = \left\Vert \frac{1}{n_2}\sum_{i=1}^{n_2}\bzeta_{\ss,i}(\hcM\widehat{\q}^*_{n_1}(\widehat{e}_{n_1});\widehat{\btheta}_{n_1})-\mathbb{E}[\bzeta_{\ss,i}(\cM\q^*_{\bS^1}(e);\btheta)]\right\Vert&\notag\\
        \le o_p(1)+\Bigg\Vert \frac{1}{n_2}\sum_{i=1}^{n_2}\bxi_{\ss,i}(\cM;\cM\q^*_{\bS^1}(e);\btheta)-\mathbb{E}\left[\bxi_{\ss,i}(\cM;\cM\q^*_{\bS^1}(e);\btheta)\right]\Bigg\Vert&\label{eq:consistent_alg1} \\
        +\left\Vert \left(\frac{1}{n_2}\sum_{i=1}^{n_2}\bxi_{\ss,i}(\cM;\hcM\widehat{\q}^*_{n_1}(\widehat{e}_{n_1});\widehat{\btheta}_{n_1})-\mathbb{E}\left[\bxi_{\ss,i}(\cM;\hcM\widehat{\q}^*_{n_1}(\widehat{e}_{n_1});\widehat{\btheta}_{n_1})\right]\right)\right. &\notag\\
        \left.-\left(\frac{1}{n_2}\sum_{i=1}^{n_2}\bxi_{\ss,i}(\cM;\cM\q^*_{\bS^1}(e);\btheta)-\mathbb{E}\left[\bxi_{\ss,i}(\cM;\cM\q^*_{\bS^1}(e);\btheta)\right]\right)\right\Vert&\label{eq:consistent_alg2} \\
        +\left\Vert \mathbb{E}\left[\bxi_{\ss,i}(\cM;\hcM\widehat{\q}^*_{n_1}(\widehat{e}_{n_1});\widehat{\btheta}_{n_1})\right]-\mathbb{E}\left[\bxi_{\ss,i}(\cM;\cM\q^*_{\bS^1}(e);\btheta)\right]  \right\Vert&.\label{eq:consistent_alg3} 
    \end{align}
    \end{subequations}
    By the same argument as in the proof of Proposition~\ref{prop:PFconsistency}, the terms in Eqs.~\eqref{eq:consistent_alg1}-\eqref{eq:consistent_alg2} are $o_p(1)$. 
    We next show that the same holds for the term in Eq.~\eqref{eq:consistent_alg3}.
    Take any $\delta>0$,
    \begin{align*}
        &\Big\Vert \mathbb{E}\left[\bxi_{\ss,i}(\cM;\hcM\widehat{\q}^*_{n_1}(\widehat{e}_{n_1});\widehat{\btheta}_{n_1})\right]-\mathbb{E}\left[\bxi_{\ss,i}(\cM;\cM\q^*_{\bS^1}(e);\btheta)\right]\Big\Vert \\
        \lesssim &~\sup_{\Delta\bvartheta\in\Theta_n}\max_{\v\in\{-\mathfrak{u}_1,\mathfrak{u}_2\}}\mathbb{E}\Bigl[\big|\v^\intercal\cM\Delta\btheta\big|\mathds{1}\bigl\{\big|
        \q^*_{\bS^1}(e)^\intercal\cM\Delta\btheta \big|<\big|
        \widehat{\q}^*_{n_1}(\widehat{e}_{n_1})^\intercal\hcM\Delta\bvartheta-\q^*_{\bS^1}(e)^\intercal\cM\Delta\btheta \big|\bigr\}\Bigr]\\
        \lesssim &~\sup_{\Delta\bvartheta\in\Theta_n}\mathbb{E}\Bigl[\mathds{1}\bigl\{\big|\q^*_{\bS^1}(e)^\intercal\cM\Delta\btheta\big|<\big|\widehat{\q}^*_{n_1}(\widehat{e}_{n_1})^\intercal\hcM\Delta\bvartheta-\q^*_{\bS^1}(e)^\intercal\cM\Delta\btheta\big|\bigr\}\Bigr]^{1/2}\\
        \leq&\left(\mathbb{E}\Bigl[\mathds{1}\bigl\{\big|\q^*_{\bS^1}(e)^\intercal\cM\Delta\btheta\big|<\delta\big\}\Bigr]+\mathbb{E}\Bigl[\mathds{1}\bigl\{\Vert\hcM\widehat{\q}^*_{n_1}(\widehat{e}_{n_1})-\cM\q^*_{\bS^1}(e)\Vert>\delta\big\}\Bigr]\right)^{1/2}=o_p(1),
    \end{align*} 
using Assumption~\ref{asm:smooth-X}, $\Vert \hcM\cM^{-1}-\bI\Vert_{\max}=O_p(n^{-1/2})$, Markov inequality, and that $\delta>0$ is arbitrary.
\end{proof}

\subsection{Proofs for Section \ref{sec:test}}
\begin{proof}[Proof of Proposition \ref{prop:weak-GB}]
By the same argument as in the proof of Proposition~\ref{prop:PFlimit_distr} and the discussion in Section~\ref{subsec:pareto_frontier},
\begin{align}
    \liminf_{n\to\infty}\bP\big\{(\R,\B)\in \mathcal{CS}_n(\R,\B)\big\}\geq 1-\alpha.\label{eq:CS-validity}
\end{align} 
To obtain the result in Eq.~\eqref{eq:testRB-validity}, observe that
 under the null in Eq.~\eqref{eq:null-weak-GB}, by Eq.~\eqref{eq:CS-validity},
\begin{align*}
    \mathbb{E}\left[\varphi_n^{\texttt{skew}}\right]&=1-\mathbb{\bP}\bigg\{\sup_{(\Tilde{\R}, \Tilde{\B})\in \mathcal{CS}_n(\R,\B)}\big((\mathfrak{u_1}-\mathfrak{u_2})^\intercal \Tilde{\R}\big)\big((\mathfrak{u_1}-\mathfrak{u_2})^\intercal \Tilde{\B}\big)\geq0\bigg\}\\
    &\leq 1-\mathbb{\bP}\big\{(\R,\B)\in \mathcal{CS}_n(\R,\B)\big\}\leq \alpha.
\end{align*}
\end{proof}

\begin{proof}[Proof of Proposition \ref{prop:lda}]
For $g\in\{r,b\}$, recall $Z_i^g$ defined in Eq.~\eqref{ehat} and note that
\begin{align*}
\sqrt{n}\left(\widehat{e}_g^*-e_g^*\right)=&\sqrt{n}\left(\frac{1}{n}\sum_{i=1}^n \frac{Z_i^g}{\widehat{\mu}_g}-\frac{\mathbb{E}[Z_i^g]}{\mu_g}\right)\\
=&\frac{1}{\mu_g}\mathbb{G}_n[Z_i^g]-\frac{\mathbb{E}[Z_i^g]}{\mu_g^2}\mathbb{G}_n[\mathds{1}\{G_i=g\}]+ o_p(1)
=\mathbb{G}_n[Z_i^{g,*}]+o_p(1),
\end{align*}
where 
\vspace{-.7cm}
\begin{align}
    Z_i^{g,*}\equiv\frac{Z_i^g}{\mu_g}-\frac{\mathbb{E}[Z_i^g]}{\mu_g^2}\mathds{1}\{G_i=g\}.\label{ehat*}
\end{align}
Since $\max_{d\in\{0,1\},g\in\{r,b\}}\mathbb{E}[(\LL_d^g)^2]<\infty$, $Z_i^{g,*}$ has finite first and second moments. By the Lindeberg–Lévy central limit theorem, we have $\sqrt{n}\left(\widehat{e}_g^*-e_g^*\right)=\mathbb{G}[Z_i^{g,*}]+o_p(1)$, where $\mathbb{G}[Z_i^{g,*}]$ is a mean-zero Gaussian random variable with variance $\mathbb{E}\bigl[(Z_i^{g,*})^2\bigr]$. 

By Theorem \ref{thm:gaussian} and the Cramér–Wold theorem, we have that jointly,
\begin{align}
\sqrt{n}\begin{bmatrix}
	\hn_\E(\q; \widehat{\btheta})-\h_\E(\q)\\
	\widehat{e}_r^*-e_r^*\\
	\widehat{e}_b^*-e_b^*
\end{bmatrix}\xrightarrow[]{d}\begin{bmatrix}
	\mathbb{G}[\zeta_i^*(\cM\q;\btheta)]\\
	\mathbb{G}[Z_i^{r,*}]\\
	\mathbb{G}[Z_i^{b,*}]
\end{bmatrix}\equiv\mathbb{G}_{he^*} \quad \text{ in }\ell^\infty(\bB^1),\label{eq:e-h-joint}
\end{align}
Next, we analyze the two parts of the test statistic in Eq.~\eqref{eq:Tn_LDA}, beginning with the first part. By the same argument as in the proof of Proposition \ref{prop:PFlimit_distr}, under the null that $e^*\in\sF$ we have $\max_{\q\in\bS^1}(\q^\intercal e^* - \h_\E(\q)) = 0$.
Hence, for $\phi_{\bS^1}\big(e,f(\cdot)\big)\equiv\sup_{\q\in\bS^1}\big(\q^\intercal e-f(\q)\big)$, we can write $\sqrt{n}\left[\max_{\q\in\bS^1}(\q^\intercal \widehat{e}^*-\widehat{\h}_\E(\q;\widehat{\btheta}))\right]_{+}=\max\left\{\sqrt{n}\left(\phi_{\bS^1}\big(\widehat{e}^*,\hn_\E(\q;\widehat{\btheta})\big)-\phi_{\bS^1}\big(e^*,\h_\E(\q)\big)\right),0\right\}$.
\citet[Lemma D.3]{kai16} shows that $\phi_{\bS^1}$ is Hadamard directionally differentiable at $\big(e^*,\h_\E(\cdot)\big)$ with derivative $\phi^\prime_{\bS^1,(e^*,\h_\E(\cdot))}\big(e,f(\cdot)\big)=\sup_{\q\in Q^*_{\bS^1}(e^*)}\big(\q^\intercal e-f(\q)\big)$.  

For the second part of the test statistic in Eq.~\eqref{eq:Tn_LDA}, note 
that for any $s,t\in\mathbb{R}$, $\min\{s,2t-s\}=(2t-s)-\max\{2(t-s),0\}$ and $\min\{t,2s-t\}=t-\max\{2(t-s),0\}$,
so we can plug in $t=e_r^*$ and $s=e_b^*$ to rewrite Eq.~\eqref{eqn:sf-C} as
\begin{align*}
&\h_{\sC}(\q)=\max\biggl\{\q_1(e_r^*-\max\{2(e_r^*-e_b^*),0\})+\q_2 e_b^*,\,\,\, \q_1 e_r^*+\q_2\bigl(2e_r^*-e_b^*-\max\{2(e_r^*-e_b^*),0\}\bigr)\biggr\}\\
&=\max\biggl\{2\q_2\bigr(e_r^*-e_b^*\bigr)+(\q_1-\q_2)\max\{2(e_r^*-e_b^*),0\},\,\,\,0 \biggr\}+\q_1 e_r^* + \q_2 e_b^* - \q_1\max\{2(e_r^*-e_b^*),0\}.
\end{align*}
It follows that we can write
\vspace{-.25cm}
\begin{align}
\mathfrak{\h}_1\big(\h_\E(\cdot), e_r^*, e_b^*;~\q\big) \equiv&-\h_{\sC}(\q)-\h_\E(-\q)\notag\\
=&-\bigg(\max\biggl\{2\q_2\bigr(e_r^*-e_b^*\bigr)+(\q_1-\q_2)\max\{2(e_r^*-e_b^*),0\},\,\,\,0 \biggr\}\notag\\
&\hspace{1.5cm}+\q_1 e_r^* + \q_2 e_b^* - \q_1\max\{2(e_r^*-e_b^*),0\}\bigg)-\h_\E(-\q).\label{eq:LDA_inner_map}
\end{align}

\vspace{-.25cm}
\noindent Hence, the second part of the test statistic in Eq.~\eqref{eq:Tn_LDA} is the composition of two mappings applied to $\{\hn_\E(\cdot; \widehat{\btheta}),\widehat{e}_r^*,\widehat{e}_b^*\}$: $\mathfrak{\h}_1(\,\cdot\,;\q)$ in Eq.~\eqref{eq:LDA_inner_map} and  $\mathfrak{\h}_2(\cdot)\equiv\max_\q(\cdot)$. 
Each of these mappings is Hadamard directionally differentiable at $\{\h_\E(\cdot),e_r^*,e_b^*\}$ tangentially to $\ell^\infty(\tilde{\bS}^1)\times\mathbb{R}^2$ (\citeauthor{fan:san19}, \citeyear{fan:san19}, Example 2.1; \citeauthor{car:cue:alb20}, \citeyear{car:cue:alb20}, Theorem 2.1).
By \citet[Proposition 3.6]{Shapiro90},
\vspace{-.2cm}
\begin{align}
\max_{\q\in\tilde{\bS}^1}\left(-\h_{\sC}(\q)-\h_\E(-\q)\right)=\mathfrak{\h}_2\circ\mathfrak{\h}_1\big(\h_\E(\cdot), e_r^*, e_b^*;\,\q\big)\equiv\mathfrak{\h}\big( \h_\E\big(\cdot), e_r^*, e_b^*\big)\equiv\mathfrak{\h}\label{eq:composite_h}
\end{align}

\vspace{-.25cm}
\noindent is directionally differentiable at $(\h_\E(\cdot),e_r^*,e_b^*)$ tangentially to $\ell^\infty(\tilde{\bS}^1)\times\mathbb{R}^2$, with
\vspace{-.25cm}
\begin{align}
\mathfrak{\h}'(\cdot)=\mathfrak{\h}'_{2,\mathfrak{h}_1^*(\q)}\circ\mathfrak{\h}'_{1,s^*}(\,\cdot~;\q),\label{eqn:deriv-compo}
\end{align}

\vspace{-.25cm}
\noindent where $\mathfrak{h}_1^*(\q)\equiv\mathfrak{h}_1\big(\h_\E(\cdot),e_r^*,e_b^*\,;\q\big)$ and $s^*\equiv(\q_1-\q_2)\max\{2(e_r^*-e_b^*),0\}+2\q_2\bigr(e_r^*-e_b^*\bigr)$;
for any $s_r,s_b,s\in\mathbb{R}$, $s_h\in\ell^\infty(\tilde{\bS}^1)$ (the space of bounded functions over the compact set $\tilde{\bS}^1$) and continuous $f\in\ell^\infty(\tilde{\bS}^1)$,
\vspace{-.25cm}
\begin{align}    
    \mathfrak{\h}'_{1,s^*}\big(s_h(\cdot), s_r, s_b\,;\q\big)&=-\phi_{1,s^*}^\prime(2\q_2[s_r-s_b]+2(\q_1-\q_2)\phi_{1,e_r^*-e_b^*}^\prime[s_r-s_b])\notag\\
    &\hspace{1.5cm}-\q_1 s_r-\q_2 s_b+2\q_1\phi_{1,e_r^*-e_b^*}^\prime[s_r-s_b]-s_h(-\q),\label{eq:der_h1}\\
    \mathfrak{\h}'_{2,\mathfrak{h}_1^*(\q)}(f)&= \max_{\bigl\{\q\in\tilde{\bS}^1:\mathfrak{h}_1^*(\q)=\mathfrak{h}\bigr\}} f(\q).\label{eq:der_h2}
\end{align}

\vspace{-.25cm}
\noindent In Eq.~\eqref{eq:der_h1}, for any $t\in\mathbb{R}$
\vspace{-.25cm}
\begin{align}
\phi_{1,s^*}'(t)\equiv\begin{cases}
	t, &\text{ if $s^*>0$}\\
	\max\{t,0\}, &\text{ if $s^*=0$} \\
	0, &\text{ if $s^*<0$}
\end{cases}\label{max-deriv}
\end{align}
is the Hadamard directional derivative of $\phi_1(s)\equiv\max\{s,0\}$ at $s^*$ \citep[Example 2.1]{fan:san19}.
 Eq.~\eqref{eq:der_h2} results from \citet[Corollary 2.3]{car:cue:alb20}, as $\mathfrak{\h}_1(\cdot~;\q)$ is a continuous function over compact support.
 By \citet[Proposition 3.6]{Shapiro90},  $T_n^{\texttt{LDA}}$ is Hadamard directionally differentiable at $(\h_\E(\cdot),e_r^*,e_b^*)$ tangentially to $\ell^\infty(\tilde{\bS}^1)\times\mathbb{R}^2$, with directional derivative given by the sum of the two directional derivatives derived above.  
By \citet[Theorem 3.4]{kai16}, \citet[Theorem 2.1]{fan:san19}, and \citet[Theorem 2.2]{car:cue:alb20}, 
for $\mathbb{G}[\mathbf{Z}_i^{*}]\equiv\left[\mathbb{G}[Z_i^{r,*}]~~~~~\mathbb{G}[Z_i^{b,*}]\right]^\intercal$, $\psi^{\texttt{LDA}}$ takes the form
\begin{multline}
\psi^{\texttt{LDA}}=\left[\sup_{\q\in Q^*_{\bS^1}(e^*)}\q^\intercal\mathbb{G}[\mathbf{Z}_i^{*}]-\mathbb{G}[\zeta_i^*(\cM\q;\btheta)]\right]_+\\
+\left[\mathfrak{\h}'_{2,\mathfrak{h}_1^*(\q)}  \left\{ \mathfrak{h}_{1,s^*}'\big(\mathbb{G}[\zeta_i^*(-\cM(\,\cdot\,);\btheta)], \mathbb{G}[Z_i^{r,*}], \mathbb{G}[Z_i^{b,*}]\,;\q\big)\right\}\right]_-. \label{psi-inf}
\end{multline}

If $c^{\texttt{LDA}}_{1-\alpha+\varsigma}+\varsigma$ is a continuity point of the distribution of $\psi^{\texttt{LDA}}$,

\vspace{-1cm}
\begin{align*}
\lim_{n\to\infty}\mathbb{P}\bigl(T_n^{\texttt{LDA}}>c^{\texttt{LDA}}_{1-\alpha+\varsigma}+\varsigma\bigr)
=\mathbb{P}\bigl(\psi^{\texttt{LDA}}>c^{\texttt{LDA}}_{1-\alpha+\varsigma}+\varsigma\bigr)\le\alpha.
\end{align*}
For $\varsigma$ small enough, if $c^{\texttt{LDA}}_{1-\alpha+\varsigma}+\varsigma$ is a discontinuity point of $\psi^{\texttt{LDA}}$, then $\psi^{\texttt{LDA}}$ is continuous at $c^{\texttt{LDA}}_{1-\alpha}$, and since $\mathbb{P}(T_n^{\texttt{LDA}}>c^{\texttt{LDA}}_{1-\alpha+\varsigma}+\varsigma)\le\mathbb{P}(T_n^{\texttt{LDA}}>c^{\texttt{LDA}}_{1-\alpha})$, the result follows.

If Eq.~\eqref{eq:no_kinks_cond} is satisfied, $\E$ has no kinks or flat faces (by Assumption~\ref{asm:smooth-X}). Under the null, since $e^*\in\sF$, $\left\{\q\in\tilde{\bS}^1:\mathfrak{h}_1^*(\q)=\mathfrak{h}\right\}=\arg\max_{\q\in\tilde{\bS}^1}\left(-\h_{\sC}(\q)-\h_\E(-\q)\right)=\{\q^*_{\tilde{\bS}^1}\}$
and $-(\q^*_{\tilde{\bS}^1})^\intercal e^*-\h_\E(-\q^*_{\tilde{\bS}^1})=0$, so that $e^*=\ss_\E(-\q^*_{\tilde{\bS}^1})$. 
It then follows that under Eq.~\eqref{eq:no_kinks_cond},
\begin{multline}
\psi^{\texttt{LDA}}=\left[\q^{*\intercal}_{\bS^1}\mathbb{G}[\mathbf{Z}_i^{*}]-\mathbb{G}[\zeta_i^*(\cM\q^*_{\bS^1};\btheta)]\right]_+ \\
+ \left[\mathfrak{h}_{1,s^*}'\big(\mathbb{G}[\zeta_i^*(-\cM(\cdot);\btheta)], \mathbb{G}[Z_i^{r,*}], \mathbb{G}[Z_i^{b,*}]\,;\q^*_{\bS^1}\big)\right]_-. \label{psi-inf_no_kink}
\end{multline}

\end{proof}

\begin{proof}[Proof of Proposition \ref{prop:bootstrap}]
We verify the assumptions in Theorem 3.2 of \citet{fan:san19}, from which it follows that Proposition \ref{prop:bootstrap} holds and that $\widehat{c}_\beta=c_\beta+o_p(1)$ \citep[Online Appendix, Eq. (S.13), p. 4]{fan:san19} when $c_\beta$ is a point at which the cdf of $\phi_{he^*}'(\mathbb{G}_{he^*})$ is continuous and
increasing.

Assumptions 1, 3(i), 3(iii), and 3(iv) of \citet{fan:san19} hold by construction; their Assumption 2 holds by Eq. \eqref{eq:e-h-joint}, and Assumption 4 holds by Lemma S.3.8 of \citet[Online Appendix]{fan:san19}. Lastly, to show Assumption 3(ii) holds, let $\mathcal{O}_n\equiv\{(\Y_i,\G_i,\X_i)\}_{i=1}^n$. In the next paragraph, we establish that
    \begin{align}
       &\sup_{f\in\mathcal{BL}_1}\left|\mathbb{E}\left[f\left(\sqrt{n}\{\widetilde{he^*}(\widehat{\btheta})-\widehat{he^*}(\widehat{\btheta})\}\right)\,\bigg|\,\mathcal{O}_n\right]-\mathbb{E}[f(\mathbb{G}_{he^*})]\right|\notag\\
       &= \sup_{f\in\mathcal{BL}_1}\left|\mathbb{E}\left[f\left(\begin{bmatrix}
           \mathbb{G}_n[(W_i-1)\zeta_i^*(\cM\q;\btheta)]\\
           \mathbb{G}_n[(W_i-1)Z_i^{r,*}]\\
           \mathbb{G}_n[(W_i-1)Z_i^{b,*}]
       \end{bmatrix}\right)\,\bigg|\,\mathcal{O}_n\right]-\mathbb{E}[f(\mathbb{G}_{he^*})]\right|+o_p(1),\label{eq:bs-nice}
    \end{align}
where $\mathbb{G}_{he^*}$ is defined in Eq. \eqref{eq:e-h-joint}, and by Theorem 3.6.13 of \citet{van:wel13}, Eq. $\eqref{eq:bs-nice}=o_p(1)$, and therefore Assumption 3(ii) of \citet{fan:san19} holds.  
    
To show that the equality in Eq.~\eqref{eq:bs-nice} holds, we let $\widetilde{\mu}_g^W\equiv\overline{W}\widetilde{\mu}_g$, $\tcM^W\equiv\diag(1/\widetilde{\mu}_r^W, 1/\widetilde{\mu}_b^W)$. Recall $\bxi_{\ss,i}(\check{\cM};\mathring{\cM}\q;\widehat{\btheta})$ defined in Eq.~\eqref{eq:bxi}. Decompose the bootstrapped process by
    \begin{align*}
        &\sqrt{n}\big\{\widetilde{he^*}(\widehat{\btheta})-\widehat{he^*}(\widehat{\btheta})\big\}=\sqrt{n}\big\{\widetilde{he^*}(\widehat{\btheta})-he^*\big\}-\sqrt{n}\big\{\widehat{he^*}(\widehat{\btheta})-he^*\big\}\\
        &=\sqrt{n}\begin{bmatrix}
            \frac{1}{n}\sum_{i=1}^nW_i\q^\intercal\bxi_{\ss,i}(\tcM^W;\tcM\q;\widehat{\btheta})-\mathbb{E}[\zeta_i(\cM\q;\btheta)]\\
            \frac{1}{n}\sum_{i=1}^nW_i\frac{Z_i^r}{\widetilde{\mu}_r^W}-\frac{\mathbb{E}[Z_i^r]}{\mu_r}\\
            \frac{1}{n}\sum_{i=1}^nW_i\frac{Z_i^b}{\widetilde{\mu}_b^W}-\frac{\mathbb{E}[Z_i^b]}{\mu_b}
        \end{bmatrix}-\begin{bmatrix}
            \mathbb{G}_n[\zeta_i^*(\cM\q;\btheta)]\\
            \mathbb{G}_n[Z_i^{r,*}]\\
            \mathbb{G}_n[Z_i^{b,*}]
        \end{bmatrix}+o_p(1),
    \end{align*}
where the second equality follows from Theorem \ref{thm:gaussian} and the proof of Proposition \ref{prop:lda}. Next, observe that
\begin{align*}
    &\sqrt{n}\left(\frac{1}{n}\sum_{i=1}^nW_i\q^\intercal\bxi_{\ss,i}(\tcM^W;\tcM\q;\widehat{\btheta})-\mathbb{E}[\zeta_i(\cM\q;\btheta)]\right)\\
    &=\underbrace{\sqrt{n}\left(\frac{1}{n}\sum_{i=1}^n\q^\intercal(\tcM^W\cM^{-1})\big(W_i\bxi_{\ss,i}(\cM;\tcM\q;\widehat{\btheta})-\mathbb{E}[W_i\bxi_{\ss,i}(\cM;\cM\q;\btheta)]\big)\right)}_{\equiv\widetilde{A}}\\
    &\hspace{5cm}+\underbrace{\sqrt{n}\q^\intercal\left(\tcM^W\cM^{-1}-\bI\right)\mathbb{E}[\bxi_{\ss,i}(\cM;\cM\q;\btheta)]}_{\equiv\widetilde{B}}\\
    &=\mathbb{G}_n[W_i\zeta_i^*(\cM\q;\btheta)]+o_p(1),
\end{align*}
where $\mathbb{E}[W_i\bxi_{\ss,i}(\cM;\cM\q;\btheta)]=\mathbb{E}[\bxi_{\ss,i}(\cM;\cM\q;\btheta)]$  by independence of $W_i$ and $\mathbb{E}[W_i]=1$, $\widetilde{A}$ (resp., $\widetilde{B}$) is the bootstrapped analogue of $A$ (resp., $B$) in the proof of Theorem \ref{thm:gaussian}, and the last equality follows from a similar argument used in the proof of Theorem \ref{thm:gaussian}. 

In addition, for $g\in\{r,b\}$, by the Delta method,
\begin{align*}
    &\sqrt{n}\left(\frac{1}{n}\sum_{i=1}^nW_i\frac{Z_i^g}{\widetilde{\mu}_g^W}-\frac{\mathbb{E}[Z_i^g]}{\mu_g}\right)\\
    &=\sqrt{n}\left(\frac{1}{n}\sum_{i=1}^nW_i\frac{Z_i^g}{\mu_g}-\frac{\mathbb{E}[Z_i^g]}{\mu_g}\right)+\sqrt{n}\left(\frac{1}{\widetilde{\mu}_g^W}-\frac{1}{\mu_g}\right)\left(\frac{1}{n}\sum_{i=1}^nW_iZ_i^g\right)\\
    &=\mathbb{G}_n[W_iZ_i]\frac{1}{\mu_g}-\mathbb{G}_n[W_i\mathds{1}\{G_i=g\}]\frac{\mathbb{E}[Z_i^g]}{\mu_g^2}+o_p(1)=\mathbb{G}_n[W_iZ_i^{g,*}]+o_p(1).
\end{align*}
Therefore,
\begin{align*}
    \sqrt{n}\big\{\widetilde{he^*}(\widehat{\btheta})-\widehat{he^*}(\widehat{\btheta})\big\}=\begin{bmatrix}
           \mathbb{G}_n[(W_i-1)\zeta_i^*(\cM\q;\btheta)]\\
           \mathbb{G}_n[(W_i-1)Z_i^{r,*}]\\
           \mathbb{G}_n[(W_i-1)Z_i^{b,*}]
       \end{bmatrix}+o_p(1),
\end{align*}
yielding Eq. \eqref{eq:bs-nice}.

\end{proof}

\subsection{Proofs for Section \ref{sec:distance_F}}
\begin{proof}[Proof of Proposition \ref{prop:dist-F}]
Consider the null hypothesis $\HH_0:\rho(e^*,\F)=\delta$ against $\HH_A:\rho(e^*,\F)\neq\delta$ for some $\delta>0$, and view our confidence interval as the result of inverting this hypothesis test. Let $\varphi_n^{\texttt{dist}}\equiv\mathds{1}\left\{\inf_{\rho(\Tilde{e}, \Tilde{\F})\in \mathcal{CS}_n^{\rho(e^*,\F)}}\left|\rho(\Tilde{e}, \Tilde{\F}) - \delta\right|>0\right\}$ and partition the parameter space of the location of $\E$ relative to $\mathcal{H}_{45}$ such that
\begin{align}
    \varphi_n^{\texttt{dist}} =
    &~ \varphi_n^{\texttt{dist}}\mathds{1}\big\{\E\cap\mathcal{H}_{45}^-=\emptyset\big\}\label{eq:tstat-above}\\
    &~~+\varphi_n^{\texttt{dist}}\mathds{1}\big\{\E\cap\mathcal{H}_{45}^+=\emptyset\big\}\label{eq:tstat-below}\\
    &~~~~+\varphi_n^{\texttt{dist}}\mathds{1}\big\{\E\cap\mathcal{H}_{45}^+\neq\emptyset, \E\cap\mathcal{H}_{45}^-\neq\emptyset\big\}.\label{eq:tstat-cross}
\end{align}
Note that whenever $\mathcal{CS}_n^+(e^*,\F)\neq\emptyset$,
\begin{multline}
    \inf_{\rho(\Tilde{e}, \Tilde{\F})\in \mathcal{CS}_n^{\rho(e^*,\F)}}|\rho(\Tilde{e}, \Tilde{\F})-\delta|\le \inf_{(\Tilde{e}, \Tilde{\F})\in \mathcal{CS}_n^+(e^*,\F)}|\rho(\Tilde{e}, \Tilde{\F})-\delta|,\\
    \Rightarrow \mathbb{E}[\varphi_n^{\texttt{dist}}]\le\mathbb{E}\left[\mathds{1}\left\{\inf_{(\Tilde{e}, \Tilde{\F})\in \mathcal{CS}_n^+(e^*,\F)}|\rho(\Tilde{e}, \Tilde{\F})-\delta| > 0\right\}\right],\label{eq:varphi_dist_smaller}
\end{multline}
and similarly when $(\Tilde{e}, \Tilde{\F})\in\mathcal{CS}_n^+(e^*,\F)$ is replaced with  $(\Tilde{e}, \Tilde{\F})\in\mathcal{CS}_n^-(e^*,\F)$ or $\rho(\Tilde{e}, \Tilde{\F})\in\mathcal{CS}_n^{45}(\rho(e^*,\F))$ (and under the case where these sets are non-empty).

When $\E\cap\mathcal{H}_{45}^-=\emptyset$, we have $ \F\in\mathcal{H}_{45}^+\cup\mathcal{H}_{45}$ and $\lim_{n\to\infty}\bP\big((e^*,\F)\in \mathcal{CS}_n^+(e^*,\F)\big)\geq1-\alpha$ by a similar argument to the proof of Proposition \ref{prop:weak-GB}.
It then follows that under the null $\rho(e^*,\F)=\delta$,
$\mathbb{E}[\eqref{eq:tstat-above}]\leq \alpha\mathds{1}\{\E\cap\mathcal{H}_{45}^-=\emptyset\}$ as $n\to\infty$, using Eq.~\eqref{eq:varphi_dist_smaller} and the fact that $\bP\left(\inf_{(\Tilde{e},\Tilde{\F})\in \mathcal{CS}_n^+(e^*,\F)}\left|\rho(\Tilde{e}, \Tilde{\F})-\delta\right|>0
    \right)\leq 1-\bP\big((e^*,\F)\in \mathcal{CS}_n^+(e^*,\F)\big)$.
Similarly, $\mathbb{E}[\eqref{eq:tstat-below}]\leq\alpha\mathds{1}\{\E\cap\mathcal{H}_{45}^+=\emptyset\}$ as $n\to\infty$.

We complete the proof by showing $\mathbb{E}[\eqref{eq:tstat-cross}]\leq\alpha\mathds{1}\{\E\cap\mathcal{H}_{45}^+\neq\emptyset, \E\cap\mathcal{H}_{45}^-\neq\emptyset\}$ if $\rho(e^*,\F)=\delta$. In the event $\E\cap\mathcal{H}_{45}^+\neq\emptyset$ and $ \E\cap\mathcal{H}_{45}^-\neq\emptyset$, $\F\in\mathcal{H}_{45}$ and---since $\E$ is \textit{not} tangent to $\mathcal{H}_{45}$ in this case---the direction in which $\F$ is the support set of $\E$ does \textit{not} live in the span $\{c[1~ -1]^\intercal: c\in\mathbb{R}\}$. Hence, $c^*\equiv\arg\inf_{c\in\mathbb{R}}\h_\E(\mathfrak{u}_1(c))$ for $\mathfrak{u}_1(c)\equiv\mathfrak{u}_1-c[1~-1]^\intercal$ is bounded, since otherwise it would require $\E$ to be tangent to $\mathcal{H}_{45}$. In addition, as explained in Footnote \hyperref[ftnt:bound_inf_E_tilde]{3},  $\h_\E(\mathfrak{u}_1(c))$ is bounded from below when $\F\in\mathcal{H}_{45}$. We can therefore restrict attention to $c\in[-c_3, c_3]\equiv\mathcal{C}_3\subset\mathbb{R}$ for some bounded constant $c_3>0$ in the characterization of $\F$ in Eq.~\eqref{eq:h_tildeE} so that $\h_\E(\mathfrak{u}_1(\cdot))$ is a bounded function over compact support $\mathcal{C}_3$.  Under the maintained assumptions, Eq. \eqref{eq:e-h-joint} holds and implies that
\begin{align*}
\sqrt{n}\begin{bmatrix}
	\hn_\E\big(\mathfrak{u}_1(c); \widehat{\btheta}\big)-\h_\E\big(\mathfrak{u}_1(c)\big)\\
	\widehat{e}_r^*-e_r^*\\
	\widehat{e}_b^*-e_b^*
\end{bmatrix}\xrightarrow[]{d}\begin{bmatrix}
	\|\mathfrak{u}_1(c)\|_E\mathbb{G}[\zeta_i^*(\cM\mathfrak{u}_1(c)/\|\mathfrak{u}_1(c)\|_E;\btheta)]\\
	\mathbb{G}[Z_i^{r,*}]\\
	\mathbb{G}[Z_i^{b,*}]
\end{bmatrix} ~\text{in }\ell^\infty(\mathcal{C}_3),
\end{align*}
where we use the property of support functions that for any constant $\tilde{c}>0$, $\h_\E(\tilde{c}\cdot\q)=\tilde{c}\cdot\h_\E(\q)$ \citep[see,][p.45]{sch93}. Recall $\F_{45}\equiv\mathfrak{u}\cdot\h_{\Tilde{\E}}(\mathfrak{u}_1)$ and let $\widehat{\F}_{45}$ be its estimator with $\h_\E(\cdot)$ replaced by $\widehat{\h}_\E(\,\cdot\,;\widehat{\btheta})$ in the expression of $\h_{\Tilde{\E}}(\cdot)$ given in Eq.~\eqref{eq:h_tildeE}. Observe that $\rho(e^*,\F_{45})$ is a composition of two Hadamard directionally differentiable functions: let $\mathfrak{h}_4:\ell^\infty(\mathcal{C}_3)\to\mathbb{R}$ be the $\inf_{c\in\mathcal{C}_3}(\cdot)$ function,
\begin{align*}
   \rho(e^*,\F_{45})= \rho\bigg(e^*,~\mathfrak{u}\cdot\mathfrak{h}_4\big\{\h_\E\big(\mathfrak{u}_1(c)\big)\big\}\bigg),
\end{align*}
where by \citet[Corollary 2.3]{car:cue:alb20}, $\mathfrak{h}_4$ is directionally differentiable at $\h_\E\big(\mathfrak{u}_1(\cdot)\big)$ tangentially to $\ell^\infty(\mathcal{C}_3)$, with
\begin{align*}
    \mathfrak{h}'_{4,\h_\E(\mathfrak{u}_1(\cdot))}(f)=\inf_{\bigl\{c\in\mathcal{C}_3:\h_\E\big(\mathfrak{u}_1(c)\big)=\h_{\Tilde{\E}}(\mathfrak{u}_1)\bigr\}} f(c),~~\text{for continuous } f\in\ell^\infty(\mathcal{C}_3).
\end{align*}
Denote $\rho_{(e^*,\F_{45})}' : \mathbb{R}^2\times\mathbb{R}^2\to\mathbb{R}$ the directional derivative of $\rho$ at $(e^*,\F_{45})$. Then by \citet[Proposition 3.6]{Shapiro90}, \citet[Theorem 2.1]{fan:san19}, and \citet[Theorem 2.2]{car:cue:alb20},
\begin{align*}
\sqrt{n}\begin{bmatrix}
	\rho(\widehat{e}^*,\widehat{\F}_{45})-\rho(e^*,\F_{45})
\end{bmatrix}\xrightarrow[]{d}\psi^{\rho(e^*,\F_{45})},
\end{align*}
where, for $\mathbb{G}[\mathbf{Z}_i^{*}]\equiv\left[\mathbb{G}[Z_i^{r,*}]~~~~~\mathbb{G}[Z_i^{b,*}]\right]^\intercal$,
\begin{align*}
\psi^{\rho(e^*,\F_{45})}\equiv\rho_{(e^*,\F_{45})}'\left(
        \mathbb{G}[\mathbf{Z}_i^{*}],\mathfrak{u}\cdot\mathfrak{h}'_{4,\h_\E(\mathfrak{u}_1(\cdot))}\big(\|\mathfrak{u}_1(c)\|_E\mathbb{G}[\zeta_i^*(\cM\mathfrak{u}_1(c)/\|\mathfrak{u}_1(c)\|_E;\btheta)]\big)\right),
\end{align*}
so by the continuous mapping theorem,
\begin{align}
    \psi^{45}=\left|\psi^{\rho(e^*,\F_{45})}\right|,\label{eq:psi-delta}
\end{align}
Next, if Eq.~\eqref{eq:no_kinks_cond} holds and $\E$ has no kinks, we show that the set $\left\{c\in\mathcal{C}_3:\h_\E\big(\mathfrak{u}_1(c)\big)=\h_{\Tilde{\E}}(\mathfrak{u}_1)\right\}$ is a singleton and the expression of $\psi^{\rho(e^*, \F_{45})}$ simplifies. Recall $c^*\equiv\arg\inf_{c\in\mathcal{C}_3}\h_\E\big(\mathfrak{u}_1(c)\big)$. By contradiction, assume there exists $\Tilde{c}\neq c^*$ and $\h_\E\big(\mathfrak{u}_1(\Tilde{c})\big)=\h_\E\big(\mathfrak{u}_1(c^*)\big)=\h_{\Tilde{\E}}(\mathfrak{u}_1)$. This implies that the two linear equations $(-1-\Tilde{c})e_r+\Tilde{c}e_b=\h_{\Tilde{\E}}(\mathfrak{u}_1)$ and $(-1-c^*)e_r+c^*e_b=\h_{\Tilde{\E}}(\mathfrak{u}_1)$ intersect at some point $(e_r^\star,e_b^\star)$. Replacing this value in the equations, we find $\Tilde{c}(e_b^\star-e_r^\star)=c^*(e_b^\star-e_r^\star)$. If $c^*=0$, for $\Tilde{c}$ not to equal $c^*$ it must be the case that $\Tilde{c}\neq 0$, in which case $e_b^\star=e_r^\star$, implying that $(e_r^\star,e_b^\star)=\R=\F$ and in turn by Eq.~\eqref{eq:no_kinks_cond} it must be the case that $\Tilde{c}=c^*$. Similarly, if $c^*\neq 0$, then either $e_b^\star=e_r^\star$, which implies $(e_r^\star,e_b^\star)=\F$ and hence $\Tilde{c}=c^*$, or $\Tilde{c}/c^*=1$ and the claim follows. Let $\Tilde{\mathfrak{u}}_1(c^*)\equiv\frac{\mathfrak{u}_1(c^*)}{\|\mathfrak{u}_1(c^*)\|_E}$, we get:
\begin{align}
    \psi^{45}=\left|\rho_{(e^*,\F_{45})}'\bigg(
        \mathbb{G}[\mathbf{Z}_i^{*}],\mathfrak{u}\cdot\|\mathfrak{u}_1(c^*)\|_E\cdot\mathbb{G}\left[\zeta_i^*\left(\cM\Tilde{\mathfrak{u}}_1(c^*);\btheta\right]\right)\bigg)\right|.\label{eq:psi-delta_no_kink}
\end{align}
If $\F\in\mathcal{H}_{45}$, $\lim_{n\to\infty}\bP\left(\mathcal{CS}_n^{45}(\rho(e^*,\F))=\emptyset\right)=0$.
Under the null $\rho(e^*,\F)=\delta$, and using again Eq.~\eqref{eq:varphi_dist_smaller}, if $c_{1-\alpha+\varsigma}^{\rho(\Tilde{e}, \Tilde{\F})}+\varsigma$ is a continuity point of the distribution of $\psi^{\rho(\Tilde{e}, \Tilde{\F})}$,
\begin{align*}
    \lim_{n\to\infty}\mathbb{E}[\varphi_n^{\texttt{dist}}]\le&\lim_{n\to\infty}\bP\big(T_n^{\rho(\Tilde{e}, \Tilde{\F})}>c_{1-\alpha+\varsigma}^{\rho(\Tilde{e}, \Tilde{\F})}+\varsigma \big)\leq\alpha
\end{align*}
and if it is a discontinuity point for an infinitesimal $\varsigma$, 
then $c_{1-\alpha}^{\rho(\Tilde{e}, \Tilde{\F})}$ is a continuity point and
\begin{align*}
    \lim_{n\to\infty}\mathbb{E}[\varphi_n^{\texttt{dist}}]\leq&\lim_{n\to\infty}\bP\big(T_n^{\rho(\Tilde{e}, \Tilde{\F})}>c_{1-\alpha}^{\rho(\Tilde{e}, \Tilde{\F})}\big)=\alpha.
\end{align*}
Therefore, under the null $\rho(e^*,\F)=\delta$, we conclude 
\begin{align*}
    \lim_{n\to\infty}\mathbb{E}[\varphi_n^{\texttt{dist}}]\leq\alpha.
\end{align*}
Test inversion yields the coverage result.
\end{proof}

\numberwithin{asm}{section}
\section{Auxiliary Results}\label{appn:B}
\subsection{Threshold Rules}\label{app:sec:threshold_rules}
In the paper, we allow $\mathcal A(\sX)$, the set of all algorithms that map from the input space $\sX$ to $[0,1]$, to be completely unrestricted. This includes randomized rules where the event $D=1$ occurs with probability $a(\X)$.
Here instead we consider \emph{threshold rules}.
Let $\mathcal{A}^\texttt{th}(\sX)$ denote a space of algorithms $\{a:\sX \mapsto \mathbb{R}\}$ such that $a\in\mathcal{A}^\texttt{th}(\sX)$ induces the decision rule
\begin{align*}
    D_{a} = \mathds{1}\{a(\X) \ge 0\}.
\end{align*}
One could alternatively pick a constant $\kappa$ a priori and use a decision rule of the form $\mathds{1}\{a(\X) \ge \kappa\}$, but to ease notation we absorb the threshold $\kappa$ in $a$.
We maintain the following richness assumption on $\mathcal{A}^\texttt{th}(\sX)$:
\begin{asm}
\label{asm:rich-A}
\textit{The set of algorithms $\mathcal{A}^\texttt{th}(\sX)$ is (i) convex; and (ii) sufficiently rich, in the sense that, $\X$-a.s.,
\begin{align*}
    \exists a'\in\mathcal{A}^\texttt{th}(\sX):~\mathds{1}\{a'(\X) \ge 0\}&=1,\\
    \exists a''\in\mathcal{A}^\texttt{th}(\sX):~\mathds{1}\{a''(\X) \ge 0\}&=0.
\end{align*}
    }
\end{asm}
\begin{remark}[Linear Threshold Rules]\label{remark:linear_threshold_rules1}
    Assumption~\ref{asm:rich-A} is satisfied by linear threshold rules with $\mathcal A^\texttt{th}(\sX)=\{[1;\X]^\intercal\beta:~\beta\in\mathbb{R}^{d_\X+1}\}$ and 
    $D_{\beta}=\mathds{1}\{[1;\X]^\intercal\beta \ge 0\}$.
\end{remark}
When using threshold rules, similar to Eq.~\eqref{eq:expression:e_g}, the group risks can be expressed as
\begin{align}
e_g(D_{a})&\equiv\frac{1}{\mu_g}\mathbb{E}_{\X}\left[D_{a}\theta_1^g(\X)+(1-D_{a})\theta_0^g(\X)\right]\notag\\
&= \frac{1}{\mu_g}\mathbb{E}\left[\LL_0^g+(\LL_1^g-\LL_0^g)\mathds{1}\{a(\X) \ge 0\}\right].\label{eq:expression:e_g_threshold}
\end{align}
Compare Eq.~\eqref{eq:expression:e_g_threshold} with Eq.~\eqref{eq:expression:e_g}: it follows immediately that threshold rules given by $D_{k(\btheta(\X),\cM\q)}=\mathds{1}\{k(\btheta(\X), \cM\q)\ge 0\}$ yield the extreme points of the set $\E$ in Eq.~\eqref{eq:E_as_Aumann}.
As we show next, under Assumption \ref{asm:rich-A}, the feasible set associated with threshold rules, denoted $\E^\texttt{th}$, is convex.
To see this, note that
\begin{align}
\E^\texttt{th}&=\left\{\mathbb{E}[\cM\vartheta(\X)]: \vartheta(\X)\in \left\{\btheta_0(\X), \btheta_1(\X)\right\} \right\}\equiv \mathbf{E}\left[\cM\tilde\linseg(\X)\right],\label{eq:E_as_Aumann_threshold}
\end{align}
 with $\btheta_d(\X)$ defined in Eq.~\eqref{eq:def_theta_d} and $\tilde\linseg(\X)\equiv\left\{\btheta_0(\X), \btheta_1(\X)\right\}$.
 Relative to Eq.~\eqref{eq:E_as_Aumann}, the fundamental difference here is that $\left\{\btheta_0(\X), \btheta_1(\X)\right\}$ is a two-point set instead of an interval.
 Nonetheless, the set on the right-hand-side of Eq.~\eqref{eq:E_as_Aumann_threshold} is by definition \citep[Def. 3.1]{mol:mol18} the \emph{Aumann expectation} of the two-point set $\tilde\linseg(\X)$.
 Next, observe that under Assumption \ref{asm:smooth-X} the probability space on which $(\Y,\G,\X)$ are defined is non-atomic (non-atomicity follows as long as one of the variables in $\X$ has a continuous distribution, and if all variables in $\X$ had a discrete distribution, Assumption \ref{asm:smooth-X} would fail).\footnote{To see this, let $\X$ have countable support, take $x\in\sX$ such that $\bP(\X=x)=\varsigma>0$. Then for $\q=\frac{[\{\theta_1^b(x)-\theta_0^b(x)\}/\mu_b
          \quad
          \{\theta_0^r(x)-\theta_1^r(x)\}/\mu_r]^\intercal}{\Vert \cM\{\btheta_1(x)-\btheta_0(x) \}\Vert}$,
 $\bP(|\q_1\{\theta_1^r(x)-\theta_0^r(x)\}/\mu_r+\q_2\{\theta_1^b(x)-\theta_0^b(x)\}/\mu_b|=0)\geq\varsigma>0$,
  and hence $\bP(|\q_1\{\theta_1^r(x)-\theta_0^r(x)\}/\mu_r+\q_2\{\theta_1^b(x)-\theta_0^b(x)\}/\mu_b|<\delta)\geq\varsigma>0$ for any $\delta>0$.}
  As both $\btheta_0(\X)$ and $\btheta_1(\X)$ are absolutely integrable, all conditions required for Theorem 3.4 in \citet{mol:mol18} are satisfied, yielding:
 \begin{align}
 \mathbf{E}\left[\cM\tilde\linseg(\X)\right]=\mathbf{E}\big[\cM\conv\left(\{\btheta_0(\X),\btheta_1(\X)\}\right)\big]=\mathbf{E}\left[\cM\linseg(\X)\right].\label{eq:convexification_threshold}
 \end{align}
 Theorem 3.11 in \citet{mol:mol18} also applies, and $\h_{\E}(\q)=\h_{\mathbf{E}\left[\cM\linseg(\X)\right]}(\q)=\h_{\mathbf{E}\left[\cM\tilde\linseg(\X)\right]}(\q)=\mathbb{E}[\h_{\cM\linseg(\X)}(\q)]$.
 Hence, under Assumption \ref{asm:rich-A}, the feasible set associated with threshold rules is convex and the support function fully characterizes it.
\begin{remark}[Richness of Threshold Rules]\label{remark:linear_threshold_rules2}
The result in Eq.~\eqref{eq:convexification_threshold} shows that threshold rules corresponding to a rich algorithm space can replicate any unconstrained algorithm.
Inspecting further the linear case is instructive.
If one allows for any $\beta\in\mathbb{R}^{d_\X+1}$,  Assumption~\ref{asm:rich-A} is satisfied. 
By the same argument given above, the feasible set associated with linear threshold rules, denoted $\E^\texttt{lin}$, is convex.
Indeed, for each $\beta$ one can write
\begin{multline*}
\hspace{-.4cm}e_g(\beta)\equiv\frac{1}{\mu_g}\mathbb{E}\left[\LL_0^g\big|[1;\X]^\intercal\beta< 0\right]\bP\left([1;\X]^\intercal\beta< 0\right)
+\frac{1}{\mu_g}\mathbb{E}\left[\LL_1^g\big|[1;\X]^\intercal\beta\ge 0\right]\bP\left([1;\X]^\intercal\beta\ge 0\right).
\end{multline*}
Let $\sX_\beta^+\equiv\{\X\in\sX:[1;\X]^\intercal\beta\ge 0\}$ and $\sX_\beta^-\equiv\{\X\in\sX:[1;\X]^\intercal\beta< 0\}$ denote the two sets in which a linear threshold rule with parameter $\beta$ partitions $\sX$.
As at least one component of $\X$ has continuous distribution and $\beta$ has support $\mathbb{R}^{d_\X+1}$, each realization of $\X$ can be allocated either in $\sX_\beta^+$ or in $\sX_\beta^-$ for some $\beta$.
Hence, convexification occurs.
The support function of $\E^\texttt{lin}$ can be expressed as
\begin{align*}
    \h_{\E^\texttt{lin}}(\q)=\max_{\beta\in\mathbb{R}^{d_\X+1}}\q^\intercal e_g(\beta)=\max_{\beta\in\mathbb{R}^{d_\X+1}}(\cM\q)^\intercal \mathbb{E}\left[\bLL_0\mathds{1}([1;\X]^\intercal\beta< 0)+\bLL_1\mathds{1}([1;\X]^\intercal\beta\ge 0)\right].
\end{align*}
\end{remark}

\subsection{Sufficient Conditions Yielding Strict Convexity and No Kinks}\label{app:sec:strict_convex_no_kink}
Assumption~\ref{asm:smooth-X} plays multiple roles in our analysis.
It assures that $\E$ is strictly convex, hence its support set in any direction $\q\in\bS^1$ is a singleton, and it assures Neyman orthogonality of the moment condition defining $\h_\E(\cdot)$. \citet[Section 3.1]{sem23} provides sufficient conditions for this assumption, based on joint Gaussianity of $\btheta_1(\X)-\btheta_0(\X)$. 
A mild strengthening of Assumption~\ref{asm:nuisance-structure} is sufficient both for Assumption~\ref{asm:smooth-X} to hold with $m=1$ and for Eq.~\eqref{eq:no_kinks_cond} to be satisfied, guaranteeing the absence of kinks in $\E$, as we show next.
\begin{asm}
    \label{asm:condition_smoothX1X2}
    \textit{(i) The distribution of $(\X_1,\X_2)|\X_{[3:d_\X]}$ is continuous with a bounded density and $\mathbb{E}\big[\big|\eta^g(\X_{[3:d_\X]})\big|\big]<\infty$ for $g\in\{r,b\}$; (ii) for a set $\tilde{\sX}_{[3:d_\X]}$ of realizations of $\X_{[3:d_\X]}$ with positive probability, the density of $(\X_1,\X_2)|\X_{[3:d_\X]}$ is positive on a ball of radius $c>0$ that includes $\mathbf{0}$ and the image of $\tilde{\sX}_{[3:d_\X]}$  under $\eta^g(\cdot)$ includes $0$.}
\end{asm}
\begin{prop}
\label{prop:no_kinks_condition}
    \textit{If Assumptions \ref{asm:nuisance-structure} and \ref{asm:condition_smoothX1X2}(i) hold, Assumption \ref{asm:smooth-X} is implied. If Assumption \ref{asm:condition_smoothX1X2}(ii) also holds, Eq.~\eqref{eq:no_kinks_cond} holds and the set $\E$ has no kinks.}
\end{prop}
\begin{proof}
Take $\delta>0$, and note that $\sup_{\q\in \bS^1} \mathbb{P}(|k(\btheta(\X),\cM\q)|<\delta)$ can be written as
\begin{align*}
    &\sup_{\q\in \bS^1}\mathbb{E}_{\X_{[3:d_\X]}}\left[\int_0^\delta \left|\q_1\Delta\theta^r(\X)/\mu_r+\q_2\Delta\theta^b(\X)/\mu_b \right|d\bP_{(\X_1,\X_2)|\X_{[3:d_\X]}}\right]\\
    \lesssim &\max_{g\in\{r,b\}}\mathbb{E}_{\X_{[3:d_\X]}}\left[\int_0^\delta \left(|\alpha^g|+|\beta^g|+|\eta^g(\X_{[3:d_\X]})|\right)d\bP_{(\X_1,\X_2)|\X_{[3:d_\X]}}\right]\\
    \lesssim &\max_{g\in\{r,b\}}\mathbb{E}_{\X_{[3:d_\X]}}\left[\delta \left(|\alpha^g|+|\beta^g|+|\eta^g(\X_{[3:d_\X]})|\right)\right]\lesssim\delta,
\end{align*}
where the first inequality follows from  $\sup_{\q\in \bS^1}\|\q\|_E=1$, $\mu_g\in(0,1)$, and that $(\X_1,\X_2)$ has bounded support. The second equality follows from continuity and boundedness of $d\bP_{(\X_1,\X_2)|\X_{[3:d_\X]}}$. The last inequality follows from $\mathbb{E}\big[\big|\eta^g(\X_{[3:d_\X]})\big|\big]<\infty$.

For the second result, observe that $\btheta_1(\X)-\btheta_0(\X)=\Delta\btheta(\X)$ equals:
\begin{align*}
    \underbrace{\begin{bmatrix}
        ~\alpha^r~~~ & \beta^r~\\
        ~\alpha^b~~~ & \beta^b~
    \end{bmatrix}}_{\equiv A_1}\begin{bmatrix}
        \X_1\\
        \X_2
    \end{bmatrix}+\underbrace{\begin{bmatrix}
        \eta^r(\X_{[3:d_\X]})\\
        \eta^b(\X_{[3:d_\X]})
    \end{bmatrix}}_{\equiv A_2},
\end{align*}
where the matrix $A_1$ is invertible by $\alpha^b\beta^r\neq\alpha^r\beta^b$. Under Assumption \ref{asm:condition_smoothX1X2}(ii), on the set $\tilde{\sX}_{[3:d_\X]}$, $A_1[\X_1~\X_2]^\intercal+A_2$ realizes in a set containing $0$. Hence, Eq. \eqref{eq:no_kinks_cond} holds by applying the law of iterated expectations.
\end{proof}

\numberwithin{asm}{section}
\section{Variability of the Empirical Results to Nuisance Parameter Estimation}\label{appn:C}
\subsection{Estimating the Nuisance Parameters Using Lasso}\label{app:sec:lasso}
In this subsection, we report the analogs of Figures \ref{fig:application:estimated:sets}-\ref{fig:tests}-\ref{fig:application:build:alg} and Tables \ref{tbl:LDA-results}-\ref{tbl:frac-bl-trt} using multinational logit lasso from the \href{https://glmnet.stanford.edu/index.html}{\texttt{glmnet}} package to estimate the nuisance parameter $\Delta\btheta$. That is, the only difference between the results reported in this subsection and those in Section \ref{subsec:empirical} lies in the choice of what machine learner is used to estimate nuisance parameters. 
Respectively, these are Figures \ref{fig:application:estimated:sets:lasso}-\ref{fig:tests:lasso}-\ref{fig:application:build:alg:lasso} and Tables \ref{tbl:LDA-results-lasso}-\ref{tbl:frac-bl-trt-lasso}.

\begin{figure}[H]\captionsetup[subfigure]{font=footnotesize}
    \centering
    \includegraphics[width=0.75\linewidth]{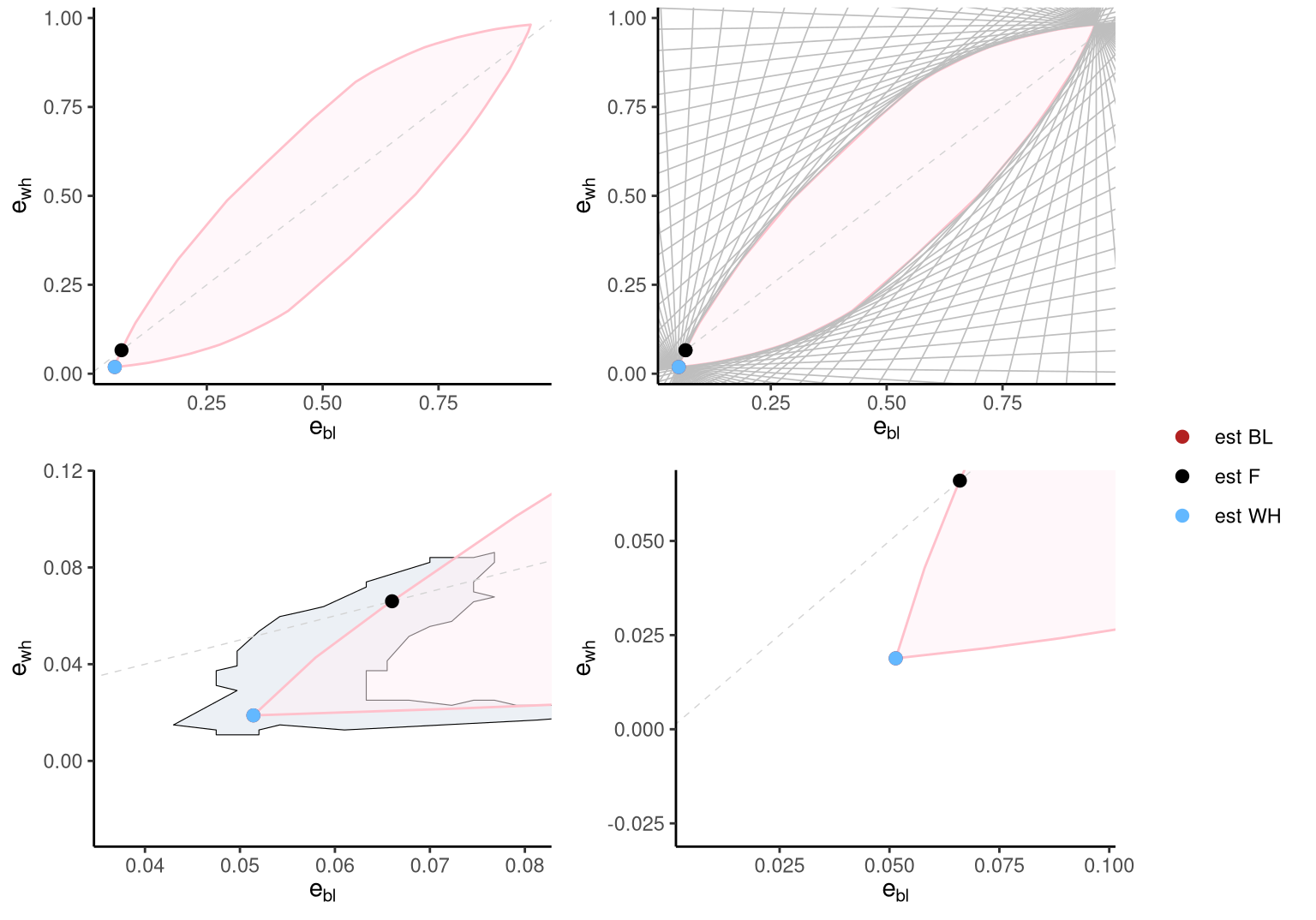}
    \caption{Top-left panel: $\widehat{\E}$; top-right panel: $\widehat{\E}$ along with one hundred supporting hyperplanes; bottom-left panel: zoom-in to $\widehat{\sF}$ and the $95\%$ confidence set around this frontier; bottom-right panel: further zoom-in to the best group-specific points $\textsf{BL}$ and $\textsf{WH}$, and the fairest point $\F$. $\Delta\btheta$ is estimated by logit lasso.}
    \label{fig:application:estimated:sets:lasso}
\end{figure}
\begin{figure}[H]\captionsetup[subfigure]{font=footnotesize}
\centering
\subfigure[Candidate values for $(\textsf{BL}, \textsf{WH})$]{
\includegraphics[width=0.4\linewidth]{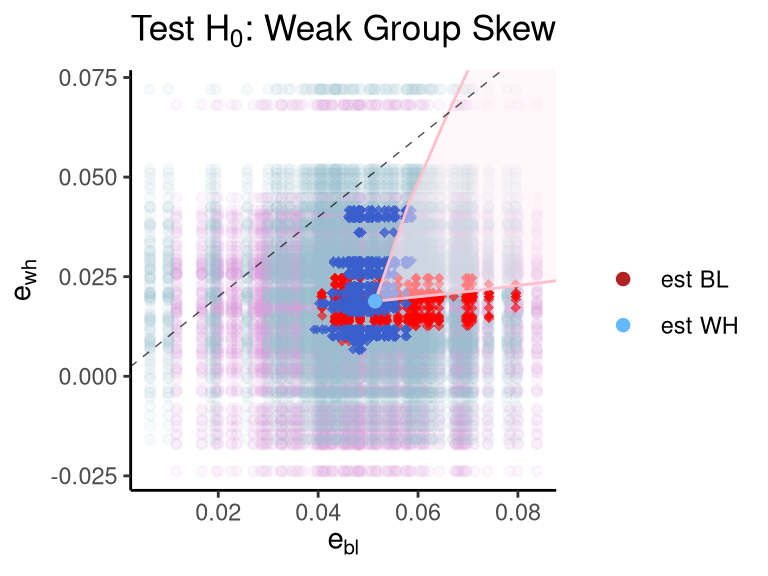}
}
\hspace{1em}
\subfigure[$3$ experimental algorithms]
{\includegraphics[width=0.4\linewidth]{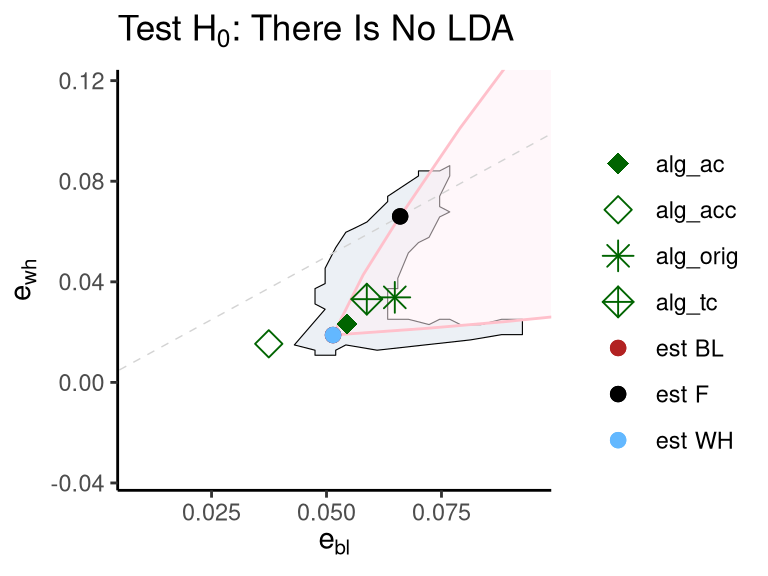}
}
    \caption{Panel (a): Plum-colored (respectively, light-blue colored) circles correspond to candidate values for $e_\textsf{BL}$ ($e_\textsf{WH}$) sampled from a normal distribution centered at $\widehat{e}_\textsf{BL}$ ($\widehat{e}_\textsf{WH}$), and red (blue) diamonds correspond to non-rejected values.
    Panel (b): $\widehat{\sF}$ along with its $95\%$ confidence set and the estimated group risks for four algorithms considered by \citetalias{obe:pow:vog:mul19}: the original algorithm used by the hospital (asterisk); one that predicts total cost (hollow diamond with a cross); one that predicts avoidable costs (filled diamond); and one that predicts the number of active chronic conditions (hollow diamond). $\Delta\btheta$ is estimated by logit lasso.}
    \label{fig:tests:lasso}
\end{figure}

\begin{figure}[H]\captionsetup[subfigure]{font=footnotesize}
    \centering
    \includegraphics[width=0.95\linewidth]{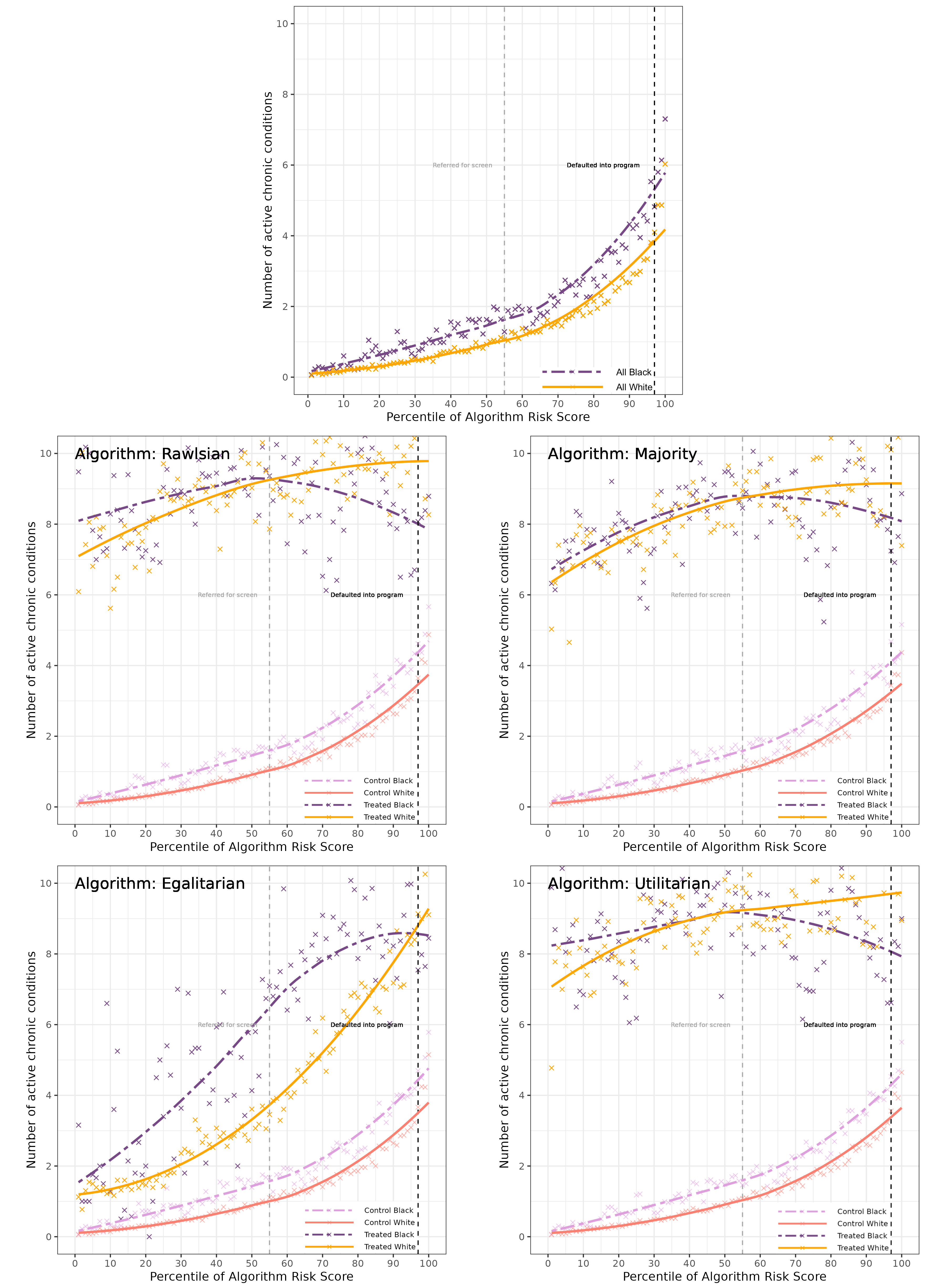}
    \caption{Average number of active chronic conditions within each risk-score percentile bin by treatment group under the alternative algorithms on the FA frontier subject to $3\%$ capacity constraint, averaged across $20$ replications of the $50$-$50$ split. $\Delta\btheta$ is estimated by logit lasso.}
    \label{fig:application:build:alg:lasso}
\end{figure}

\subsection{Variability of Empirical Results due to the Randomness in the Nuisance Estimation}
Both random forests and lasso involve randomness in their respective construction: for forests implemented by the \href{https://grf-labs.github.io/grf/index.html}{\texttt{grf}} package, randomness comes from subsampling in the construction of individual trees and randomly splitting features at each tree node, whereas for lasso implemented by the \href{https://glmnet.stanford.edu/index.html}{\texttt{glmnet}} package, randomness comes from choosing the optimal penalty parameter via cross-validation. Therefore, even if the same seed is set for reproducibility whenever possible, results will vary across different seeds. For this reason, we provide an assessment of how the empirical results reported in Section \ref{subsec:empirical} vary across different seeds by repeating the empirical exercises $20$ times, with results for forests and lasso reported respectively in Table \ref{tbl:variability-rf} and Table \ref{tbl:variability-lasso}.

\newpage
{
\begin{table}[H]
\scriptsize
\vspace{-.5cm}
\caption{\vspace{-.25cm}Results for the LDA Test and Confidence Sets for the Distance to $\F$ (for $\alpha=0.05$)\label{tbl:LDA-results-lasso}}
\vspace{.1cm}
\centering
\resizebox{0.85\columnwidth}{!}{
\begin{tabular}
{@{\extracolsep{3pt}}c@{}c@{}@{}c@{}@{}c@{}@{}c@{}@{}c@{}@{}} 

\multicolumn{5}{c}{\cellcolor{pink!20}\textbf{Test \textit{$\HH_0:$ There Is No LDA}}}\\
\cline{1-5}
\multicolumn{1}{c}{} & 
 \multicolumn{1}{c}{Original} & \multicolumn{1}{c}{Total Costs} & \multicolumn{1}{c}{Avoid. Costs} & \multicolumn{1}{c}{Act. Chr. Cond.} \\
\cline{1-5}
Estimated Risks & $ (0.065, 0.034)$ & $(0.059, 0.033)$ & $(0.054, 0.023)$ & $(0.037, 0.015)$ \\
Test Statistic & $2.043$ & $0.787$ & $0.440$ & $2.774$\\
Critical Value & $2.014$ & $1.865$ & $1.937$ & $1.886$ \\
Conclusion & Rejected & Not Rejected & Not Rejected & Rejected \\ 
\cline{1-5}
\multicolumn{5}{c}{\cellcolor{pink!20}\textbf{Distance to $\F=(0.063, 0.063)$}}\\
\cline{1-5}
Estimated Distance & $0.0009$ & $0.0009$ & $0.0017$ & $0.0029$ \\
Confidence Set & $(0.000, 0.002)$ & $(0.000, 0.002)$ & $(0.000, 0.003)$ & $(0.001, 0.006)$\\
\cline{1-5}
\end{tabular}
}
\begin{tablenotes}
\footnotesize
\item Top panel: LDA test statistics and $0.05$-level critical values associated with the original algorithm and the three experimental algorithms (predicting, respectively, total costs; avoidable costs; number of active chronic conditions) analyzed by \citetalias{obe:pow:vog:mul19}. Bottom panel: estimated squared-Euclidean distance to the $\F$ point and corresponding confidence set for this distance. $\Delta\btheta$ is estimated by logit lasso
\end{tablenotes}

\end{table}
}
{
\begin{table}[H]
\scriptsize
\vspace{-.5cm}
\caption{\vspace{-.25cm}Fraction of Black Patients Treated among All Treated \label{tbl:frac-bl-trt-lasso}}
\vspace{.1cm}
\centering
\resizebox{0.95\columnwidth}{!}{
\begin{tabular}{c|c@{\hspace{-.5cm}}c@{\hspace{-.5cm}}|cccc}
\cline{1-7}
\multicolumn{1}{c}{\cellcolor{pink!20}} & \multicolumn{2}{c}{\cellcolor{pink!20}\hspace{-.1cm}\textbf{Algorithms from \citeauthor{obe:pow:vog:mul19}}} & \multicolumn{4}{c}{\cellcolor{pink!20}\textbf{Algorithms on the FA-Frontier}}\\ 
\cline{1-7}
Capacity Threshold & $\hspace{.5cm}$Original & Counterfactual & Rawlsian & Majority & Egalitarian & Utilitarian \\
\cline{1-7}
$55$ & $\hspace{.5cm}0.120$ & $0.184$ & $0.173$ & $0.173$ & $0.174$ & $0.172$\\
$69$ & $\hspace{.5cm}0.128$ & $0.255$ & $0.217$ &	$0.200$ & $0.171$ & $0.202$\\
$82$ & $\hspace{.5cm}0.138$ & $0.327$ & $0.241$ & $0.224$ & $0.130$ & $0.223$\\
$89$ & $\hspace{.5cm}0.151$ & $0.407$ & $0.264$ & $0.247$ & $0.118$ & $0.249$\\ 
$94$ & $\hspace{.5cm}0.167$ & $0.498$ & $0.324$ & $0.284$ & $0.124$ & $0.294$\\
$97$ & $\hspace{.5cm}0.184$ & $0.592$ & $0.369$ & $0.318$ & $0.143$ & $0.339$\\
\cline{1-7}
\end{tabular}
}
\begin{tablenotes}
\footnotesize
\item The distribution of the number of active chronic conditions is such that the $55^\text{th}$ to the $68^\text{th}$ percentiles all correspond to $1$ active chronic condition, the $69^\text{th}$-$81^\text{st}$ correspond to $2$, the $82^\text{nd}$-$88^\text{th}$ correspond to $3$, the $89^\text{th}$-$92^\text{nd}$ correspond to $4$, the $94^\text{th}$-$95^\text{th}$ correspond to $5$, and the $96^\text{th}$-$97^\text{th}$ correspond to $6$. $\Delta\btheta$ is estimated by logit lasso.
\end{tablenotes}
\end{table}
}

{
\begin{table}[H]
\scriptsize
\vspace{-.5cm}
\caption{\vspace{-.25cm}Variability of Empirical Results: Random Forests\label{tbl:variability-rf}}
\vspace{.1cm}
\centering
\resizebox{0.8\columnwidth}{!}{
\begin{tabular}
{cccccccc} 
\cline{1-8}
{\cellcolor{pink!20}} & 
{\cellcolor{pink!20}Mean} & {\cellcolor{pink!20}SD} & {\cellcolor{pink!20}Min} & {\cellcolor{pink!20}$25$-th} &{\cellcolor{pink!20}$50$-th} & {\cellcolor{pink!20}$75$-th} & {\cellcolor{pink!20}Max} \\
\cline{1-8}
\multicolumn{8}{c}{Weak Skew Test}\\
Conclusion & $0$ & $0$ & $0$ & $0$ & $0$ & $0$ & $0$ \\
\cline{1-8}
\multicolumn{8}{c}{LDA Test} \\
\multicolumn{8}{l}{\textit{Original Algorithm:}}\\
Test Statistic & $3.338$ &	$0.318$ & $2.666$ & $3.173$ & $3.301$ & $3.562$ & $3.823$\\
Critical Value & $1.894$ & $0.051$ & $1.810$ & $1.866$ & $1.896$ & $1.919$ & $2.004$ \\
Conclusion & $1$ & $0$ & $1$ & $1$ & $1$ & $1$ & $1$\\
\multicolumn{8}{l}{\textit{Algorithm that Predicts Total Costs:}}\\ 
Test Statistic & $2.149$ & $0.327$ & $1.475$ & $1.977$ & $2.111$ & $2.361$ & $2.761$\\
Critical Value & $1.844$ & $0.066$ & $1.736$ & $1.804$ & $1.829$ & $1.876$ & $1.998$ \\
Conclusion & $0.8$ & $0.410$ & $0$ & $1$ & $1$ & $1$ & $1$\\
\multicolumn{8}{l}{\textit{Algorithm that Predicts Avoidable Costs:}} \\
Test Statistic & $1.488$ & $0.054$ & $1.398$ & $1.440$ & $1.493$ & $1.521$ & $1.577$\\
Critical Value & $1.721$ & $0.060$ & $1.612$ & $1.682$ & $1.724$ & $1.754$ & $1.836$ \\
Conclusion & $0$ & $0$ & $0$ & $0$ & $0$ & $0$ & $0$\\
\multicolumn{8}{l}{\textit{Algorithm that Predicts the Number of Active Chronic Conditions:}} \\
Test Statistic & $1.194$ & $0.346$ & $0.674$ & $1.007$ & $1.135$ & $1.395$ & $1.812$
\\
Critical Value & $1.632$ & $0.049$ & $1.516$ & $1.600$ & $1.634$ & $1.654$ & $1.725$ \\
Conclusion & $0.15$ & $0.366$ & $0$ & $0$ & $0$ & $0$ & $1$\\
\cline{1-8} 
\multicolumn{8}{c}{Confidence Set for the Distance to $F$} \\
Estimated $\F$ & $0.052$ & $0.003$ & $0.049$ & $0.050$ & $0.052$ & $0.054$ & $0.058$\\
\multicolumn{8}{l}{\textit{Original Algorithm:}}\\
Estimated Distance & $0.0005$ & $0.0000$ & $0.0005$ & $0.0005$ & $0.0005$ & $0.0005$ & $0.0006$\\
Lower $95\%$ CI & $0.0001$ & $0.0001$ & $0.0000$ & $0.0000$ & $0.0000$ & $0.0001$ & $0.0002$\\
Upper $95\%$ CI & $0.0012$ & $0.0003$ & $0.0008$ & $0.0009$ & $0.0011$ & $0.0014$ & $0.0018$\\
\multicolumn{8}{l}{\textit{Algorithm that Predicts Total Costs:}}\\ 
Estimated Distance & $0.0004$ & $0.0001$ & $0.0003$ & $0.0004$ & $0.0004$ & $0.0004$ & $0.0006$ \\
Lower $95\%$ CI & $0.0000$ & $0.0000$ & $0.0000$ & $0.0000$ & $0.0000$ & $0.0000$ & $0.0001$ \\
Upper $95\%$ CI & $0.0013$ & $0.0004$ & $0.0007$ & $0.0010$ & $0.0013$ & $0.0016$ & $0.0022$ \\
\multicolumn{8}{l}{\textit{Algorithm that Predicts Avoidable Costs:}} \\
Estimated Distance & $0.0009$ & $0.0002$ & $0.0007$ & $0.0007$ & $0.0008$ & $0.0009$ & $0.0013$ \\
Lower $95\%$ CI & $0.0000$ & $0.0000$ & $0.0000$ & $0.0000$ & $0.0000$ & $0.0000$ & $0.0002$ \\
Upper $95\%$ CI & $0.0024$ & $0.0005$ & $0.0018$ & $0.0020$ & $0.0023$ & $0.0027$ & $0.0033$ \\
\multicolumn{8}{l}{\textit{Algorithm that Predicts the Number of Active Chronic Conditions:}} \\
Estimated Distance & $0.0016$ & $0.0003$ & $0.0012$ & $0.0013$ & $0.0016$ & $0.0017$ & $0.0023$ \\
Lower $95\%$ CI & $0.0003$ & $0.0002$ & $0.0000$ & $0.0002$ & $0.0003$ & $0.0004$ & $0.0006$ \\
Upper $95\%$ CI & $0.0041$ & $0.0006$ & $0.0030$ & $0.0035$ & $0.0042$ & $0.0046$ & $0.0051$\\
\cline{1-8}
\end{tabular}
}
\begin{tablenotes}
\footnotesize
\item Table \ref{tbl:variability-rf} reports the variability of the empirical results in Section \ref{subsec:empirical} due to randomness in the estimation of $\Delta\btheta$ using random forests, reported as the mean, standard deviation, minimum, the 25-th percentile, 50-th percentile, 75-th percentile, and the maximum across 20 replications. Test conclusions are recorded as $1$ if the conclusion is rejection, and $0$ otherwise.
\end{tablenotes}

\end{table}
}

{
\begin{table}[H]
\scriptsize
\vspace{-.5cm}
\caption{\vspace{-.25cm}Variability of Empirical Results: Logit Lasso\label{tbl:variability-lasso}}
\vspace{.1cm}
\centering
\resizebox{0.8\columnwidth}{!}{
\begin{tabular}
{cccccccc} 
\cline{1-8}
{\cellcolor{pink!20}} & 
{\cellcolor{pink!20}Mean} & {\cellcolor{pink!20}SD} & {\cellcolor{pink!20}Min} & {\cellcolor{pink!20}$25$-th} &{\cellcolor{pink!20}$50$-th} & {\cellcolor{pink!20}$75$-th} & {\cellcolor{pink!20}Max} \\
\cline{1-8}
\multicolumn{8}{c}{Weak Skew Test}\\
Conclusion & $0$ & $0$ & $0$ & $0$ & $0$ & $0$ & $0$ \\
\cline{1-8}
\multicolumn{8}{c}{LDA Test} \\
\multicolumn{8}{l}{\textit{Original Algorithm:}}\\
Test Statistic & $2.037$ & $0.182$ & $1.784$ & $1.874$ & $2.001$ & $2.166$ & $2.390$ \\
Critical Value & $1.976$ & $0.074$ & $1.762$ & $1.937$ & $1.989$ & $2.024$ & $2.066$ \\
Conclusion & $0.550$ & $0.510$ & $0$ & $0$ & $1$ & $1$ & $1$ \\
\multicolumn{8}{l}{\textit{Algorithm that Predicts Total Costs:}}\\ 
Test Statistic & $0.804$ & $0.173$ & $0.562$ & $0.651$ & $0.788$ & $0.908$ & $1.154$ \\
Critical Value & $1.915$ & $0.055$ & $1.801$ & $1.890$ & $1.914$ & $1.961$ & $1.995$ \\
Conclusion & $0$ & $0$ & $0$ & $0$  & $0$  & $0$  & $0$  \\
\multicolumn{8}{l}{\textit{Algorithm that Predicts Avoidable Costs:}} \\
Test Statistic & $0.478$ & $0.171$ & $0.100$ & $0.375$ & $0.491$ & $0.591$ & $0.821$ \\
Critical Value & $1.913$ & $0.057$ & $1.783$ & $1.888$ & $1.923$ & $1.941$ & $2.006$ \\
Conclusion & $0$ & $0$ & $0$ & $0$ & $0$ & $0$ & $0$\\
\multicolumn{8}{l}{\textit{Algorithm that Predicts the Number of Active Chronic Conditions:}} \\
Test Statistic & $2.769$ & $0.174$ & $2.470$ & $2.668$ & $2.754$ & $2.889$ & $3.184$ \\
Critical Value & $1.873$ & $0.058$ & $1.768$ & $1.841$ & $1.880$ & $1.905$ & $2.020$ \\
Conclusion & $1$ & $0$ & $1$ & $1$ & $1$ & $1$ & $1$\\
\cline{1-8} 
\multicolumn{8}{c}{Confidence Set for the Distance to $F$} \\
Estimated $\F$ & $0.063$ & $0.001$ & $0.061$ & $0.062$ & $0.062$ & $0.063$ & $0.065$ \\
\multicolumn{8}{l}{\textit{Original Algorithm:}}\\
Estimated Distance & $0.0008$ & $0.0001$ & $0.0007$ & $0.0008$ & $0.0008$ & $0.0009$ & $0.0010$ \\
Lower $95\%$ CI & $0.0000$ & $0.0000$ & $0.0000$ & $0.0000$ & $0.0000$ & $0.0000$ & $0.0001$\\
Upper $95\%$ CI & $0.0020$ & $0.0003$ & $0.0014$ & $0.0017$ & $0.0019$ & $0.0021$ & $0.0026$ \\
\multicolumn{8}{l}{\textit{Algorithm that Predicts Total Costs:}}\\ 
Estimated Distance & $0.0009$ & $0.0001$ & $0.0008$ & $0.0008$ & $0.0009$ & $0.0009$ & $0.0010$ \\
Lower $95\%$ CI & $0.0000$ & $0.0000$ & $0.0000$ & $0.0000$ & $0.0000$ & $0.0000$ & $0.0001$ \\
Upper $95\%$ CI & $0.0022$ & $0.0003$ & $0.0017$ & $0.0021$ & $0.0022$ & $0.0025$ & $0.0025$  \\
\multicolumn{8}{l}{\textit{Algorithm that Predicts Avoidable Costs:}} \\
Estimated Distance & $0.0016$ & $0.0001$ & $0.0014$ & $0.0015$ & $0.0016$ & $0.0017$ & $0.0018$ \\
Lower $95\%$ CI & $0.0003$ & $0.0001$ & $0.0000$ & $0.0002$ & $0.0003$ & $0.0004$ & $0.0005$ \\
Upper $95\%$ CI & $0.0035$ & $0.0004$ & $0.0030$ & $0.0032$ & $0.0035$ & $0.0038$ & $0.0042$  \\
\multicolumn{8}{l}{\textit{Algorithm that Predicts the Number of Active Chronic Conditions:}} \\
Estimated Distance &  $0.0028$ & $0.0002$ & $0.0026$ & $0.0027$ & $0.0028$ & $0.0029$ & $0.0032$  \\
Lower $95\%$ CI &  $0.0009$ & $0.0003$ & $0.0005$ & $0.0008$ & $0.0009$ & $0.0011$ & $0.0016$\\
Upper $95\%$ CI & $0.0054$ & $0.0004$ & $0.0047$ & $0.0051$ & $0.0054$ & $0.0057$ & $0.0062$ \\
\cline{1-8}
\end{tabular}
}
\begin{tablenotes}
\footnotesize
\item Table \ref{tbl:variability-lasso} reports the variability of the empirical results in Section \ref{subsec:empirical} due to randomness in the estimation of $\Delta\btheta$ using logit lasso, reported as the mean, standard deviation, minimum, the 25-th percentile, 50-th percentile, 75-th percentile, and the maximum across 20 replications. Test conclusions are recorded as $1$ if the conclusion is rejection, and $0$ otherwise.
\end{tablenotes}

\end{table}
}

\end{appendix}
\end{document}